\documentclass[a4paper,onecolumn,11pt,accepted=2026-04-30]{quantumarticle}
\pdfoutput=1
\usepackage[utf8]{inputenc}
\usepackage[english]{babel}
\usepackage[T1]{fontenc}
\usepackage{url}            
\usepackage{nicefrac}       

\usepackage[margin = 1 in]{geometry}
\usepackage{amsmath, amssymb, amsthm, thm-restate}
\usepackage{braket, physics, enumitem, relsize}
\usepackage{makecell, graphicx, subfiles, multirow, wrapfig}
\usepackage{tikz}
\usepackage[font=scriptsize]{subcaption}
\usepackage[font=scriptsize, justification=raggedright]{caption}
\usepackage[colorlinks, allcolors=black, pagebackref=true]{hyperref}
\usepackage[capitalize, nameinlink, noabbrev]{cleveref}
\usepackage{tabularx, booktabs, array}
\setlist[itemize]{leftmargin=10pt}
\usepackage{float}
\usepackage{neuralnetwork}
\usepackage{adjustbox}
\usepackage{comment}
\usepackage{tikz}
\usepackage{pgfplots}  
\usetikzlibrary{shapes.geometric, positioning, trees}

\usetikzlibrary{quantikz2}
\usepackage{xcolor}
\hypersetup{
  colorlinks   = true, 
  urlcolor     = blue, 
  linkcolor    = teal, 
  citecolor   = blue 
}
\usepackage{setspace}
\usepackage{algcompatible, newfloat}

\DeclareFloatingEnvironment[
    fileext=loa,
    listname=List of Algorithms,
    name=Algorithm,
    placement=t,
]{alg}
\usepackage{algpseudocode}
\usepackage{algorithm}
\usepackage{bm}
\usepackage[numbers]{natbib}
\renewcommand*{\backref}[1]{}
\renewcommand*{\backrefalt}[4]{%
    \ifcase #1 No citations.%
    \or Cited on page~#2.%
    \else Cited on pages~#2.%
    \fi%
}

\mathcode`l="8000
\begingroup
\makeatletter
\lccode`\~=`\l
\DeclareMathSymbol{\lsb@l}{\mathalpha}{letters}{`l}
\lowercase{\gdef~{\ifnum\the\mathgroup=\m@ne \ell \else \lsb@l \fi}}%
\endgroup

\newtheorem{lemma}{Lemma}
\newtheorem{definition}{Definition}
\newtheorem{remark}{Remark}
\newtheorem{corollary}{Corollary}
\newtheorem{proposition}{Proposition}
\newtheorem{theorem}{Theorem}

\newtheorem*{theorem*}{Theorem}
\newtheorem{problem}{Problem}

\newtheorem{conjecture}{Conjecture}

\numberwithin{lemma}{section}
\numberwithin{definition}{section}
\numberwithin{remark}{section}
\numberwithin{corollary}{section}
\numberwithin{proposition}{section}
\numberwithin{theorem}{section}
\numberwithin{example}{section}

\Crefname{proposition}{Proposition}{Propositions}
\Crefname{alg-line}{Line}{Lines}
\crefname{alg}{Algorithm}{Algorithms}

\newcommand{\conj}[1]{\overline{#1}}

\newcommand{\wt}[1]{\widetilde{#1}}

\newcommand{\Bf}[1]{\textbf{#1}}

\DeclareMathOperator{\EE}{\mathbb{E}}

\DeclareMathOperator{\poly}{poly}

\renewcommand{\vec}[1]{\overrightarrow{#1}}
\renewcommand{\epsilon}{\varepsilon}

\renewcommand{\arraystretch}{1.5}

\newcommand{\C}{\mathcal{C}}

\newcolumntype{x}[1]{>{\centering\hspace{0pt}\arraybackslash}m{#1}}

\title{Bosonic Quantum Computational Complexity}

\author{Ulysse Chabaud} 
\affiliation{DIENS, ENS, PSL University, CNRS, INRIA, Paris, France} \email{ulysse.chabaud@inria.fr}

\author{Michael Joseph}
\affiliation{Tufts CS, Medford MA, USA}
\email{michael.joseph@tufts.edu}

\author{Saeed Mehraban}
\affiliation{Tufts CS, Medford MA, USA}
\email{saeed.mehraban@tufts.edu}
\author{Arsalan Motamedi}
\affiliation{IQC, Waterloo ON, Canada} \email{arsalan.motamedi@uwaterloo.ca}

\begin{document}

\maketitle

\begin{abstract}

In recent years, quantum computing involving physical systems with continuous degrees of freedom, such as the bosonic quantum states of light, has attracted significant interest. However, a well-defined quantum complexity theory for these bosonic computations over infinite-dimensional Hilbert spaces is missing. In this work, we lay foundations for such a research program. We introduce natural complexity classes and problems based on bosonic generalizations of \textbf{BQP}, the local Hamiltonian problem, and \textbf{QMA}. We uncover several relationships and subtle differences between standard Boolean classical and discrete-variable quantum complexity classes, 
and identify outstanding open problems.
Our main contributions include the following:
\begin{enumerate}
    \item\textbf{Bosonic computations.} We show that the power of Gaussian computations up to logspace reductions is equivalent to bounded-error quantum logspace (\textbf{BQL}, characterized by the problem of inverting well-conditioned matrices). More generally, we define classes of continuous-variable quantum polynomial time computations with a bounded probability of error (\textbf{CVBQP}) based on gates generated by polynomial bosonic Hamiltonians and particle-number measurements. Due to the infinite-dimensional Hilbert space, it is not a priori clear whether a decidable upper bound can be obtained for these classes. We identify complete problems for these classes, and we demonstrate a \textbf{BQP} lower bound and an \textbf{EXPSPACE} upper bound by proving bounds on the average energy throughout the computation. We further show that the problem of computing expectation values of polynomial bosonic observables at the output of bosonic quantum circuits using Gaussian and cubic phase gates is in \textbf{PSPACE}. 
    \item\textbf{Bosonic ground energy problems.}  We prove that the problem of deciding whether the spectrum of a bosonic Hamiltonian is bounded from below is undecidable in general, \mbox{\textbf{co-NP}-hard} for degree $4$ Hamiltonians and decidable in polynomial time for quadratic Hamiltonians. Furthermore, we show that the problem of finding the minimum energy of a bosonic Hamiltonian critically depends on the non-Gaussian stellar rank of the family of energy-constrained states one optimizes over: for zero stellar rank, i.e., optimizing over Gaussian states, it is \textbf{NP}-complete; for polynomially-bounded stellar rank, it is in \textbf{QMA}; for unbounded stellar rank, it is \textbf{RE}-hard, i.e., undecidable. 
\end{enumerate}
    
\end{abstract}

{
\hypersetup{linkcolor=black}
\tableofcontents
}

\hypersetup{linkcolor=teal}

\section{Introduction}

Many quantum mechanical systems, such as spin systems, are effectively described using discrete variables and are captured using qubits or qudits. The standard model of quantum computation is formulated based on such discrete degrees of freedom \cite{nielsen2001quantum}. On the other hand, quantum variables such as position, momentum, or amplitudes of electromagnetic fields are continuous. These degrees of freedom are described mathematically using infinite-dimensional Hilbert spaces. Continuous-variable systems are also present in many fundamental problems in theoretical physics, such as the solution to energy levels of a molecular system, quantum field theories, or exotic quantum states of matter such as Bose--Einstein condensates. 
 
Continuous-variable quantum architectures have recently been used in practical quantum computing implementations, with leading quantum error-correction schemes being fundamentally infinite-dimensional \cite{gottesman2001encoding,grimsmo2021quantum}. Several important physical frameworks for quantum computing, such as quantum photonic processors, are based on continuous-variable degrees of freedom. These architectures were recently demonstrated at scales sufficient to solve computational sampling problems that are believed to exceed the power of classical computations \cite{zhong2020quantum, madsen2022quantum}. 

What are the fundamental computational limitations and power of quantum degrees of freedom? This question is the subject of quantum complexity theory \cite{watrous2008quantum}. The standard quantum computing model and quantum complexity theory are formulated using discrete-variable quantum degrees of freedom over finite-dimensional Hilbert spaces. Several complexity classes have been defined to capture the computational power associated with these quantum degrees of freedom. One can define various natural models of quantum computations, such as adiabatic computation, quantum circuits, topological quantum computation, or measurement-based quantum computing, and all of these variants have been proven to have power equal to the computational complexity class \textbf{BQP}, which captures the power of polynomial-size quantum circuits in deciding logical statements with a small probability of error. Another important complexity class that captures the complexity of estimating the energy levels of a quantum physical system is \textbf{QMA}, which can be viewed as a quantum generalization of \textbf{NP} and, in particular, \textbf{MA} as the randomized generalization of \textbf{NP}. Due to a seminal result of Kitaev \cite{kempe2006complexity} (and many follow-up works, e.g., \cite{kempe20033}), the problem of estimating the ground state energy of a physical system up to inverse-polynomial additive error (known as the local Hamiltonian problem) is complete for \textbf{QMA}. As a matter of fact, the local Hamiltonian problem can be viewed as the quantum generalization of the canonical constraint satisfaction problem over the Boolean hypercube, which is itself complete for \textbf{NP}. 
In recent years, the relationship between these complexity classes and standard complexity classes such as polynomial time \textbf{P}, polynomial space \textbf{PSPACE}, nondeterministic polynomial time \textbf{NP}, $\#$\textbf{P} (corresponding to counting the number of solutions to constraint satisfaction problems), \textbf{PP} (the class of problems that are solvable on a probabilistic Turing machine with probability of error $< 1/2$) have been extensively studied. All of these complexity classes are contained in the class of problems solvable in exponential space \textbf{EXPSPACE}, itself included in the set of recursively enumerable languages \textbf{RE}. The known relationships between these complexity classes are illustrated in \cref{fig:complexity-classes}.

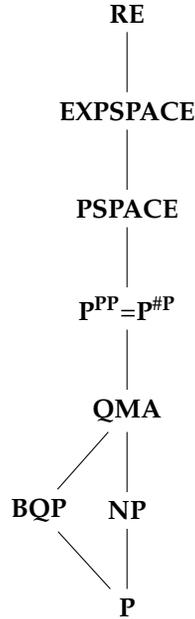
\begin{figure}
\centering

\begin{tikzpicture}[node distance=0.8cm and 0.3cm]

\tikzstyle{complexityclass} = [draw=none]

\node[complexityclass] (RE) {\textbf{RE}};
\node[complexityclass, below=of RE] (EXPSPACE) {\textbf{EXPSPACE}};
\node[complexityclass, below=of EXPSPACE] (PSPACE) {\textbf{PSPACE}};
\node[complexityclass, below=of PSPACE] (PP) {\textbf{P}$^{\textbf{PP}}$=\textbf{P}$^{\# \textbf{P}}$};
\node[complexityclass, below=of PP] (QMA) {\textbf{QMA}};
\node[complexityclass, below=of QMA] (NP) {\textbf{NP}};
\node[complexityclass, below=of NP] (P) {\textbf{P}};

\node[complexityclass, left=of NP] (BQP) {\textbf{BQP}};

\draw (RE) -- (EXPSPACE);
\draw (EXPSPACE) -- (PSPACE);
\draw (PSPACE) -- (PP);
\draw (PP) -- (QMA);
\draw (QMA) -- (NP);
\draw (NP) -- (P);

\draw (BQP) -- (QMA);
\draw (BQP) -- (P);

\end{tikzpicture}

\caption{Known relations between complexity classes used in this work. If a line connects complexity class $\mathbf C_1$ to $\mathbf C_2$, with $\mathbf C_1$ being on the right of $\mathbf C_2$, it is implied that $\Bf C_2 \subseteq \Bf C_1$.}
\label{fig:complexity-classes}
\end{figure}

What are the basic relationships between the power of computations using bosonic continuous-variables over infinite-dimensional Hilbert spaces and these standard discrete-variable complexity classes over finite-dimensional ones? As it turns out, one can robustly encode discrete-variable computations into infinite-dimensional Hilbert spaces. For instance, Knill, Laflamme and Milburn showed that one can simulate arbitrary discrete-variable quantum computations using linear optics, particle-number measurements and feed-forward \cite{knill2001scheme}, i.e. adapting the computation based on the result of intermediate measurements, while Gottesman, Kitaev and Preskill gave a protocol to encode a qubit into the continuous degrees of freedom of a quantum harmonic oscillator in a fault-tolerant way \cite{gottesman2001encoding}; the latter is among the leading proposals for reaching fault-tolerance in continuous-variable quantum architectures. As a consequence of these results, \textbf{BQP} can be robustly simulated using continuous quantum degrees of freedom. However, a formal complexity theory of bosonic computations is still missing.

In the continuous-variable setting, a model of quantum computation was first proposed by Lloyd and Braunstein \cite{lloyd1999quantum}. In this model, the state of a quantum system evolves according to Hamiltonians which are finite-degree polynomials in canonical continuous-variable (unbounded) operators such as the position operator $X$ and the momentum operator $P$. These operators have fundamentally different mathematical properties from their discrete-variable counterparts. For instance, for any pair of discrete-variable operators acting on a finite-dimensional Hilbert space, $A$ and $B$, $\mathrm{Tr}([A,B]) = 0$, it is well-known that $[X,P] = i I$, which is not trace-class (i.e., does not have finite trace). The operators $X$ and $P$ have unbounded spectra with respective formal eigenbases $\ket{x}, x \in \mathbb {R}$ such that $X \ket{x} = x \ket{x}$, and $\ket{p}, p \in \mathbb {R}$ such that $P \ket{p} = p \ket{p}$, that are related to each other by a Fourier transform $\ket{x} \propto \int e^{i p x} \ket{p} dp$, with their inner product giving Dirac delta functions $\braket{x}{x'} = \delta(x-x')$, to be interpreted in the sense of distributions. Another striking example of the peculiarities of infinite dimensions comes from quantum states such as $(\sqrt{6}/\pi)\sum_{n \geq 1} \frac{1}{n} \ket{n}$, which are normalized but have infinite energy with respect to the harmonic oscillator Hamiltonian $X^2 + P^2$. 
Yet another famous example, known as Tsirelson's problem, is that the tensor product and the commuting operator formulations of quantum mechanics are equivalent in finite dimensions but not in infinite dimensions, because of complexity-theoretic reasons in particular \cite{ji2021mip}.

Due to such nontrivial features, standard quantum complexity results in the discrete-variable setting, such as the Solovay--Kitaev theorem \cite{kitaev1997quantum}, do not have a clear counterpart in the continuous case. In particular, if one defines continuous-variable analogs of standard discrete-variable quantum computational complexity classes \textbf{BQP} and \textbf{QMA}, does it lead to relationships similar to the discrete-variable case?

Beyond its fundamental relevance, developing such a complexity-theoretic foundation for continuous quantum degrees of freedom would widen our horizon of knowledge in several ways. 
Firstly, it would help us to understand the computational power of continuous quantum degrees of freedom for real-world applications. Many physical architectures that are promising candidates for realising a full-fledged universal discrete-variable quantum computer, such as photonic systems, superconducting qubits, and acoustic modes like phonons, are in fact inherently continuous. Developing a fundamentally continuous language for studying these systems could be essential to fully utilize the potential of such quantum computational devices, e.g., for devising fast compilation strategies \cite{sefi2011decompose} or performing Hamiltonian simulation with optimal constants \cite{peng2023quantum}.

Secondly, it touches on fundamental questions in theoretical physics. For instance, a standard approach to understanding continuous systems is to discretize them to a certain precision. We face an immediate question: can we reduce the study of these systems by approximating them with discrete variables, or are continuous-variable systems fundamentally distinct and require their own separate formulation? Theoretical physics provides some evidence for the latter. For instance, the renormalization group provides the insight that the fundamental behavior of a physical system can vary significantly at different scales of precision. Or, in quantum field theory, one needs to impose energy cut-offs to avoid divergences in the computation of scattering amplitudes. Computational complexity studies the fundamental limits of finding reductions between two models, and we believe that tools from this framework would play a fundamental role in studying reductions between discrete and continuous degrees of freedom.

In this paper, we lay the groundwork for a computational complexity theory for quantum continuous variables. We formulate several Boolean quantum computational complexity classes based on continuous-variable bosonic generalizations of \textbf{BQP} and \textbf{QMA}. We study the relationship between these and their discrete-variable counterparts, highlight their fundamental similarities and differences, 
and identify outstanding open problems.

\subsection{Main models and results}\label{sec:models-and-main}

A bosonic quantum state can be described using a vector in the infinite-dimensional Hilbert space of square-summable complex sequences $\ell^2(\mathbb{C})$ which we refer to as a bosonic mode, or a qumode. Let $\mathcal{H}_n\cong \ell^2 (\mathbb{C}^{\otimes n})$ be the tensor-product space of $n$ bosonic modes. We consider bosonic computations based on Hamiltonians $H$ acting on $\mathcal H_n$ that are polynomials of degree $d$ in the canonical operators $X$ and $P$. Commonly known as position and momentum operators, they satisfy the canonical commutation relations
\begin{equation} 
[X_a, P_b] = i \delta_{a,b} I, \quad [X_a, X_b] = 0, \quad [P_a, P_b] = 0.
\end{equation}
Polynomial Hamiltonians are ubiquitous in quantum physics and allow us to describe most existing models of bosonic quantum computations. The spectrum of such Hamiltonians is a subset of the reals, which can be discrete (e.g., for $X^2+P^2$), continuous (e.g., for $X$) or both.
In our presentation, we always distinguish between the cases of Gaussian Hamiltonians (polynomials of degree $d\le2$), which are typically less powerful, and non-Gaussian Hamiltonians (polynomials of degree $d>2$).

We will often make use of the particle number operator (or simply, number operator) defined in terms of the position and momentum operators by,
\begin{equation}
    N:= \frac{X^2+P^2-I}{2} 
\end{equation}
where $N\ket{n} = n\ket{n}$ for the Fock basis states $\ket{n}$. The number operator plays a key role in our analysis as it quantifies the quanta of excitation, or 'particles,' in a bosonic mode. In terms of the energy of a bosonic state, we frequently refer to the average particle number, which corresponds to the expectation value $\langle N \rangle$.

\subsubsection{Structure and summary of results}

\begin{figure}[t]
    \centering
    \begin{tikzpicture}
\node at (0,0) (exp-space) {$\mathbf{EXPSPACE}$};
\node at (4,-2) (cv-qma) {$\mathbf{CVQMA}$};
\node at (4,-3) (cv-bqp) {$\mathbf{CVBQP}$};
\node at (2,-5) (cv-bqp3) {$\mathbf{CVBQP}[X^3]$};
\node at (4, -7.5) (bqp) {$\mathbf{BQP}$};
\node at (6, -5) (cvbqp-bar) {$\mathbf{CVBQP}[\overline{\sigma_x}, \overline{\sigma_z}\otimes\overline{\sigma_z}]$};
\node at (-1, -2) (pspace) {$\mathbf{PSPACE}$};
\node (box) [draw, minimum size=1cm, align=center] at (-2,-5) {Observable estimation\\ for $\mathbf{CVBQP}[X^3]$};

\draw (pspace) --node[above, sloped] {\cref{th:expvalinPSPACE}} (box);

\draw (exp-space) -- node [above, sloped]{\cref{th:CVBQPinEXPSPACE}} (cv-bqp3);

\draw (cv-qma) -- (cv-bqp);

\draw (cv-bqp) -- (cvbqp-bar);

\draw (cv-bqp) -- (cv-bqp3);


\draw (cvbqp-bar) -- node [sloped, above] {\!\!\!\!\cref{th:CVBQPcontainsBQP}} (bqp);

\draw (pspace) -- (exp-space);

\draw [line] (-4.5,0) -- (-4.5,-7);

\node (gdc) at (-6.5,-3) {$\mathbf{GDC} = \mathbf{BQL}$};

\node (citation) at (-6.5,-3.5)
{\cref{thm:GDC=BQL}};

\node at (-6.7, -6.5)
{\textit{Gaussian computations}};

\node at (-1.5, -6.5) {\textit{Non-Gaussian computations}};
\end{tikzpicture}
    \caption{Relations between complexity classes and computational problems proven in this paper and described below, up to logspace reductions on the left and polynomial-time reductions on the right. If a line connects two complexity classes, it is implied that the one below in the diagram is included in the one above. If a line connects a problem to a complexity class, it is implied that the problem is in that class if it is below, and hard for that class if it above. Square brackets indicate specific choices of non-Gaussian gates, where $X^3$ and $\bar X$ are polynomial Hamiltonians of degree $3$, while $\bar Z\otimes\bar Z$ is a polynomial Hamiltonian of degree $4$. The other inclusions are given by the corresponding theorems or trivial from the definitions of the classes.}
    \label{fig:CVBQP}
\end{figure}
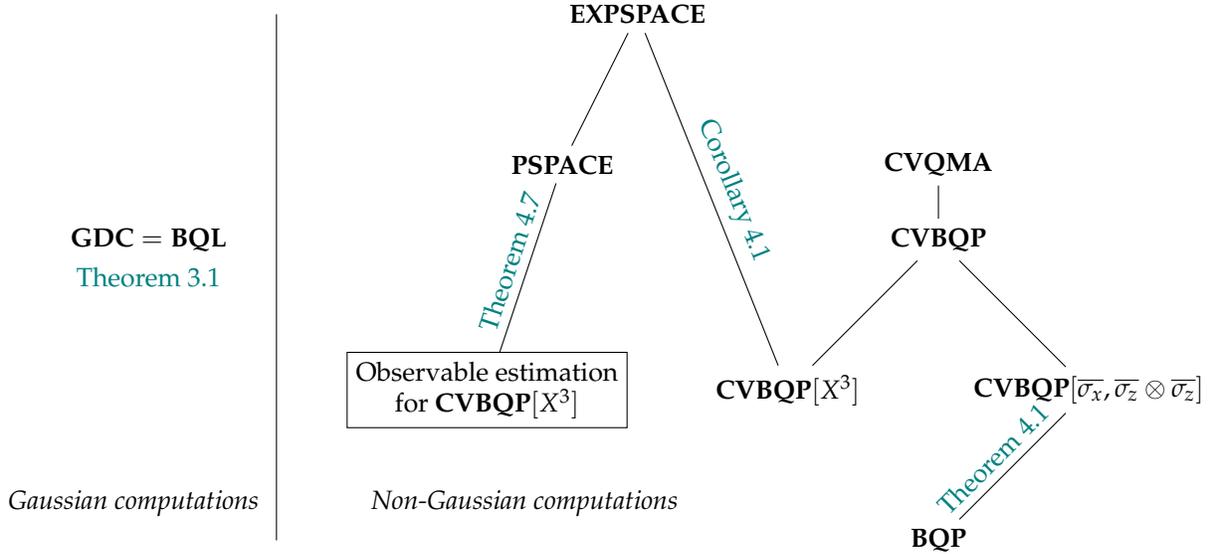

The structure of the paper is as follows.
After preliminary background in \cref{sec:preliminaries}, we first consider Boolean complexity classes based on bounded-error continuous-variable quantum computations in \cref{sec:GDC} (Gaussian) and \cref{sec:CVBQP} (non-Gaussian). These contributions are summarised in \cref{fig:CVBQP}:

\begin{itemize}
    \item We define a complexity class corresponding to the power of Gaussian dynamical computations (\textbf{GDC}) with logspace preprocessing. We prove that the power of this model is equivalent to the complexity class \textbf{BQL} corresponding to quantum logspace (See \cref{thm:GDC=BQL}). A complete problem for this class is inverting well-conditioned matrices \cite{fefferman2021eliminating, ta2013inverting}. This can be viewed as a continuous-variable version of a result by Aaronson and Gottesman \cite{aaronson2004improved}, which places the problem of simulating Clifford circuits in the complexity class $\oplus$\textbf{L} (a complete problem for this class is computing the determinant over the finite field $\mathbb{Z}_2$ \cite{damm1990problems}).
    \item We then define bounded-error quantum computations using gates generated by polynomials in the bosonic position and momentum operators (\textbf{CVBQP}). In this model, we sample the output of the quantum computation in the particle number basis (eigenstates of $X^2+P^2$). We show that a specific instance of this class contains \textbf{BQP} (\cref{th:CVBQPcontainsBQP}) and that continuous-variable quantum computations using only Gaussian and cubic phase gates (generated by $X^3$) can be strongly simulated in \textbf{EXPSPACE} when no energy upper bound is assumed (\cref{thm:schrod-in-expspace}). We explain why proving a stronger upper bound of \textbf{PSPACE} or even \textbf{PP} as in the discrete-variable case (if true) would require nontrivial ideas. Finally, we give a polynomial-time parallel algorithm for computing output continuous-variable observable expectation values of such computations (\cref{th:expvalinPSPACE}), which is in \textbf{PSPACE} by standard results \cite{trahan1992multiplication}.
    We explain in \cref{sec:schrod-vs-heis} that the $\mathbf{EXPSPACE}$ and $\mathbf{PSPACE}$ results for particle number measurements v.s.\ continuous-variable observable measurements are due to fundamental computational differences in Schr\"odinger v.s.\ Heisenberg dynamical evolution in continuous-variable systems. We conjecture that a strong simulation of the former model is strictly more difficult in terms of computational complexity.
\end{itemize}

\noindent Next, we consider ground state energy problems and Boolean non-deterministic quantum complexity classes in \cref{sec:gsp}. These contributions are summarised in \cref{tab:CVQMA}:

\begin{itemize}
\item We formulate the problem of deciding whether a polynomial Hamiltonian of degree $d$ is bounded from below ($\mathsf{HBound}^d$); we show that this problem is solvable in \textbf P for Gaussian Hamiltonians $d =2$ (\cref{th:gspGaussian}) and \textbf{co-NP}-hard for polynomial Hamiltonians of degree $4$ (\cref{th:boundedness}); we furthermore show that the problem is undecidable in general for Hamiltonians of degree $\geq 10$. Furthermore, we give an efficient classical algorithm for verifying the boundedness of the ground state energy of a subclass of Hamiltonians of degree $d = 4$ via a reduction to a classical sum-of-squares method (\cref{prop:soswitness}). 
\item We study the continuous-variable local Hamiltonian problem ($\mathsf{CVLH}^d_{\mathcal S}$) of estimating the lowest energy of a bosonic Hamiltonian of degree $d$ in the position and momentum operators over a set of states $\mathcal S$. For Gaussian Hamiltonians, we show that this problem is in \textbf P whenever the set $\mathcal S$ contains the set of Gaussian states (\cref{th:gspGaussian}). For non-Gaussian Hamiltonians, we prove that the complexity of this problem critically depends on the stellar rank $r$, a measure of the non-Gaussian character of a continuous-variable quantum state \cite{chabaud2020stellar,chabaud2022holomorphic}. We parametrize the relevant family of states one optimizes over using  simple constraints on the energy (average particle number):
    \begin{itemize}
    \item For $r=0$ (corresponding to optimization over Gaussian states) with at most $\mathsf{exp}(n) := e^{n^{O(1)}}$ energy (average particle number), we prove that the problem is \textbf{NP}-complete using a reduction from deciding when a matrix is not copositive (\cref{thm:gaussian-optimization}).
    
    \item When energy and stellar rank are both at most $\mathsf{poly}(n) := n^{O(1)}$, we prove that the problem is in \textbf{QMA} (\cref{thm:bounded-squeezing-and-displacement}). The same proof technique also shows that for arbitrary $r$ and at most $\mathsf{exp}(n)$ energy (average particle number), the problem is in \textbf{NTIME} ($n^{O(r)}$), where \textbf{NTIME} $(t)$ is the class of problems that are solvable by a nondeterministic Turing machine that runs in time $O(t)$ on each branch.

    \item For $r = \infty$ with no energy bound, using an observation of \cite{kieu2003quantum}, we encode the solution to Hilbert's tenth problem in the ground state of a Hamiltonian. As a consequence, the problem is \textbf{RE}-hard (undecidable) when we assume no bound on the stellar rank (\cref{thm:undecidability}).
    \end{itemize}
\item Finally, we introduce a continuous-variable version of \textbf{QMA}, based on a \textbf{CVBQP} verifier (\textbf{CVQMA}). We give preliminary ideas for relating the complexity of the continuous-variable local Hamiltonian problem to this class, based on a continuous-variable analog of Kitaev's history state construction \cite{kempe20033,kempe2006complexity}.
\end{itemize}

\begin{table}[t]
\centering
\setlength{\tabcolsep}{4pt}
\renewcommand{\arraystretch}{1.15}
\small
    \begin{tabular}{|c||c|c|}
    \hline
        Ground state problems & Gaussian Hamiltonians & Non-Gaussian Hamiltonians \\
    \hline
    $\mathsf{HBound}_d$ & $\in\mathbf P$ (\cref{th:gspGaussian}) & \textbf{co}-\textbf{RE}-hard for $d\geq 10$ (\cref{corol:boundedness}); \\
        &&\textbf{co-NP}-hard for $d = 4$ (\cref{th:boundedness}) \\
    \hline
        $\mathsf{CVLH}$ over ${\mathcal S^{\mathsf{exp}}_0}$ &  $\in\mathbf P$ (\cref{th:gspGaussian}) & \textbf{NP}-complete (\cref{thm:gaussian-optimization})\\
    \hline
        $\mathsf{CVLH}$ over ${\mathcal S^{\mathsf{poly}}_{\mathsf{poly}}}$ &  $\in\mathbf P$ (\cref{th:gspGaussian}) &  $\in\mathbf{QMA}$ (\cref{thm:bounded-squeezing-and-displacement}) \\
    \hline
        $\mathsf{CVLH}$ over ${\mathcal S^{\mathsf{exp}}_r}$ &  $\in\mathbf P$ (\cref{th:gspGaussian}) &  $\in\mathbf{NTIME}(n^{O(r)})$ (\cref{thm:bounded-squeezing-and-displacementNP}) \\
    \hline
        $\mathsf{CVLH}$ &  $\in\mathbf P$ (\cref{th:gspGaussian})&  $\mathbf{RE}$-hard (\cref{thm:undecidability})\\
    \hline
    \end{tabular}
    \caption{Computational complexity of the Hamiltonian boundedness problem and the continuous-variable local Hamiltonian problem for Gaussian and non-Gaussian polynomial bosonic Hamiltonians over $n$ modes. $\mathcal S^{\mathsf{exp}}_{O(1)}$ denotes the set of states with constant stellar rank with at most exponential energy (average particle number) and $\mathcal S^{\mathsf{poly}}_{\mathsf{poly}}$ the set of states with polynomially bounded stellar rank and energy. For all non-Gaussian results, assuming a degree $d=4$ is sufficient except for the undecidability result where we show that $d=8$ is sufficient (for $d=4$, the problem is $\mathbf{QMA}$-hard as a consequence of \cite{childs2014bose}).}
    \label{tab:CVQMA}
\end{table}

\subsubsection{Details of the contributions}

In what follows, we detail our contributions and provide some intuition. All definitions and results are stated informally and we refer to the subsequent sections for formal statements and proofs. 

\paragraph{Gaussian computations.} Let us first consider the power of bounded-error quantum computations using polynomial Hamiltonians of degree at most $2$. The unitary gates generated by such Hamiltonians are known as Gaussian gates. It is well-known that Gaussian gates are efficiently simulatable in polynomial time \textbf{P} when acting on Gaussian states \cite{bartlett2002efficient}, i.e., states that may be obtained from the vacuum using Gaussian gates \cite{ferraro2005gaussian}, together with Gaussian measurements, i.e., projection onto Gaussian states. We define a model of Gaussian computations as follows (see \cref{def:Gaussian-circuit} for a formal statement):

\begin{definition}[Gaussian dynamical computations, informal]
Gaussian dynamical computation ($\mathbf{GDC}$) is the class of problems that can be solved by evolving input Gaussian states via logspace uniform quadratic Hamiltonians for polynomial time, followed by measuring a single mode in the position basis (see \cref{sec:phase-space-formalism} for more details about the formalism). The computation accepts if the measured outcome has a value greater than a fixed constant $b$ and rejects if it is below a fixed constant $a$.
\end{definition}

$\mathbf{GDC}$ allows evolving a quantum state according to different Gaussian Hamiltonians one after the other, so long as the total time and number of Hamiltonians do not exceed a polynomial bound. Gaussian computations are known to be continuous-variable analogs of the so-called Clifford computations in the discrete-variable case. Clifford computations are also known to be classically simulatable in polynomial time, by the Gottesmann--Knill theorem \cite{gottesman1998heisenberg}. In \cite{aaronson2004improved}, Aaronson and Gottesman showed that one can actually simulate Clifford computations (i.e., sample from one qubit) in the complexity class $\oplus$\textbf{L}, which is believed to be strictly contained in \textbf{P}. Performing linear algebra (such as computing the determinant) over $\mathbb{Z}_2$ is a complete problem for this class. It is natural to ask whether a continuous-variable analog of this result holds. Our first result resolves this question in the affirmative (see \cref{thm:GDC=BQL} for a formal statement): 

\begin{theorem} [The computational power of Gaussian dynamics, informal]
    The power of Gaussian computations up to logspace reductions is captured by bounded-error quantum logspace (\textbf{BQL}) and, equivalently, the problem of inverting a well-conditioned matrix. 
\end{theorem}

\noindent The proof is based on the symplectic formulation of Gaussian operators \cite{ferraro2005gaussian}. 

Recently and independently, in \cite{barthe2024gate}, it was shown that simulating a particular class of gate-based Gaussian computations over exponentially many modes is $\mathbf{BQP}$-complete. Recall that approximate matrix inversion is a $\mathbf{BQP}$-complete problem when the matrix under consideration is sparse and well-conditioned \cite{dervovic2018quantum}. Although our models and results are technically different, both reveal a strong connection between Gaussian computations and linear algebra.

\paragraph{Bounded-error continuous-variable quantum polynomial time (CVBQP).} 
Arguably, there are many ways one can define a continuous-variable version of \textbf{BQP} based on how computation is being performed. We consider bosonic computations using gates generated by polynomials of constant degree in the position and momentum bosonic operators, which are ubiquitous in quantum physics, with particle-number measurements. We define gate set-dependent classes (see e.g., \cref{def:informalCVBQP} below) and investigate the computational power of these models for different choices of gate sets. In particular, we show that bosonic computations based on specific gates generated by degree-4 and $2$-local Hamiltonians can perform universal (discrete-variable) quantum computation on an \textit{input vacuum state and without requiring feed-forward of measurement outcomes} (see \cref{th:CVBQPcontainsBQP}). Interestingly, the equivalence between these different continuous-variable classes is unknown, in part due to the delicate features of unbounded operators such as those pointed out in the introduction. Some of these relationships, such as the fast compiling of polynomial-degree Hamiltonians into Gaussian and cubic phase gates, are outstanding open questions; see, e.g., \cite{sefi2011decompose,kalajdzievski2021exact}.

We then consider bosonic computations based on a family of circuits generated by cubic phase gates $e^{i X^3}$ and Gaussian gates $e^{i H}$, where $H$ is a quadratic Hamiltonian in $X$ and $P$. Notably, this gate set is believed/conjectured to be universal for the set of unitaries generated by arbitrary polynomials over an arbitrary number of modes \cite{sefi2011decompose, kalajdzievski2021exact}, a claim that was proven in a controllability sense in \cite{wu2006smooth}. Due to \cite{gottesman2001encoding}, this computational model, when equipped with the ability of performing Gaussian measurements and feed-forward, is capable of performing universal (discrete-variable) quantum computation, when complex input states known as Gottesman--Kitaev--Preskill states are available.

Gaussian gates are closed under multiplication, which is one way we can understand the classical simulation of Gaussians. The single-mode Gaussian dynamics can furthermore be understood via specific integrable classical equations of motion (known as Calogero--Moser dynamics) \cite{chabaud2022holomorphic}. However, once we add higher-degree gates to the gate set, the operators generated by the resulting gates generate a vastly larger set of operators, and the single-mode dynamics becomes chaotic \cite{lloyd1999quantum,wu2006smooth}. Define $\mathbf{CVBQP}[X^3]$ as the class of decision problems that are solvable using quantum circuits based on Gaussian and cubic phase gate sets (see \cref{def:CVBQPfixednG} for a formal definition):

\begin{definition}[$\mathbf{CVBQP}{[}X^3{]}$, informal]\label{def:informalCVBQP}
$\mathbf{CVBQP}[X^3]$ is the class of decision problems that are solvable with a bounded probability of error by applying a polynomial-time uniform sequence of Gaussian and cubic phase gates to the vacuum state and measuring the number of particles at the end, with the promise that the energy (average particle number) at the output is polynomially bounded.
\end{definition}

\noindent We show that, in between the computation, the energy of states prepared from the vacuum by polynomial-time sequences of Gaussian and cubic phase gates is upper bounded by a doubly-exponential function of the number of gates and modes (see \cref{prop:energyboundcubic}). The reason for such drastic energy growth in the system is that consecutive application of cubic and Gaussian gates can lead to repeated squaring of basic observables. This, in turn, leads to an upper bound of \textbf{EEXP} on the strong simulation (computing output amplitudes up to exponential precision) of $\mathbf{CVBQP}[X^3]$ when the final measurement is made in the computational basis. We then apply standard depth reduction techniques to bring the complexity down to $\mathbf{EXPSPACE}$ (see \cref{thm:schrod-in-expspace} for a formal statement):

\begin{theorem}[Upper bound on the computational power of Gaussian and cubic phase gates, informal]\label{th:informal_CVBQPinEEXP}
    Bosonic computations consisting of Gaussian and cubic phase gates on input vacuum and measurement in the Fock basis can be strongly simulated in $\mathbf{EXPSPACE}$. 
\end{theorem}

\noindent This theorem assumes no energy upper bound for each of the computational steps. When introducing an energy (average particle number) upper bound of $E^*$ in the above result, we note that the space complexity of the classical simulation in the single-mode case scales as $O(\textsf{log}(E^*))$.

Next, we focus on the problem of computing expectation values at the output of $\mathbf{CVBQP}[X^3]$ circuits for a low-degree observable $O$ (see \cref{def:ExpValfixednGprob} and \cref{th:expvalinPSPACE} for formal statements):

\begin{theorem}[The computational complexity of bosonic expectation values, informal]
The problem of computing expectation values of low-degree observables for states prepared by applying Gaussian and cubic phase gates to the vacuum can be solved in $\mathbf{PSPACE}$.
\end{theorem}

\begin{proof}[Proof sketch]
Let $W_i = U_i \cdots U_1$ and suppose that $W_i X_1 W_i^\dagger$ takes the form $O_i = \sum_{\mu, \nu} \alpha_{\mu, \nu} X^{\mu} P^\nu$. Here we used the multi-index notations $\mu = (\mu_1, \cdots, \mu_n) \in \mathbb{Z}_+^n$ and $X^{\mu} = X_1^{\mu_1} \cdots X^{\mu_n}_n$. Let $O_{i+1} = U_{i+1}O_i U_{i+1}^\dagger$ be the value after the application of the next gate. We can write $O_{i+1} = \sum_{\mu, \nu} \alpha_{\mu, \nu} (U_iXU_i^\dagger)^{\mu} (U_i P U^\dagger_i)^\nu$. If $U_{i+1}$ is a Gaussian gate, then the degree of $X$ and $P$ in $O_i$ will be the same as that of $O_{i+1}$. However, if $U_{i+1}$ is a cubic phase gate, then the degree of $O_{i+1}$ is at most twice the degree of $O_i$. That is because of the unitary evolution due to the cubic-phase gate $U_i P U_i^\dagger = P + c \cdot X^2$, where $c$ is a constant. As a result, the degree of $O_T$ is at most $2^T$. We note that coefficients involving the expansion of $O_T$ into a normal form may be doubly exponentially large. Hence, the naive brute-force approach runs in exponential space. To bring the complexity upper bound down to \textbf{PSPACE}, we give a polynomial-time \textit{parallel algorithm} using exponentially many processors. Standard results in computational complexity imply a \textbf{PSPACE} upper bound (see, for instance, \cite{trahan1992multiplication}). We further show that the upper bound remains true for the multimode case, using the fact that the only multimode operators we need to add are the so-called two-mode $SUM$ gates. 
\end{proof}

We emphasize that the standard relationships outlined in \cref{fig:complexity-classes} indicate a \textbf{PP} upper bound on \textbf{BQP}. Proving a $\mathbf{EXPSPACE}$ upper bound on $\mathbf{CVBQP}[X^3]$ (outlined in \cref{th:informal_CVBQPinEEXP}) already utilizes nontrivial tools, and we do not know if a stronger upper bound such as \textbf{PP} would hold in this case based on current techniques. Indeed, from the previous proof sketch, a remarkable feature of the cubic phase gate becomes apparent: starting with a single position operator $X$ and applying $T$ cubic phase gates interleaved with suitable Gaussian gates (the Fourier gate $e^{i\frac\pi4(X^2+P^2)}$, mapping $X$ to $P$ and $P$ to $-X$), one can perform repeated squaring, i.e., obtain an observable of the form $\propto X^{2^T}+\cdots$ after $T$ rounds \footnote{We thank Francesco Arzani for pointing out this fact.}. This implies that doubly exponentially large numbers may naturally arise after polynomially many gates in the continuous-variable setting, in stark contrast with the discrete-variable setting.

\paragraph{The boundedness problem.} Next, we aim to study the complexity of estimating the ground state energy of a bosonic Hamiltonian. Due to the infinite-dimensional setting, however, the spectrum may be unbounded. Hence, we first formulate and study the problem of deciding whether a bosonic Hamiltonian has a bounded ground state energy (see \cref{def:boundednessprob} for a formal definition). Note that this problem is equivalent to deciding whether a bosonic Hamiltonian has a bounded spectrum, by checking boundedness of the ground energy for $H$ and $-H$. The spectrum of any Hamiltonian that is a polynomial of odd degree (such as $X$ or $X^3$ or $X^2 P + P X^2$, etc.) is not bounded, as can be seen by computing the expectation value for an arbitrary coherent state (eigenstates of the operator $X+iP$), so we focus on polynomial Hamiltonians of even degree. 

It turns out that the quadratic (Gaussian) case is solvable in polynomial time via a reduction to the problem of deciding whether a polynomial-size matrix, which may be computed efficiently from the coefficients of the Hamiltonian, is positive semi-definite (see \cref{th:gspGaussian}). 

In the non-Gaussian case, we prove that the problem is significantly harder, even for degree-4 Hamiltonians (see \cref{th:boundedness} for a formal statement):

\begin{theorem}[Complexity of the boundedness problem, informal]\label{th:informal_boundedness}
The problem of deciding whether the spectrum of a bosonic Hamiltonian with degree $4$ is bounded is $\mathbf{co}$-$\mathbf{NP}$-hard. Furthermore, for degree $10$ or higher, the problem becomes undecidable.
\end{theorem}

\begin{proof}[Proof sketch] The $\mathbf{NP}$-hardness proof proceeds via reduction from matrix copositivity, which is the problem of deciding, given $M \in \mathbb{R}^{n \times n}$, whether $\bra{x}M\ket{x}\geq 0$ for all $\ket{x} \in \mathbb{R}_{\geq 0}^n$ with non-negative entries. This problem is known to be \textbf{co}-\textbf{NP}-complete \cite{murty1987some}.

The undecidability for degree $10$ or higher, is a corollary of ground-energy undecidability (discussed later in \cref{thm:informal-gs-re-hard}). We first show that deciding whether the ground energy of a Hamiltonian $H$ is less than $a$ or larger than $b$ (for $b-a\geq 1$) is undecidable. This holds if $H$ is of degree $8$ or higher. We then construct another Hamiltonain $H'=N\otimes (H-c)$ with $c = \frac{b-a}{2}$. We note that if $\lambda_{\min}(H) \geq b$, then $H'\geq 0$ is bounded from below. However, if $\lambda_{\min}(H)\leq a$ then $H'$ is not bounded from below. As a result, we have a reduction from the ground energy (an $\mathbf{RE}$-hard problem) to the boundedness problem.
\end{proof}

How hard is it to find a witness for the boundedness of the ground state energy? From the above, we cannot find it in polynomial time in general. Even for degree $4$ Hamiltonians, finding the witness in polynomial time would at least collapse \textbf{co}-\textbf{NP} to \textbf{P}, which is highly unlikely. However, we may find a procedure to achieve this goal at least in some instances. This is what we do next, based on a sum-of-squares technique. 

Consider a polynomial Hamiltonian in the form $H = \sum_{\mu, \nu \in \mathbb{Z}^n_+} \frac12 h_{\boldsymbol\mu, \boldsymbol\nu} \{X^{\boldsymbol\mu}, P^{\boldsymbol\nu}\}$, where we used multi-index notation (in \cref{prop:SOS} we show that any polynomial Hamiltonian can be brought to this form with real coefficients $h_{\boldsymbol{\mu},\boldsymbol{\nu}}\in\mathbb R$). The classical polynomial $p_H : \mathbb{R}^{2n} \rightarrow \mathbb{R}$ corresponding to $H$ is defined as 
$$
p_H (x_1, \cdots, x_n, p_1, \cdots, p_n) = \sum_{\boldsymbol\mu, \boldsymbol\nu \in \mathbb{Z}^n_+} h_{{\boldsymbol\mu}, {\boldsymbol\nu}} x^{\boldsymbol\mu} p^{\boldsymbol\nu}.
$$
We prove the following (see \cref{prop:soswitness} for a formal statement):

\begin{proposition}[Checking boundedness for degree-4 Hamiltonians, informal]
    Let $H$ be a bosonic Hamiltonian of degree $4$. If $p_H$ is a sum-of-squares polynomial, then the spectrum of $H$ is bounded from below by an efficiently computable constant. 
\end{proposition}

As a result, a sum-of-squares approach provides a sound algorithm for deciding boundedness of the ground state energy of polynomial Hamiltonians in the $d = 4$ case, meaning that if we find a valid sum-of-squares decomposition for $p_H$, then $H$ is bounded, but if $p_H$ is not a sum-of-squares that does not imply unboundedness for $H$. Since the degree of $p_H$ is constant, one can look for a sum-of-squares decomposition by running a polynomial-time semi-definite program \cite{parrilo2003semidefinite}. Note that if $p_H$ is not a sum of squares and we conjugate $H$ by a Gaussian unitary $U$ (which does not change the degree of $H$) and try again, we may find a different valid witness for boundedness.

\paragraph{The continuous-variable local Hamiltonian problem.} We consider the case of general polynomial Hamiltonians and define the continuous-variable local Hamiltonian problem as follows (see \cref{def:cvlh} for a formal definition):

\begin{definition}[The continuous-variable local Hamiltonian problem, informal]
    Let $\mathcal S\subset\mathcal H$ be a subset of continuous-variable quantum states. The continuous-variable local Hamiltonian problem $\mathsf{CVLH}^d_{\mathcal S}$ is the problem of estimating the lowest energy of a poynomial Hamiltonian $H$ of degree $d$ over the set $\mathcal S$.
\end{definition}

\noindent Note that the name \textit{local} for this problem comes here from the fact that any polynomial Hamiltonian of degree $d$ is at most $d$-local by definition. 

For Gaussian Hamiltonians, our solution to the boundedness problem also provides a polynomial-time algorithm to estimate the ground state energy, thus placing the continuous-variable local Hamiltonian problem $\mathsf{CVLH}^2_{\mathcal H}$ for Gaussian Hamiltonians in $\mathbf P$ (see \cref{th:gspGaussian}).

For non-Gaussian Hamiltonians, in the case of states with bounded particle number, a result of \cite{childs2014bose} proves that the problem of estimating the ground state energy of the Bose--Hubbard model at finite (polynomial) number of bosons is \textbf{QMA}-complete. This Hamiltonian is of the form $H_{\mathrm{bh}} = t \cdot \sum_{i,j} A_{i,j} \, a_i^\dagger a_j + J \cdot \sum_{i} N_i (N_i-1)$, where $N=(X^2+P^2-1)/2$ and $a=(X+iP)/\sqrt2$, and where $A_{i,j}$ is the $(i,j)$-th entry of the adjacency matrix of an undirected graph. Note that the Hamiltonian $H_{\mathrm{bh}}$ conserves the number of particles, and hence, this operator may be thought of as a finite-dimensional Hamiltonian. In particular, denoting by $\mathcal H_n$ the set of states with less than $n$ particles, this shows that $\textsf{CVLH}^d_{\mathcal H_n}$ is \textbf{QMA}-hard already for $d=4$ (and thus $\textsf{CVLH}^4_{\mathcal H}$ as well).

In general, the complexity of the problem for non-Gaussian Hamiltonians depends significantly on the complexity of the set of states one optimizes.
Following \cite{chabaud2020stellar, chabaud2022holomorphic}, we consider the stellar rank of a quantum state as a parameter for specifying this set of states (see \cref{sec:cv-review} for a brief review of the stellar rank). In short, to any continuous-variable quantum state $\ket{\psi}$ over $n$ modes, one can associate a holomorphic function $F^\star_\psi: \mathbb{C}^n \rightarrow \mathbb{C}$. When this holomorphic function can be decomposed as a product of a polynomial $P$ and a Gaussian $G$, i.e., $F^\star_\psi (z_1, \cdots, z_n) = P (z_1, \cdots, z_n) G(z_1, \cdots, z_n)$, the degree of the polynomial $P$ defines the stellar rank $r$ of $\ket{\psi}$. Otherwise, the stellar rank is infinite. When $r = 0$, we obtain the set of all Gaussian states, which can be produced by Gaussian gates applied to the vacuum state. The stellar rank can be finite (e.g., the stellar rank of $n$ indistinguishable particles is $n$). It can also be infinite, e.g., $\alpha_1\ket{\psi_1} + \alpha_2\ket{\psi_2}$ for Gaussian states $\ket{\psi_1}, \ket{\psi_2}$ ($\alpha_1,\alpha_2\in\mathbb{C}$). We further constrain the energy of these states with respect to the number operator $N$ (see \cref{def:e-cons-stellar} for a formal definition).

For zero stellar rank, we show (see \cref{thm:gaussian-optimization} for a formal statement):

\begin{theorem}[Lowest energy over Gaussian states, informal]
    The problem of estimating the lowest energy of an $n$-mode polynomial bosonic Hamiltonian of constant degree over the set ${\mathcal S^{\mathsf{exp}}_0}$ of states of stellar rank $0$ (Gaussian states) with energy (average particle number) at most $\mathsf{exp}(n)$ is \textbf{NP}-complete. 
\end{theorem}

\begin{proof}[Proof sketch]
    \textbf{NP}-hardness comes from a reduction from matrix non-copositivity (rather than matrix copositivity for \cref{th:informal_boundedness}). To place the problem in \textbf{NP}, note that the expectation value of a constant-degree polynomial Hamiltonian over a Gaussian state with energy at most $\mathsf{exp}(n)$ may be computed efficiently, so the \textbf{NP} witness is a description of the Gaussian state of lowest energy.
\end{proof}

For logarithmically and polynomially-bounded stellar ranks, we show (see \cref{thm:bounded-squeezing-and-displacement} for a formal statement):

\begin{theorem}[Lowest energy over bounded stellar rank, informal]
    The problem of estimating the lowest energy of an $n$-mode polynomial bosonic Hamiltonian of constant degree over the set ${\mathcal S^{\mathsf{poly}}_{\mathsf{poly}}}$ of  states of stellar rank $r=\mathsf{poly}(n)$ with energy (average particle number) at most $\mathsf{poly}(n)$ is in $\mathbf{QMA}$.
\end{theorem}

\begin{proof}[Proof sketch]
    To place the problem in \textbf{QMA}, we use the fact that any state of finite stellar rank is related by a Gaussian unitary $G$ to a state of bounded particle number $\ket c$ \cite{chabaud2022holomorphic}. This allows us to efficiently rewrite the optimisation in terms of a (sparse) finite-dimensional Hamiltonian, once the Gaussian unitary corresponding to the lowest energy state is known. The \textbf{QMA} witness is then given by a classical description of that Gaussian unitary $G$ provided in the computational basis, together with a finite-dimensional state $\ket c$ that is the ground state of the finite-dimensional Hamiltonian, such that $G\ket c$ is the lowest energy state of the original Hamiltonian. 
    
    When the stellar rank is logarithmically bounded instead, $\ket c$ has an efficient classical description and the witness can be made fully classical.
\end{proof} 

Finally, we consider the general case of unbounded stellar rank, with no restrictions on the set of states over which the optimisation takes place, and we show (see \cref{thm:undecidability} for a formal statement):

\begin{theorem}[The complexity of the continuous-variable local Hamiltonian problem, informal]\label{thm:informal-gs-re-hard}
    The problem of estimating the ground energy of a polynomial bosonic Hamiltonian of constant degree is undecidable. This problem is already undecidable for $d=8$.
\end{theorem}

\begin{proof}[Proof sketch]
The proof of this result proceeds via a reduction from Hilbert's tenth problem (a N\"ullstellensatz problem over the integers) due to \cite{kieu2003quantum}, combined with the fact that there exist explicit undecidable polynomial equations over non-negative integers of degree $4$ and a constant number of unknowns \cite{jones58three}.
\end{proof}

\paragraph{Continous-Variable Quantum Merlin Arthur games.} Finally, we introduce a continuous-variable analog of \textbf{QMA} (see \cref{def:CVQMA} for a formal definition):

\begin{definition} [Continuous-variable quantum Merlin-Arthur, informal]
$\mathbf{CVQMA}$ is the class of decision problems, for which a solution encoded in a continuous-variable quantum state can be verified efficiently by a $\mathbf{CVBQP}$ machine.  
\end{definition}

\noindent Motivated by the relationship between the local Hamiltonian problem and \textbf{QMA} in the discrete-variable case based on Kitaev's history state construction \cite{kempe2006complexity}, we aim to connect the class \textbf{CVQMA} to the continuous-variable local Hamiltonian problem. We give the basis for a continuous-variable history state construction providing such a connection (see \cref{sec:history}), and identify the challenges associated with such a construction.





\subsection{Outlook}

The study of continuous-variable quantum computations may have various interactions with the foundations of computational complexity, computability, and quantum mechanics, which we discuss in the next section. After that, in \cref{sec:open}, we list several open questions.

\subsubsection{Discussion}
\label{sec:discussionE}

\paragraph{Energy as a fundamental resource in continuous-variable computations.} 
One of the key insights of this work is that by alternating Gaussian and cubic phase gates, the average energy of the system—even for a single mode—can grow to be doubly exponential in the number of cubic phase gates used in the circuit. This observation suggests that strong simulation (calculating each amplitude individually) of continuous-variable quantum systems could be significantly more challenging than for discrete-variable ones, potentially being hard for complexity classes like \textbf{PSPACE} or even \textbf{EXPSPACE}\footnote{A follow-up work \cite{upreti2025bounding} has since improved the uupperbound on \textbf{CVBQP}$[X^3]$ to \textbf{PSPACE}, which still leaves open the possibility for a gap.}.
This suggests that on top of time and space complexity, energy plays a significant role in the computational power of bosonic systems.

In practical physical experiments, the energy must be supplied by the experimentalist, requiring a physical definition of computational cost that accounts for time, space, and energy. Specifically, if one is willing to expend up to \( E^* \) units of energy, our results show that a single mode can be simulated within $\mathbf{SPACE}(O(\log E^*))$. This highlights a trade-off between time, space, and energy, which we believe deserves a thorough examination. Understanding time, energy, and space tradeoffs for multimode systems equipped with multi-mode non-Gaussian gates is an interesting open question.

In our definition of \textbf{CVBQP}, we enforce the ``promise'' that the quantum state in the beginning and the end has limited energy, as measured by the average particle number, but it may take any value in between. A natural question to ask is: is the power of this model the same as the one where we impose restrictions on average energy at any point during the computation? See \cref{fig:energy-promise} for a visual depiction. We can ask a similar question about classical continuous (time and/or space) models of computation where we are promised that the system's state is effectively discrete (e.g., a two-level system) at the beginning and the end. Still, the system can utilize its continuous degrees of freedom in between to an arbitrary precision. Is the power of this hypothetical model the same as digital computation? We first note that without such promise in the beginning, it is not clear how one can program such a system, and without the promise on the energy in the end, no physically viable device would be able to measure the quantum state. Even if we assume the existence of such a hypothetical device, it is not difficult to approximately predict the output of such a hypothetical device (we output random numbers because, at high energy, no concentration of measure is expected).

We may face the objection that assuming no mechanism to bound energy (e.g., dissipation), this promise is unreasonable because we cannot verify it. We first note that it is not difficult to come up with trivial examples that satisfy such a promise. Consider the unitary $U U^\dagger$ where $U\ket{0}$ has high energy. Clearly, $U U^\dagger$ has bounded energy at the beginning and the end and very high energy in between. Second, promises that are possibly difficult to verify are common in quantum complexity theory, e.g., promise on the gap of a Hamiltonian or promise that a $\textbf{BQP}$ computation either accepts with probability $\geq 2/3$ or $\leq 1/3$. We have a similar scenario for the energy promise. 

As a thought experiment, assume we have a fragile quantum processing device that breaks if it holds more than $N_0$ particles. Now, we design a quantum experiment such that the system (on average) has $\ll N_0$ particles at the beginning and the end but may (on average) have $\gg N_0$ particles in a way that many of the computational paths involve $\ll N_0$ particles and many involve $\gg N_0$. In the end, we measure the device's output and measure $\ll N_0$ particles. Do we measure the device to be broken or unbroken? This is similar to the Schr\"odinger's cat (or Wigner's) paradox. In the mentioned thought experiment, the device played the cat's role. In other words, if the device's condition (i.e., broken or unbroken) is determined only at the time of measurement, then \textbf{CVBQP} with mild (or no) energy restriction in between might be a plausible model. Otherwise, if broken paths are forbidden, then \textbf{CVBQP} with energy bound on the entire computation path is a more reasonable model from a physical standpoint. We note that in an actual experiment, energy is pumped from an outside source, and the closedness of the experiment is only an approximation.

\paragraph{Connections with the extended Church--Turing thesis.} Our result leaves open the possibility that \textbf{CVBQP} $\nsubseteq$ \textbf{BQP}; due to doubly exponential growth of energy in between computations, such a separation is plausible. What would that imply about the nature of computation in the physical world?
If \textbf{CVBQP} is a plausible model for computation in the physical world, then \textbf{CVBQP} $\nsubseteq$ \textbf{BQP} seems to imply a contradiction to the \textit{extended Church Turing Thesis}. But how realistic is \textbf{CVBQP} as a model of computation? In particular, how should one determine the energy cap in between the computations? Consider scattering amplitude for quantum field theories, where the specific amplitude sets the number of particles at the beginning and the end, and still, fluctuations in the vacuum may lead to many particle creations, and annihilations can occur in between. Are bounds on the fluctuations of the vacuum (see, for instance, \cite{ford2010negative}) such that the computational complexity of scattering amplitudes do not exceed that of \textbf{BQP}? Are there computational phase transitions depending on $\gamma_1$ and $\gamma_2$?

What can we say about noisy systems? Suppose we define $\mathbf{DissCVBQP}(\gamma_1, \gamma_2)$ as a dissipated model where we have a mechanism that pumps bosons into the system with rate $\gamma_1$ and another mechanism that bosons are emitted to the environment with rate $\gamma_2$. How does the power computational complexity of $\mathbf{DissCVBQP}(\gamma_1, \gamma_2)$ depend on $\gamma_1$ and $\gamma_2$? Can we show for physically relevant parameters the power of this model is \textbf{BQP}-complete?



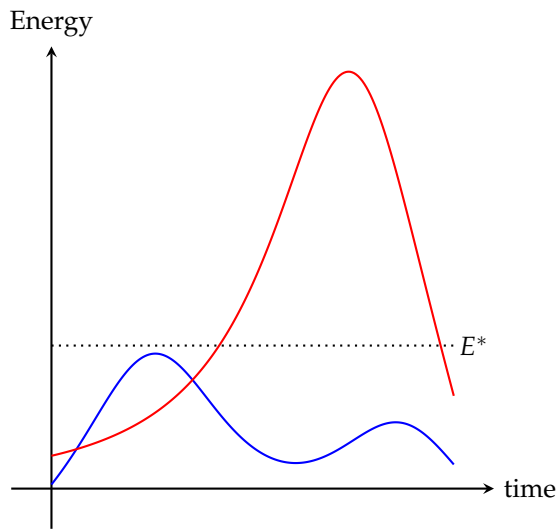
\begin{figure}
\centering
\begin{tikzpicture}
    \begin{axis}[
        axis lines=middle,
        xlabel={time},
        ylabel={Energy},
        grid=both,
        samples=100,
        domain=0:5,
        ymin=0, ymax=2.5,
        width=8cm,
        height=8cm,
        thick,
        axis line style={thick},
        xtick=\empty, 
        ytick=\empty, 
        enlargelimits,
        xlabel style={right},
        ylabel style={above}
    ]
    \addplot[domain=0:4, smooth, thick, blue] {1.5/(1+(x-1)^2) + 1/(1+(x-3.5)^2)-1.6/2};

    \addplot[domain=0:4, smooth, thick, red] {3/(1+(x-3)^2) - 3/(1+(x-5.5)^2)};
    
    \draw[dotted, thick] (axis cs:0,1.77898/2) -- (axis cs:4,1.77898/2);
    
    \node at (axis cs:4.2,1.77898/2) {$E^\ast$};
    
    \end{axis}
\end{tikzpicture}
\caption{Energy bound conditions for quantum computations. The red curve depicts a system that is promised to have low energy at the beginning and the end of the computation, but may have arbitrary high energy in between. The blue curve depicts a model where we impose energy restrictions at any point during the computation.}
\label{fig:energy-promise}
\end{figure}

\paragraph{Logarithmic number of cubic phase gates.} Standard results in the discrete-variable model of quantum computing (e.g., \cite{bravyi2016trading}) imply that starting with a computational basis, discrete-variable quantum circuits with a logarithmic number of $T$ gates and polynomially many Clifford gates can be efficiently strongly simulated on classical computers. Can we prove a similar result for continuous-variable systems, i.e., can we show that starting from the vacuum state, a polynomial-size quantum circuit with a logarithmic number of cubic phase gates can be simulated on a classical computer? Based on our results, the best upper bound we could prove for this problem was \textbf{PSPACE}. We furthermore note that, due to repeated squaring of energy, the "effective dimension" that the quantum system explores is exponential even if the circuit has, at most, a logarithmic number of cubic phase gates. Based on this intuition, we conjecture that (noiseless) continuous variable quantum circuits with a logarithmic number of cubic phase gates cannot be simulated efficiently on a classical computer. 

\paragraph{Noiseless model as a foundation to study noisy systems.}
Even though a realistic model for the experimental implementation of bosonic computations is a dissipative (or noisy) model, noise models are various, and by understanding the computability limits of the noiseless system, one can gain a reliable foundation to study computational complexity after imposing different restrictions on the model. A similar approach has been pursued for computational complexity in the discrete-variable domain where we define \textbf{BQP} as an idealized noiseless quantum computing model and pin down its computational complexity. We then use this foundation to study the model's different variants and restrictions.  For instance, we know that computing the amplitudes of a noiseless quantum circuit is $\#\textbf{P}$-complete. Obviously, this result is not directly relevant to understanding the cost of simulating practical problems. However, this observation has been utilized via tools in computational complexity to lay the foundations for demonstrating quantum speedup in noisy or restricted models such as Boson Sampling \cite{aaronson2011computational}. Moreover, we describe a way to encode BQP computations in bosonic subspaces (see \cref{th:CVBQPcontainsBQP}) relying on the Solovay--Kitaev theorem \cite{dawson2005solovay}, but this result does not appear to be robust to noise, because the noise may induce leakage to unbounded regions of state space where the theorem no longer holds. This motivates further study of robust quantum computations in continuous-variable systems.

\paragraph{Truncating CV systems using the stellar rank.} Standard approaches to describe continuous-variable quantum systems using a finite number of variables include (i) truncating their infinite-dimensional Hilbert space, which amounts to restricting to bounded particle-number supports, or (ii) keeping track only of the covariance matrix and displacement vector of the state. These have immediate shortcomings: (i) is not stable under Gaussian evolutions, which may be thought of as computationally easy \cite{bartlett2002efficient}, and truncating a Gaussian state usually makes it non-Gaussian; (ii) only faithfully describes Gaussian computations. 

Our results suggest that the stellar rank may be a meaningful approach for describing continuous-variable quantum systems, in particular for studying ground state problems for continuous-variable systems \cite{stornati2023variational}. Informally, the stellar rank combines both approaches (i) and (ii), as states of finite stellar rank can always be expressed as (a mixture of) Gaussian unitary operators acting on (core) states of finite support.

\paragraph{Comparison with classical models of computation over continuous variables:} Several classical models of computation based on continuous-variable degrees of freedom have been studied, and our work is partly motivated by developing a quantum generalization of these frameworks. A prominent example is the classical Blum–Shub–Smale (BSS) model of computation \cite{BSS89}. In this model, machines operate over arbitrary real numbers: each memory cell may store a real value, and basic arithmetic operations (such as addition and multiplication) are performed at unit cost. The original motivation for introducing this model was to provide a rigorous computational foundation for areas of continuous mathematics, including topology and algebraic geometry.

A natural $\mathsf{NP}$-complete problem in the BSS model is the following: given a system of multivariate polynomials, decide whether they admit a common real zero. In this sense, the BSS model captures the complexity of fundamental decision problems in real algebraic geometry. The bosonic quantum computational model without energy restrictions may be viewed as a quantum analogue of the full BSS model. We conjecture that the problem of deciding if given a bosonic Hamiltonian, the lowest energy Gaussian states has energy less than $0$ or above an inverse polynomial threshold is complete for $\mathsf{NP}$ for the BSS model; right now we can prove $\mathsf{NP}$-completeness for standard Turing machines  for the polynomially bounded energy version of this problem. 

In addition to the standard BSS framework, weaker variants have also been proposed. For example, Koiran~\cite{koiran1997weak} introduced a modified model in which arithmetic operations incur a cost that depends on parameters such as the size or algebraic complexity of the numbers stored during the computation. This weak model was shown to be equivalent to $\mathsf{P}/\mathrm{poly}$.
The bosonic quantum model with an energy constraint can be interpreted as a quantum generalization of this weaker BSS framework. Indeed, subsequent work~\cite{chabaud2025energy} shows that the energy-restricted bosonic model is contained in $\mathsf{BQP}/\mathrm{poly}$.

For a more comprehensive review of analog classical models of computations see~\cite{bournez2021survey}.

\paragraph{Comparison between Gaussian computation and dynamics of coupled oscillators:} We also note recent work on the computational complexity of networks of classical and quantum oscillators \cite{babbush2023exponential, barthe2024gate}. In particular, \cite{barthe2024gate} shows that simulating exponentially many Gaussian oscillators is $\mathsf{BQP}$-complete. Our results are complementary to this line of work. Specifically, in Theorem 3.1 we prove that simulating polynomially many Gaussian oscillators, each with at most polynomially bounded energy, is $\mathsf{BQL}$-complete.

\paragraph{Comparison with Hybrid Oscillator-Qubit Devices:} Recently, there has been a growing line of work exploring the power of hybrid Oscillator-qubit, aka Hybrid CV-DV, systems \cite{crane2024hybrid, liu2026hybrid, andersen2015hybrid, brenner2025factoring}. In such systems, qubit-controlled continuous Hamiltonians such as $\ket{1}\bra{1}\otimes \hat X$, or $\hat \sigma_z \otimes \hat N$ have been used to achieve non-Gaussianity in the bosonic modes. A major advantage of these devices is that qubit-controlled continuous operations enable more stable application of non-Gaussian operations. These models furthermore inherit capabilities of both discrete and continuous-variable modes. They were furthermore shown to be advantageous for applications such as simulating continuous systems, including lattice gauge theories \cite{liu2026hybrid}, compared with discrete-variable devices, and were used as variational ansatzes to find the ground energy of these systems. Exploring computational complexity in this family of devices is an important and interesting line of work. Recently, \cite{brenner2025factoring} showed that one can solve the factoring problem with only three oscillators and a single qubit. Indeed, their algorithm requires an exponential amount of energy, demonstrating that one can trade off system energy for other resources, such as the number of modes. More recently, this line of work was extended to characterize trade-offs between energy and the number of modes in hybrid oscillator-qubit devices \cite{brenner2025trading}.

An important future direction for this work is to compare the computational power of Hybrid Oscillator-Qubit devices with that of pure oscillator devices, such as those considered here. A critical difference between these two models is that the energy growth in a pure continuous-variable device with oscillator-oscillator interactions, such as the one we consider in this work, can increase dramatically faster than that of oscillator-qubit devices. Rapid energy growth can be both an opportunity and a curse. For instance, in designing continuous-variable fault-tolerance schemes, one needs to encode logical information in low-energy subspaces and dissipate excess non-useful energy from the system~\cite{matsuura2024continuous}. An important question to study is whether oscillator-qubit interactions can be reliably embedded in oscillator-oscillator interactions. Theorem~\ref{th:CVBQPcontainsBQP} gives a recipe to embed discrete variable operations inside oscillator modes (see also~\cite{arzani2025effective} for a recent development). However, these embeddings may involve high-degree oscillator-oscillator interactions. So, embedding oscillator-qubit interactions into oscillator-oscillator interactions is, in principle, possible but may be complex. Further work is required to study the practicality of these embeddings and their robustness to noise. An interesting line of work would be to study the resource-theoretic foundations of a model that combines non-Gaussian oscillator-oscillator interactions with the hybrid (qubit-controlled gate) set studied in~\cite{liu2026hybrid, brenner2025factoring}.

Another interesting future direction is studying Hamiltonian complexity for systems with qubit-oscillator interactions, such as the Jaynes-Cummings model. As described in \cref{sec:history}, a major obstacle in defining circuit-to-Hamiltonian reductions for a purely bosonic Hamiltonian is that the standard construction in discrete variables involves projectors onto the computational basis (such as qubit $\ket{0}\bra{0}$ in the initialization Hamiltonian term). However, basic projectors onto Fock or coherent states are not polynomials in the standard quadrature basis. Even if we use the embedding of discrete variable construction inside the oscillator mode, we obtain very high-degree Hamiltonian terms. By adding qubit-controlled oscillator operations, we can get around this issue and use qubit projectors. As a result, an interesting intermediate direction is understanding Hamiltonian complexity for Hybrid systems such as the Jaynes-Cummings model. 

\subsubsection{Open questions}
\label{sec:open}

Since the aim of the present work is to lay the foundations for a theory of bosonic quantum complexity, it naturally leads to many open questions, some of which we list in the following:

\begin{enumerate}
    \item The most immediate open question from this work is whether we can bring the \textbf{EXPSPACE} upper bound on \textbf{CVBQP} to smaller complexity classes such as \textbf{PSPACE} or even \textbf{PP}. What about lower bounds? Due to the doubly exponentially large dimensionality of the effective Hilbert space, it is natural to conjecture that strong simulation of \textbf{CVBQP} is hard for a complexity class such as \textbf{PSPACE} (or even \textbf{EXPSPACE}) which is believed to be strictly larger than \textbf{PP}. Such a result would be important evidence that \textbf{CVBQP} may surpass the power of \textbf{BQP}.

    \item A related question is: what energy bound on \textbf{CVBQP} makes it equal to \textbf{BQP}? In an actual physical computation (such as a bosonic system subject to dissipation), one would expect that the energy stays polynomially bounded. Is it the case that $\textbf{CVBQP}=\textbf{BQP}$ under the promise that the energy stays polynomially bounded throughout the computation? Can we prove that a variant of $\textbf{CVBQP}$ subject to dissipation is equivalent to \textbf{BQP}? 
    
    \item It is usually assumed in the continuous-variable quantum information literature that a single non-Gaussian gate together with all Gaussian ones is sufficient to perform ``universal'' quantum computations. However, this notion of universality is somewhat restricted, as it relates to the ability to approximate evolutions generated by polynomial Hamiltonians \cite{lloyd1999quantum}. In particular, is the cubic phase gate (for instance) and Gaussian gate set \textit{universal}, in the sense that \textit{any state} can be reached to arbitrary precision from the vacuum state using unitary gates from this set? See \cite{wu2006smooth} for a formal statement of this open problem, and \cite{mendes2011problem} for an example of a continuous-variable gate set satisfying this property.
    
    \item In the discrete-variable setting, the computational power of quantum circuits is essentially independent of the choice of universal gate set. Is it also the case in the continuous-variable setting? If so, can a continuous-variable Solovay--Kitaev theorem \cite{kitaev1997quantum,dawson2005solovay} be derived for these gates? See \cite{becker2021energy} for such a result in the case of Gaussian gates. This also relates to the existence of fast compilation algorithms for bosonic gates \cite{sefi2011decompose,kalajdzievski2019exact,kalajdzievski2021exact}.
    
    \item What is the precise complexity of the Hamiltonian spectrum boundedness problem? We prove \textbf{co-NP}-hardness for $ d=4$ and conjecture hardness for \textbf{co-QMA}. Can we show \textbf{co-QMA} is an upper bound? We describe a sound algorithm for verifying boundedness in the $d = 4$ case. Do similar results exist for $d>4$? We know that the problem becomes undecidable for $d\geq 10$. 
    
    \item After possible conjugations with arbitrary unitary matrices, can we write a Hamiltonian as a sum-of-squares of other Hamiltonians (c.f.\ \cite{hastings2022perturbation})? We have obtained partial results in this direction for degree-4 polynomial bosonic Hamiltonians.
    

    \item We can define families of Hamiltonians with a ground state of a given exact or approximate stellar rank $r$. Is there a procedure to decide the opposite? I.e., given the description of the Hamiltonian and an approximation parameter, decide whether it has an approximate ground state of stellar rank $<r$.


    \item Our results indicate if we minimize the energy of a continuous-variable Hamiltonian over an ensemble with polynomially bounded energy (particle number) and stellar rank $r$, then for $r = \mathsf{poly}(n)$ it is contained in \textbf{QMA}. Can we show that these containments are tight? What about $r = \mathsf{exp(n)}$ and higher particle numbers? Proving such results would involve continuous-variable versions of Kitaev's history state \cite{kempe20033,kempe2006complexity} in the continuous-variable setting. 
    This construction would involve nontrivial ideas, which we leave for future work.

    \item What is the complexity of the continuous-variable local Hamiltonian problem for other natural families of continuous-variable quantum states? For instance, the family of quantum states that are superpositions of $\mathsf{poly}(n)$ Gaussians? This class of states has recently been considered in the context of classical simulation of bosonic computations, leading to the introduction of the continuous-variable notion of \textit{Gaussian rank} \cite{dias2024classical,hahn2024classical}, akin to the discrete-variable notion of \textit{stabilizer rank} \cite{bravyi2016trading}. Note that the stellar rank for such states is typically infinite.

    \item What is the complexity of deciding whether the spectrum of a polynomial bosonic Hamiltonian has a continuous part or if is it discrete? When it is discrete, is \textsf{CVLH} \textbf{RE}-complete?
\end{enumerate}

Bosonic quantum computations also provide a natural setting for generalizing classical complexity theory over the reals (e.g., in the Blum--Shub--Smale model \cite{blum1998complexity}) to the quantum case, given that bosonic Hamiltonians may have a continuous spectrum. We expect that most of our results should have counterparts over the reals.

\subsection{Acknowledgements}

S. M.\ and M. J.\ are grateful to the National Science Foundation (NSF CCF-2013062) for supporting this project. U.C.\ acknowledges support from the Plan France 2030 project NISQ2LSQ (ANR-22-PETQ-0006). U.C.\ and S.M.\ are grateful to have been supported by the Institute for Quantum Information and Matter, an NSF Physics Frontiers Center (NSF Grant PHY- 1733907), where preliminary ideas were discussed. U.C.\ acknowledges inspiring discussions with F.\ Arzani, R.\ I.\ Booth, F.\ Grosshans, D.\ Markham, O.\ Bournez, S.\ Lloyd, T.\ Vidick, W.\ Slofstra, L.\ Lami, and A.\ Winter. S.M.\ thanks A.\ Natarajan, P.\ Love, V.\ Podolskii, V.V.\ Albert, M.\ Hafezi, A.\ Gorshkov, D.\ Jacobs, and J.\ Jeang for inspiring discussions. A.M.\ acknowledges discussions with P.\ Sinha, Z.\ Mann, and R.\ Cleve. M. J.\ thanks V.\ Podolskii, D.\ Jacobs, W.\ Khangtragool and A.\ Saha for helpful exchanges. We also thank S. Gharibian and Dorian Rudolph for insightful conversations. The authors are grateful to S.\ Aaronson for helpful discussions and to T.\ Vidick and F.\ Grosshans for comments on earlier versions of this manuscript.


\section{Preliminaries}
\label{sec:preliminaries}

We refer the reader to standard textbooks (e.g., \cite{nielsen2001quantum} for quantum information theory, \cite{ferraro2005gaussian,weedbrook2012gaussian} for continuous-variable quantum information theory, \cite{vourdas2006analytic} for analytic representations in quantum mechanics, \cite{schmudgen2020invitation} for unbounded operator algebras, \cite{arora2009computational} for classical Boolean complexity theory, 
and \cite{watrous2008quantum} for quantum complexity theory) and provide a brief overview of the necessary material hereafter.

\subsection{Computational complexity theory}

We assume familiarity with basic complexity classes such as polynomial time (\textbf{P}), nondeterministic polynomial time (\textbf{NP}), bounded-error polynomial time (\textbf{BPP}), bounded-error quantum polynomial time (\textbf{BQP}), logarithmic space (\textbf{L}), polynomial space (\textbf{PSPACE}), and bounded-error quantum logspace (\textbf{BQL}).
\subsubsection{Polynomial satisfaction problems}\label{sec:constraint-satisfaction}

Consider Hilbert's tenth problem: given a polynomial $P(X_1, \cdots, X_M)$ with integer coefficients, is it satisfiable? In other words,
\begin{equation}
\exists ? (n_1, \cdots, n_M)\in\mathbb Z^M: P(n_1, \cdots, n_M) = 0.
\end{equation}
David Hilbert asked in his tenth problem \cite{hilbert1902mathematical}  whether there exists an algorithm to make this decision. More generally, Hilbert was interested in whether there exists an algorithm that takes the statement of a theorem and tells us whether or not the statement is true based on a more fundamental axiomatic system. Developments at the beginning of the 20th century, specifically the works of G\"odel and Turing, proved the impossibility of this program by demonstrating that for any axiomatic system that is strong enough to reason about basic arithmetic operations, there always exists a statement that is ``true but not provable''. A specific instance of such a problem is \textsf{Halting}, which, in modern language, is the problem of deciding whether a computer program ever halts. It is a canonical \textbf{RE}-complete problem, as Turing showed that \textsf{Halting} is undecidable on a Turing machine. Around seventy years after Hilbert's original proposal, the pioneering works of Matijasevi\u c, Davis, Putnam, and Robinson (cf \cite{M71, D73, DMR76}), proved that Hilbert's tenth problem is equivalent to \textsf{Halting}. This was later refined to optimise the number of unknowns and the degree of the polynomial, showing in particular that degree-$4$ equations are sufficient for undecidability (see, e.g., \cite{jones58three}).

While polynomial satisfiability is unsolvable using Turing machines, the constraint satisfiability problem over finite fields, such as $\mathbb F_2$, is captured by the complexity class \textbf{NP}. Recall that \textbf{NP} is the class of problems for which a ``yes'' answer can be verified using a polynomial-time Turing machine. For instance, the problem of deciding whether a system of polynomials over $\mathbb F_2$ has a common zero is complete for \textbf{NP} \cite{smale1998mathematical}.

\subsubsection{Parallel computations}
\label{sec:parallelComplexity}

This section recalls a few details of classical parallel computation and highlights a key result that characterizes the complexity of parallel algorithms. The tools in this section will be used in \cref{sec:CVBQP}.

Recall that the \textbf{random access machine} (RAM) is a convenient model for a sequential classical computer, consisting of a single processor or computation unit that executes a user-defined program. It possesses a read-only input tape, a write-only output tape, and an unbounded number of local memory cells, each of which may store an integer of unbounded size. A program executed by a RAM consists of a labeled sequence of instructions, consisting of primitive operations like addition, subtraction, branching and comparisons.

The canonical model for classical parallel computation is a generalization of the above model called a \textbf{parallel random access machine} (PRAM), which consists of a number of RAMs that can access a shared memory in parallel. The machines can read and write to this shared memory in $O(1)$ time, and can communicate with each other in $O(1)$ time. Formally,

\begin{definition}[PRAM]
    A PRAM consists of a sequence of processors $\set{P_0,P_1,\cdots}$, an unbounded global memory labelled by countably-infinite sequences $(x_0,x_1,\cdots) \in \mathbb Z^{\infty}$, and a finite program. Each processor is a machine with unbounded local memory $(y_0,y_1,\cdots) \in \mathbb Z^\infty$, a program counter and a flag indicating whether the machine is running or not. 
    A program consists of a finite sequence of (labeled) unit cost instructions. (cf. \cite{trahan1992multiplication} for a full specification)
    \label{def:pram}
\end{definition}

The fork operation above is what enables the parallelism in the PRAM model. If a processor $P_i$ forks on $m$, it can activate an unutilized machine $P_j$ and initialize its accumulator with its own, such that $P_j$ starts its computation at label $m$ of the program. In each time step every processor can read from memory, perform a basic instruction as above and write to memory. The running-time of a PRAM program is the maximum number of steps taken by any processor before it halts.

For convenience, we assume that the processors can concurrently read from, and write to (abbreviated as, CRCW) the shared memory. If multiple processors attempt to write to the same location, the value written is the one from the processor with the lowest index. If processors are simultaneously trying to read from and write to the same global memory location, the reads are performed before the write. Note that since other kinds of read/write configuration PRAMs can simulate CRCW PRAMs with at-most logarithmic overhead \cite{balcazar2012structural}, there is no loss of generality for a polynomial-length computation.  Note also that the above is a variation on the standard definition of the PRAM assuming unit-cost multiplication as a primitive operation, sometimes denoted PRAM[*]. One surprising result of the work of Trahan et al. \cite{trahan1992multiplication} is that a PRAM[*] can be simulated by an ordinary PRAM with at most polynomial overhead. Let PRAM[*]-PTIME be the class of languages recognized by a PRAM[*] in time polynomial in the length of the input. Then from the same work as the above,

\begin{theorem}[\cite{trahan1992multiplication}]\label{thm:PRAMtimeEqualsPSPACE}
    \(PRAM[*]-PTIME = PSPACE\)
    \label{thm:par-pspace}
\end{theorem}

\subsection{Hamiltonian complexity theory over discrete variables}

The quantum analog of the constraint satisfaction problem is the problem of estimating the ground state of a local Hamiltonian. In quantum physics, Hamiltonians are self-adjoint operators which encode the energy levels of a physical system. The ground state energy is the lowest energy level. It is a significant quantity because, at low temperatures, systems are believed to be in the vicinity of the ground state. Formally, the local Hamiltonian problem is:

\begin{problem} [The local Hamiltonian problem]
    Let $n$ be an integer and $m = n^{O(1)}$. Given a Hamiltonian $H = \sum_{i=1}^m H_i$ over $n$ qudits, where $H_i$ acts non-trivially only on $\leq k$ qudits. Let
    \begin{equation} 
    E = \min_{\substack{X \in \mathbb C^n\\ \|X\|_2 = 1}} \langle X|H|X\rangle.
    \end{equation}
    Given two numbers $a - b > 1/\mathsf{poly}(n)$ decide whether $E \geq a$ or $E < b$, promised that one of these is the case. 
\end{problem}
One can view each term $H_i$ as a constraint. In particular, if each $H_i$ is a projector, and $a = 1/\mathsf{poly}(n)$ and $b = 0$, then we can interpret $H_i |X\rangle = 0$ as $H_i$ accepting $\ket X$. A Hamiltonian is called \textbf{frustration-free} if the ground state simultaneously satisfies all constraints. In other words, if $H_i$ is a projector, then the problem of deciding whether $H$ is frustration-free is the problem of deciding whether

\begin{equation} 
\exists ? X: H_1 \ket X = \cdots = H_m \ket X = 0.
\end{equation}

Kitaev showed that this problem is complete for the quantum complexity class \textbf{QMA}, where \textbf{QMA} is the class of problems that can be verified efficiently on a quantum computer:

\begin{definition}[\textbf{QMA}]
    \textbf{QMA} is the class of languages $L$ for which there exists a uniform family of polynomial-time quantum verifiers $V_n(x, \ket \psi)$ such that on input $x \in \{0,1\}^*$:
    
    \begin{enumerate}
        \item If $x \in L$ there exists $\ket \psi \in \mathbb C^{\mathsf{poly}(n)}$ such that $V_{|x|} (x, \ket \psi)$ accepts with probability $\geq 2/3$.
        
        \item If $x \notin L$ then for any $\ket \phi \in \mathbb C^{\mathsf{poly}(n)}$ $V_{|x|} (x, \ket \phi)$ rejects with probability $\geq 2/3$.
    \end{enumerate}
\end{definition}

\noindent 
In the definition of \textbf{QMA}, if the verification is performed by a polynomial-time classical computer, then we obtain \textbf{MA}, where \textbf{MA} is a generalization of \textbf{NP} by allowing probabilistic verification with bounded-error. In summary, the relationship between these classes can be captured by
\begin{equation}
    \textbf{NP}\subseteq\textbf{MA}\subseteq\textbf{QMA}.
\end{equation}

\subsection{Continuous-variable quantum information}
\label{sec:cv-review}

The (bosonic) quantum states of continuous-variable systems of a single mode may be identified with square-integrable real functions, the space of which ($L^2(\mathbb R)$) is isomorphic to the space of square-summable sequences $\ell^2 (\mathbb C)$. In simple words, continuous-variable quantum states have a discrete manifestation as a square-summable sequence of complex numbers. A pure state in $\ell^2 (\mathbb C)$ is a quantum state 
\begin{equation}
    \ket\psi=\alpha_0\ket0+\alpha_1\ket1+\cdots,
\end{equation}
such that $\sum_{n \geq 0} |\alpha_i|^2 < \infty$ and where $\{\ket n\}_{n\ge0}$ is an orthonormal basis known as the Fock basis in the context of bosonic systems, where $n$ captures the number of bosons per mode. The state $\ket 0$ is the state of zero particles, the so-called vacuum state. Fock basis vectors are eigenvectors of a quantum operator known as the number operator:
\begin{equation}
     N \ket m = m \ket m.
\end{equation}
An important family of bosonic states is known as the (Glauber) coherent states, which capture the states of important quantum systems such as the light coming out of a laser. For $\alpha \in \mathbb C$ we define
\begin{equation}
    \ket\alpha=e^{-\frac{|\alpha|^2}2}\sum_{n\geq 0}\frac{\alpha^n}{\sqrt{n!}}\ket n.
\end{equation}
Coherent states are eigenvalues of an operator known as the annihilation operator $ a$:
\begin{equation}
     a\ket\beta=\beta\ket\beta.
\end{equation}
The adjoint of this operator $ a^\dag$ is called the creation operator. These operators satisfy the canonical commutation relation $[ a, a^\dag]= I$, and $ n= a^\dag a$. Coherent states are not orthogonal. They satisfy instead
\begin{equation}
    |\langle\alpha|\beta\rangle| = e^{-\frac{|\alpha-\beta|^2}2}.
\end{equation}
These states actually form an \textit{overcomplete} basis. In particular, it is known that any convergent infinite subsequence of coherent states is overcomplete \cite{perelomov2002completeness}. As a matter of fact, for any point in the complex plane, coherent states with amplitudes corresponding to the complex numbers within an open ball around that point also span the whole Hilbert space. In what follows, we denote the continuous-variable Hilbert space by $\mathcal H$, with the number of modes depending on the context.

\subsubsection{Analytic representation and the stellar rank}
\label{sec:stellar_rank}

For any quantum state $\ket\psi=\sum_{n=0}^{\infty}\psi_n\ket n\in\ell^2(\mathbb C)$, we can define a function corresponding to its inner product with an unnormalized coherent state:
\begin{equation}
    F^\star_\psi(z) = e^{\frac{|z|^2}2} \langle \conj z|\psi\rangle=\sum_{n=0}^{\infty}\frac{\psi_n}{\sqrt{n!}}z^n,
\end{equation}
for all $z\in\mathbb C$, where $\conj z$ denotes the complex conjugate of $z$. $F^\star_\psi$ is known as the stellar function \cite{chabaud2020stellar}, or Bargmann function \cite{bargmann1961hilbert,segal1963mathematical}, of the quantum state $\ket\psi$. Because $\ket\psi$ is square-summable, $F^\star_\psi$ is square-integrable with respect to the Gaussian measure over the complex plane $\smash{d\mu(z):=\frac1\pi e^{-|z|^2}d^2z}$, where $d^2z=d\Re(z)d\Im(z)$. To see this, observe the following equivalence between discrete normalization and the Gaussian integration of stellar functions.

\begin{lemma}
    \begin{align}
        \sum_{n \geq 0} |\psi_n|^2 = \int_{z \in \C} d\mu(z) |F^\star_\psi(z)|^2\cdot
    \end{align}
\end{lemma}

\begin{proof}
We can expand:
\begin{equation}
    \int_{z \in \C} d\mu(z) |F^\star_\psi(z)|^2 = \sum_{m,n \geq 0} \conj\psi_{m} \psi_n \EE_{z \sim \mu}\left[\frac{\conj z^{m} z^n}{\sqrt {m! n!}}\right].
\end{equation}
The statement is implied by using the following moments of the complex Gaussian distribution:
\begin{equation}
    \EE_{z \sim \mu} [z^n \conj z^{m}] = \delta_{m,n} \cdot n!\;.
\end{equation}
\end{proof}

\noindent In summary, the stellar representation is a unitary representation over the space of entire functions (i.e., complex-valued functions that are converging series everywhere in the complex plane) which are square-integrable with respect to the Gaussian measure on the complex plane. This space is known as the Segal--Bargmann space. 

Due to the Hadamard--Weierstrass factorisation theorem, any stellar function $F^\star_\psi(z)$ of a single-mode pure state $\ket\psi$ admits the following canonical decomposition according to its zeros $\lambda_i$:
\begin{equation}
    F^\star_\psi(z) = z^k \prod_{n=1}^{r-k}\left(1-\frac z{\lambda_n}\right) e^{\frac z {\lambda_n} + \frac {z^2}{2 \lambda^2_n}} e^{a z^2 + b z + c},
\end{equation}
where $a,b,c\in\mathbb C$ and $k\in\mathbb N$ is the multiplicity of $z=0$ as a zero of $F^\star_\psi$.

The number $r\in\mathbb N\cup\{+\infty\}$ in the above description, known as the stellar rank \cite{chabaud2020stellar,chabaud2022holomorphic}, is the number of zeros of $F^\star_\psi$ counted with multiplicity. This representation is called the stellar representation because, up to a Gaussian function, $F^\star_\psi$ can  be captured by its complex zeros (stars), which can be thought of as the infinite-dimensional limit of the Majorana stellar representation for qudit states \cite{chabaud2022holomorphic}. When the stellar rank is finite, we consider the following canonical form for these states as
\begin{equation}
    F^\star_\psi(z) = P_r(z)G(z),
\end{equation}
where $P_r$ is a polynomial of degree $r$ and $G$ is a generalized Gaussian function. The stellar rank is extended to mixed states by a convex roof construction.

The representation of creation and annihilation operators acting on the Segal--Bargmann space are given by the derivative with respect to the complex variable and the multiplication by the complex variable, respectively:
\begin{equation}
     a^\dag\leftrightarrow\times z,\quad a\leftrightarrow\partial_z.
\end{equation}
The Segal--Bargmann representation of observables are hermitian functions of $\times z$ and $\partial z$. An observation of the observable $ O$ of the state $\psi$ gives the expectation value
\begin{equation}
    \langle\psi| O|\psi\rangle=\int_{z \in \C} \conj {F^\star_\psi(z)}O(z,\partial_z)F^\star_\psi(z)d\mu(z).
\end{equation}
\paragraph{Multiples modes:} All the above formulations can be extended to multiple modes. Suppose we have $m$ bosonic modes. The creation and annihilation operators acting on the mode $1\le i\le m$ are, respectively, denoted by $ a^\dag_i= \times z_i$ and $ a_i = \partial_{z_i}$. The stellar representation of a multimode system is given by a holomorphic Segal--Bargmann function $\bm z=(z_1,\dots,z_m,)\mapsto F^\star(z_1, \cdots, z_m)$ over $m$ complex variables. A multimode observable is an essentially self-adjoint function of $ a_1, a_1^\dag,\dots, a_m, a_m^\dag$. The multimode Gaussian measure is $\smash{d\mu(z_1,\dots,z_m)=\frac1{\pi^m}e^{-\|\bm{z}\|_2^2}}$. A stellar function $F^\star$ over $m$ modes is said to have finite stellar rank $r$, if 
\begin{equation}
    F^\star(z_1, \cdots, z_m) = P_r(z_1, \cdots, z_m) G(z_1, \cdots, z_n),
\end{equation}
with $P_r$ a multivariate polynomial of degree $r$, while $G$ is a Gaussian function which takes the most general form
\begin{equation}
    G(z_1, \cdots, z_m) = \frac 1 {\mathcal N} e^{- \frac 12\bm{z}^T A \bm z + B \bm z + C},
\end{equation}
where $B = U^T \bm b$, $U$ is an $m \times m$ unitary matrix, and $\bm b = (b_1, \cdots, b_m)^T$ is a vector of complex numbers, and
\begin{equation}
    A = U^T\mathrm{Diag}(t_1, \cdots, t_m) U,\;|t_i|<1.
\end{equation}
Otherwise, the stellar rank is defined to be $+\infty$.

A bosonic Hamiltonian $ H$ is an essentially self-adjoint operator (typically a function of the creation and annihilation operators), which encodes the energy levels of a continuous-variable quantum system. Its action in (a dense subspace of) the Segal--Bargmann space is a self-adjoint function of the multiplication and derivative operators and governs the dynamics of the system according to the Schr\"odinger equation $i \partial_t F^\star(\bm z,t) = H(\bm z,\partial_{\bm z})F^\star(\bm z,t)$. The solution to the Schr\"odinger equation is given by the unitary dynamics
\begin{equation}\label{eq:dynstellar}
    F^\star(\bm z,t)=e^{itH(\bm z,\partial_{\bm z})}F^\star_0(\bm z,t).
\end{equation}
Hereafter, a Hamiltonian $ H$ is of degree $d$ if it is a polynomial of degree $d$ in creation and annihilation operators. It is $k$-local if it is a sum of terms, each of which affects at most $k$ modes at a time. Note that any polynomial Hamiltonian of degree $d$ is at most $d$-local.

Besides the dynamical representation~(\ref{eq:dynstellar}), there is the usual circuit-like representation 
\begin{align}
    F^\star(\bm z,t) = U(\bm z,\partial_{\bm z})F^\star(\bm z,0) = U_t(\bm z,\partial_{\bm z})\cdots U_1(\bm z,\partial_{\bm z})F^\star(\bm z,0),
\end{align}
where each $U_i(\bm z,\partial_{\bm z})$ is a unitary operator acting on Segal--Bargmann space, i.e., $U_i(\bm z,\partial_{\bm z})=e^{i H_i(\bm z,\partial_{\bm z})}$ for some Hamiltonian $H_i$ affecting only a few qumodes.

\subsubsection{Continuous-variable quantum processes}\label{sec:CVprocesses}

A unitary operator $ G$ is called Gaussian if it a product of terms of the form $e^{i H}$ for quadratic Hamiltonians $ H$. 

\paragraph{Single mode:}
Consider the following special single-mode Gaussian processes:
\begin{align}\label{eq:simple-gaussian}
\begin{split}
 R(\phi) &= e^{i\phi  a^\dagger  a}\quad\quad\quad\text{(Rotation)}\\
 D(\delta) &= e^{\delta  a - \conj \delta  a^\dagger} \quad\quad\text{(Displacement)}\\
 S(\xi) &= e^{\xi  a^2 + \conj \xi { a^{\dagger2}} } \quad\;\text{(Squeezing)}.
\end{split}
\end{align}
They have the following action on the creation and annihilation operators:
\begin{align}
\begin{split}
 R(\phi)  a^\dagger  R(\phi)^\dagger &= e^{i\phi}  a^\dagger\\
 D(\delta)  a^\dagger  D^\dagger (\delta) &=  a^\dagger -\conj \delta I\\
 S(r e^{i \phi})  a^\dag  S^\dagger(r e^{i\phi}) &= (\cosh r)  a^\dagger - e^{- i \phi}(\sinh r)  a.
\end{split}
\end{align}
It turns out any single-mode Gaussian unitary operator can be described as a product $ S (\xi) D (\delta) R (\theta)$ with suitable parameters.

\paragraph{Position and momentum operators:} Another convenient operator basis on one mode is obtained using the position $X$ and momentum $P$ operators. Unlike the creation and annihilation operators, they are (essentially) self-adjoint. These operators are related according to (with the convention $\hbar=1$):
\begin{eqnarray}
     a &= \frac{1}{\sqrt{2}} (X + iP), \quad X = \frac{1}{\sqrt{2}} ( a +  a^\dagger)\\
     a^\dagger &= \frac{1}{\sqrt{2}} (X - iP), \quad P = \frac{1}{\sqrt{2i}} ( a -  a^\dagger).
\end{eqnarray}

It is insightful to view the evolution of these operators under Gaussian operators:
\begin{align}
\begin{split}
R(\theta) X R(\theta)^\dagger &= \cos (\theta) X + \sin(\theta) P,\\
R(\theta) P R(\theta)^\dagger &= -\sin (\theta) X + \cos(\theta) P,\\
D(\delta) P D^\dagger (\delta) &= P -\sqrt {2}Re (\delta),\\
D(\delta) X D^\dagger (\delta) &= X -\sqrt {2}\Im (\delta),\\
S(r e^{i \phi}) X S^\dagger(r e^{i \phi}) &= (\cosh r - \cos \phi \sinh r) X + \sinh r \sin \phi P,\\
S(r e^{i \phi}) P S^\dagger(r e^{i \phi}) &= -(\cosh r + \cos \phi \sinh r) P + \sinh r \sin \phi X.
\label{eq:affine}
\end{split}
\end{align}
From these equations, we can conclude that conjugation by a Gaussian operator is a (symplectic) affine map, implying that the degree remains unchanged by conjugating any polynomial in $X$ and $P$ by any Gaussian operator. 

\paragraph{Multiple modes:} A passive linear operator $ U$ is the following linear transformation of creation and annihilation operators
\begin{equation}
     U  a_i^\dagger  U^\dagger = \sum_j  U_{ij} a^\dagger_j.
\end{equation}
These operators do not change the total number of bosons in the system. An important example of a passive linear optical element is the $50:50$ beam splitter, where
\begin{align}\label{eq:beam-splitter}
U = \frac1{\sqrt2}\begin{pmatrix}
    1 & 1\\
    -1 & 1
\end{pmatrix}.
\end{align}

Gaussian unitary operators can be decomposed as follows:

\begin{theorem} [Euler--Bloch--Messiah decomposition]
\label{thm:EBMdecomp}
Any multimode Gaussian can be decomposed as
\begin{equation}
     G =  V \bigotimes_{i=1}^m  G_i  U,
\end{equation}
where $ U$ and $ V$ are passive linear operators and $ G_i$ is a single-mode Gaussian unitary operator acting mode $i$. 
\end{theorem}

Let $\ket\psi$ be a quantum state with finite stellar rank $r$. We say $\ket{c}$ is a core state of $\ket \psi$ if 
\begin{equation}
    \ket \psi= G\ket c,
\end{equation}
where $ G$ is a Gaussian operator of the form $ G =  V \bigotimes_{i=1}^m  G_i$, and,
\begin{equation}
    F_c^\star(z_1,\cdots, z_m) = P_r (z_1, \cdots, z_m),
\end{equation}
for some polynomial $P_r$ of degree $r$.

More generally, a \textit{core state} is defined as a quantum state with polynomial stellar function. These are the states with bounded support over the Fock basis.

\begin{theorem}[Core states \cite{chabaud2020stellar}]
The core state of a stellar function is unique up to a passive linear operator and has the same stellar rank as the quantum state. Conversely, if $\ket c$ is a core state of stellar rank $r$, for any Gaussian operator $ G$, $ G \ket c$ has rank $r$.
\label{thm:core}
\end{theorem}

\subsubsection{Phase space formulation and symplectic representation}\label{sec:phase-space-formalism}
Corresponding to any Gaussian state, is a Gaussian Wigner function
\begin{align}
W_\psi(\mathbf q) = \frac{1}{\mathcal N} \exp(-\frac12 \mathbf q^T \boldsymbol\Sigma^{-1} \mathbf q + \boldsymbol\mu^T \mathbf q),
\end{align}
where $\mathbf q:=\begin{pmatrix}
\Bf x\\
\Bf p
\end{pmatrix}\in\mathbb R^{2m}$ contains our variables of the phase space, for an $m$-mode Gaussian state $\psi$. Here $\mathcal N$ is the normalisation factor ensuring $\int \mathrm d^{2m} \mathbf q \, W_\psi(\mathbf q) = 1$, and therefore, $W_\psi$ defines a classical Gaussian distribution on the phase space. We recall that the parameters $\boldsymbol\Sigma$ and $\boldsymbol\mu$ can be also understood via the following
\begin{align}
\boldsymbol\mu = \tr\left[ \ket\psi\bra\psi \mathbf q \right], \quad \boldsymbol\Sigma = \mathrm{tr}\left[ (\mathbf q - \boldsymbol\mu) \ket\psi\bra\psi (\mathbf q - \boldsymbol\mu)^T \right].
\end{align}
Note that $\boldsymbol\mu\in\mathbb R^{2m}$ is a real vector of size $2m$, and similarly, we have $\boldsymbol\Sigma\in\mathbb R^{2m\times 2m}$ is a $2m\times 2m$ real symmetric matrix, and we employ the notation $\boldsymbol\mu = \begin{pmatrix}
\boldsymbol\mu^{(x)}\\
\boldsymbol\mu^{(p)}
\end{pmatrix}$
and $\boldsymbol\Sigma = \begin{pmatrix}
\boldsymbol\Sigma^{(xx)} & \boldsymbol\Sigma^{(xp)}\\
\boldsymbol\Sigma^{(xp)} {}^T & \boldsymbol\Sigma^{(pp)}
\end{pmatrix}$, where $\boldsymbol\mu^{(x)}, \boldsymbol\mu^{(p)}$ are vectors of size $m$, and submatrices of $\boldsymbol\Sigma$ are each of size $m\times m$. Upon measuring a mode in the position basis, which is mathematically described by POVM elements $\ket{x}\bra{x}$ for $x\in\mathbb R$, we get samples from a Gaussian distribution with mean $\mu^{(x)}_i$, and variance $\Sigma^{(xx)}_{ii}$. As an example, the vacuum state has $\boldsymbol\Sigma_{\ket0} = \frac12 \mathbb I$ and $\mu_{\ket0} = 0$. We refer the reader to \cite{serafini2017quantum,martin2022quantum} for a detailed explanation of the features of the phase-space representation.

We can express quadratic Hamiltonians in a short form, by employing the following convention
\begin{align}
\widehat r := \begin{pmatrix}
X_1&
\cdots&
X_m&
P_1&
\cdots&
P_m
\end{pmatrix}^T,
\end{align}
and we also use $\mathbf d. \widehat r := \sum_{i=1}^{2m} d^{(x)}_{i} X_i + d^{(p)}_{i} P_i$ for $\mathbf d\in\mathbb R^{2m}$, and $\widehat r^T \mathbf M \widehat r:= \sum_{i,j} M_{ij}^{(xx)} X_i X_j + M_{ij}^{(pp)} P_i P_j + M_{ij}^{(xp)} X_iP_j$, for $\mathbf M\in\mathbb R^{2m\times 2m}$. Note that we can always pick $\mathbf M$ to be a real symmetric matrix.

Let $U = \exp(-iH_G t)$ for some Gaussian Hamiltonian $H_G$ with the quadratic part $\Bf M$ and linear part $\Bf d$. Applying this unitary on a Gaussian state determined by $(\boldsymbol\Sigma,\boldsymbol\mu)$ provides a Gaussian with parameters $(\boldsymbol\Sigma',
\boldsymbol\mu')$ such that
\begin{align}
\boldsymbol\Sigma' = \exp(\boldsymbol\Omega\Bf Mt) \boldsymbol\Sigma \exp(\boldsymbol\Omega\Bf Mt)^T,
\end{align}
where
\begin{align}
\boldsymbol\Omega = \begin{pmatrix}
0 & -\mathbb I_m\\
\mathbb I_m & 0
\end{pmatrix},
\end{align}
with $\mathbb I_m$ being the $m\times m$ identity matrix. Recall that symplectic matrices are defined as
\begin{align}
\mathrm{Sp}_{2n}(\mathbb R):= \{\mathbf S\in\mathbb R^{2n\times 2n}: \mathbf S^T \boldsymbol\Omega\mathbf S = \boldsymbol\Omega \}.
\end{align}
Note that $\exp(\boldsymbol\Omega \mathbf M) \in \mathrm{Sp}_{2n}(\mathbb R)$ for any symmetric $\mathbf M$ as the Lie algebra corresponding to $\mathrm{Sp}_{2n}(\mathbb R)$ can be readily shown to be the set of matrices $\mathbf A\in\mathbb R^{2n\times2n}$ such that $\mathbf A^T \boldsymbol\Omega + \boldsymbol\Omega \mathbf A = 0$. Note that $\mathbf A = \boldsymbol\Omega\mathbf M$ satisfies the latter relation as $\mathbf M$ is symmetric.

\subsubsection{Elementary calculations involving canonical operators}\label{sec:elementary_calculations}

We use the Heisenberg picture to understand the evolution of operators under the influence of unitary transformations. The mathematical tools for computing these evolutions are well-known results from the mathematical development of Lie algebras and Lie groups. We begin with the hallmark result of Baker--Campbell--Hausdorff-Dynkin (usually abbreviated as BCH). 
\begin{theorem}[Baker--Campbell--Hausdorff--Dynkin]\label{thm:BCH}
    Let $A$ and $B$ be (possibly non-commuting) elements of an associative algebra $\mathcal A$ over a field of characteristic $0$ with a multiplicative unit element $I$. Then the algebra element $H:=\log{e^Ae^B}$ is given by a Lie series in the operators $A$ and $B$ of the form,
    \begin{equation}
        H = A + B + \frac{1}{2}[A,B] + \frac{1}{12}([A,[A,B]] - [B,[A,B]]) + h.o.,
    \end{equation} 
    where $h.o.$ represents higher-order nested commutators in $A$ and $B$.
\end{theorem}
Let $G$ be a matrix Lie group and $g$ its Lie algebra. Define the \textbf{adjoint map} as the linear operator for fixed $A \in g$ by $ad_A B = [A,B]=AB-BA$. The corresponding linear transformation $Ad_U$ for fixed $U \in G$ is given by $Ad_U B = UBU^{-1}$. Their relationship can be summarized by the following identity \cite{hall2015lie},
\begin{equation}
    Ad_{e^A} (B) = e^{ad_A} (B).
\end{equation}
Let $A$ and $B$ be operators as above, define the iterated commutator as,
\begin{equation} 
    [(A)^n,B] := 
    \begin{cases}
    [A,\cdots [A, B]] \text{ where $n$ is the depth of the nested commutator,} \\
    B \text{ when $n=0$.}
    \end{cases}
\end{equation} 
Then we have the following especially useful form of the BCH formula,
\begin{lemma}[Complex BCH formula]\label{lem:complexBCH}
    \begin{equation}
        e^{itA} B e^{-itA} = \sum_{n=0}^\infty \frac{(it)^n}{n!}[(A)^n, B].
    \end{equation}
\end{lemma}
\begin{proof}
    For $t \in \mathbb R$ define the operator-valued function $f(t):=e^{tA}Be^{-tA}$. Then $f(0)=B$ and taking derivatives with respect to the variable $t$ we get that,
    \begin{align*}
        f'(0) &= (Ae^{tA}Be^{-tA} - e^{tA}BAe^{-tA}) \vert_{t=0} 
        = [A,B], \\
        f''(0) &= (e^{tA}A^2Be^{-tA} - e^{tA}ABAe^{-tA} - e^{tA}ABAe^{-tA} + e^{tA}BA^2e^{-tA}) \vert_{t=0}
        = [A,[A,B]] = [(A)^2,B],
    \end{align*}
    and so on. The Taylor expansion of $f$ is then given by,
    \begin{equation}
        f(t) = \sum_{n=0}^\infty \frac{f^{(n)}(0) t^n}{n!}.
    \end{equation}
    Substituting $t=1$ gives us the (differential form of) the real BCH formula, while using $t=i$ gives us the the complex formula above.
\end{proof}
Another useful tool is the following dual form of the BCH formula, attributed to an unpublished calculation of Zassenhaus \cite{magnus1954exponential}.
\begin{lemma}[Zassenhaus formula]
Let $\mathcal L(A,B)$ be the free algebra generated by possibly non-commuting operators $A, B$. Then $e^{A+B}$ can be uniquely decomposed as,
\begin{equation}
e^{A+B} = e^A e^B \prod_{n=2}^\infty e^{C_n(A,B)},    
\end{equation}
where $C_n(A,B) \in \mathcal L(A,B)$ are homogeneous Lie polynomials in $A$ and $B$ of total degree $n$. Moreover for $t \in \mathbb R$ the first few terms are given by,
\begin{equation}
e^{t (A + B)} = e^{tA} e^{tB} e^{- t^2/2 [A,B]} e^{\frac{t^3}{6} (2 [B, [A,B]] + [A, [A,B]])} \times {\text{h.o.}},
\end{equation}
where h.o.\ are higher order commutator exponentials.  \label{lem:zassenhaus}
\end{lemma}
\begin{proof}(Sketch.)
     By the BCH formula above, it follows that $e^{-A}e^{A+B}=e^{B+D}$ where $D$ is a Lie polynomial of degree $>1$. Similarly, $e^{-B}e^{B+D} = e^{C_2 + D'}$ where $D'$ is of degree $>2$. Recursive applications of the BCH formula give, $e^{-tA}e^{t(A+B)} = \prod_n e^{t^n C_n}$.
\end{proof}

The evolutions induced on the variable $X$ by a single application of unitaries from the standard gate set $\set{e^{it_1X}, e^{it_2X^2}, e^{i\frac{\pi}{4}(X^2 + P^2)}, e^{it_3X^3}}$, where $t_1,t_2,t_3 \in \mathbb R$, can be seen to be very simple except for the Fourier transform $e^{i\frac{\pi}{4}(X^2 + P^2)}$. In short $X$ commutes with each of $X, X^2$ and $X^3$, so $[X,X]=[X^2,X]=[X^3,X]=0$ and similarly for all high-order commutators. Thus for $k\in\set{1,2,3}$,
\begin{equation}
    X \xmapsto{(t_k, X^k)} X.
\end{equation}
However for the Fourier transform, 
\begin{equation}
[X^2+P^2,X] = [X^2,X] + [P^2,X] = -2iP,
\end{equation}
where we have used the bilinearity and Leibniz rule\footnote{This is a generalized form of the product rule from calculus, 
\(D(xy)=D(x)y + xD(y)\). A linear map on an algebra of operators that satisfies the Leibniz rule is called an \textbf{algebraic derivation}. Here, the function that acts as a derivation is the adjoint map defined by $ad_X:=[X,\cdot]$ and $ad_P:=[P,\cdot]$.} for the commutator. Likewise,
\begin{equation}
[X^2+P^2,P] = 2iX,
\end{equation}
so looking at the Taylor expansion and simplifying,
\begin{align*}
    e^{i\frac{\pi}{4}(X^2+P^2)} X e^{-i\frac{\pi}{4}(X^2+P^2)} &= X \left[1 - \frac{1}{2!}\left(\frac{\pi}{2}\right)^2 + \frac{1}{4!}\left(\frac{\pi}{2}\right)^4 - \cdots \right] + P\left[ \left(\frac{\pi}{2}\right) - \frac{1}{3!}\left(\frac{\pi}{2}\right)^3 + \cdots \right] \\
    &= X\cos(\pi/2) +  P\sin(\pi/2) \\
    &= P.
\end{align*}
The overall transformation is given by
\begin{equation}
    X \xmapsto{(\pi/4, X^2 + P^2)} P.
\end{equation}
So we can conclude that even if we only wish to track the evolution of $X$ we must also track $P$ since their coefficients can effectively be swapped by the Fourier transform. The computation of the commutators of $P$ is straightforward and gives the following transformations, where in each case the pair of terms $(t,O)$ above the arrow indicate the parameter and unitary,
\begin{align*}
    P &\xmapsto{(t_1, X)} P - t_1 I,\\
    P &\xmapsto{(t_2, X^2)} P - 2t_2 X,\\
    P &\xmapsto{(\pi/4, X^2 + P^2)} -X,\\
    P &\xmapsto{(t_3, X^3)} P - 3t_3 X^2.
\end{align*}
As expected, the only non-linear transformation of the position and momentum variables occurs due to the cubic phase gate $X^3$.

From the above calculations note that we can implement a squaring operation in the conjugate variables using just two steps,
\begin{equation}
    X \xmapsto{(X^2+P^2)} P \xmapsto{(c, X^3)} P - 3c X^2,
\end{equation}    
where we set $t_3=c$. Applying the same operations again,
\begin{align}
    P - 3c X^2 \xmapsto{(X^2+P^2)} -X - 3c P^2 &\xmapsto{(c, X^3)} -X - 3c(P-3cX^2)^2  \\
    &= -X - 3cP^2 - (3c)^3X^4 - (3c)^2(PX^2 + X^2P)\\
    &= X (-1+6ci) - (3c)^3X^4 - 3cP^2 + 2(3c)^2 X^2P.
\end{align}

If this operation is applied polynomially many times, we can get terms of single-exponential degree in $X$ and $P$ and coefficients with doubly-exponential magnitude, which is possible when the quantity $c$ has exponential size. This can be achieved using polynomially many Gaussian gates followed by polynomially many `squaring' rounds.

There are a few more commutators of interest, since we want to know how the standard gate set affects higher-order terms. To begin with, we know that $X$ and $P$ act as derivations on the algebra generated by the $X$ and $P$ variables. Thus for fixed $k \in \mathbb N$,
\begin{align}
    [P,X^k] &= k X^{k-1} [P, X] = - ik X^{k-1}, \\
    [X,P^k] &= k P^{k-1} [X, P] = ik P^{k-1}.
\end{align}
We can generalize this to their action on arbitrary monomials of the form $X^mP^n$ by the Leibniz rule,
\begin{align}
    [P, X^mP^n] &= [P,X^m]P^n + X^m[P,P^n] = -imX^{m-1}P^n,\\
    [X, X^mP^n] &= [X,X^m]P^n + X^m[X,P^n] = imX^{m}P^{n-1}.
\end{align}

In the following, we give a simple proof that the action of the $X^n$ gate on a term $P^k$ of arbitrarily high degree, produces only $n$ terms, each of degree $\leq k$. We are interested in the cases $n\in\{2,3\}$. For a more general formula see \cref{app:commutation}. 
\begin{lemma}\label{lemma:CommutatorsInNormalForm} \hfill\break
    \begin{enumerate}[label=(\roman*)]
        \item For fixed $k\geq 2$:
        \begin{equation}
            [X^2, P^k] = i2k XP^{k-1} + k(k-1) P^{k-2}.
        \end{equation}
        \item For fixed $k\geq 3$:
        \begin{equation}
            [X^3, P^k] = i3k XP^{k-1} + k(k-1) P^{k-2}.
        \end{equation}
    \end{enumerate}
\end{lemma}
\begin{proof} \hfill\break
    \begin{enumerate}[label=(\roman*)]
        \item We begin by decomposing the commutator as a sum of two terms,
        \begin{align}
            [X^2,P^k] &= [X^2,P]P^{k-1} + P[X^2,P^{k-1}] \label{eqn:1.1} \\
            &= [X^2,P]P^{k-1} + P\left([X^2,P]P^{k-2} + P[X^2,P^{k-2}]\right) \label{eqn:1.2} \\
            &= 2i XP^{k-1} + \sum_{j=1}^{k-1} P^{j}[X^2,P] P^{k-j-1} \\
            &= 2i XP^{k-1} + 2i \sum_{j=1}^{k-1} P^{j}XP^{k-j-1} \\
            &= 2i XP^{k-1} + 2i\sum_{j=1}^{k-1} \left(XP^{k-1} + [P^j,X]P^{k-j-1}\right)  \label{eqn:1.4}\\
            &= 2ki XP^{k-1} + k(k-1) P^{k-2}.
        \end{align}
        The second term of Eq.~\eqref{eqn:1.1} is decomposed in the same way, and converted into a telescopic sum. In line Eq.~\eqref{eqn:1.4}, we use the commutator relation $P^jX=[P^j,X]+XP^j$. The result follows in the last line by computing elementary commutators and collecting terms.

        \item We could use an analogous method to part (i), but it is easier to use the previous result directly.
        \begin{align}
            [X^3,P^k] &= X^2[X,P^k] + [X^2,P^k]X \\
            &= X^2 \left(ikP^{k-1} \right) + \left(2ki XP^{k-1} + k(k-1) P^{k-2} \right) X \\
            &= (ki) X^2P^{k-1} + (2ki) XP^{k-1}X + k(k-1)P^{k-2}X \\
            &= (3ki) XP^{k-1} + (2ki) \left([P^{k-1},X] + XP^{k-1} \right).
        \end{align}
    \end{enumerate}
\end{proof}

\noindent One can generalize such commutation relations as provided in \cref{lem:app-commutation}.

In some of the calculations, notice that we have implicitly been putting putting operator polynomials in a standard form, which we will now treat more formally. In the literature of Quantum Field Theory, an operator $ O$ which is an arbitrary product of creation and annihilation operators, is said to be in its normal ordered form, commonly written as $: O:$ or $\mathcal N( O)$, when all creation operators appear to the left of all annihilation operators. Note that this notion of normal ordering only makes sense for products of operators, but it is not a linear functional on the space of operators, so it  doesn't make sense to talk about the normal ordering of a sum of operators. 

We can however define a similar, linearizable notion for polynomials in the operators $X$ and $P$.
\begin{definition}[Normal form]\label{def:normalForm}
    A Hamiltonian operator $H$ which is a polynomial of degree $d$ in the operators $\set{X,P}$ is said to be in \textbf{normal form} if it is written as a degree-ordered sum of monomials $X^mP^n$,
    \begin{equation}
    H=\sum_{m+n\leq d} h_{m,n} X^m P^n.
    \end{equation}
\end{definition}

It would be convenient if we could show that all the applicable transformations of the position and momentum operators under the standard gate set described before can be put in normal form. This would allow us to more easily track the evolution of the operators under the action of the gates. Recall that the binomial theorem gives the coefficients of any power of a sum for any formal variables $X$ and $Y$ that commute,
\begin{equation}
    (X+Y)^n = \sum_{k=0}^n \binom{n}{k} X^k Y^{n-k}.
\end{equation}
The difficulty is that the position and momentum operator algebra is inherently non-commutative with the defining relation $[X,P]=iI$. However it turns out that a generalization of the binomial theorem for non-commutative operators can be derived from either the recurrence relationships between repeated commutators \cite{wyss2017two}, or by comparing coefficients in the Zassenhaus formula \cite{zaimi2011binomial}. We describe Wyss' approach below.

Define the $n$-th binomial operator $B_n(X,Y)$ for $X,Y \in \mathcal A$ and $n \in \set{1,2,3,\cdots}$ by,
\begin{equation}
  B_n(X,Y) := \sum_{k=0}^n \binom{n}{k} X^k Y^{n-k},  
\end{equation}
where $\mathcal A$ is an associative algebra with unit $I$. Then for all choices of $X,Y \in \mathcal A$, possibly non-commuting, we can show that $B_n$ has the following properties,
\begin{lemma}[\cite{wyss2017two}] \hfill
    \begin{enumerate}
        \item \( B_0=I,\quad B_1=X+Y, \)
        \item \( B_{n+1} = B_1B_n - [Y, B_n], \)
        \item \( B_{1}^n = B_n + \sum_{k=0}^{n-2} B_1^k [Y, B_{n-1-k}].\)
    \end{enumerate}
\end{lemma}

Using $B_1 = X+Y$ leads to:

\begin{theorem}[\cite{wyss2017two}]\label{thm:WyssBinomialThm}
    
    \begin{equation}
        (X+Y)^n = B_n(X,Y) + \sum_{k=0}^{n-2}(X+Y)^k ad_Y (B_{n-1-k}(X,Y)).
    \end{equation}
Here for an operator $C$ we defined $ad_Y(C) := [Y, C]$. 
\end{theorem}

Let $C:=[Y,X]=[Y,B_1(X,Y)]$. If we also know that $[X,C]=[Y,C]=0$ then we can improve the above result as per \cite[Section 7]{wyss2017two} to the closed form,
\begin{corollary}[\cite{wyss2017two}]\label{cor:closedFormBinomialExpansion}
    \begin{equation}
        (X+Y)^n = \sum_{k=0}^{\lfloor \frac{n}{2} \rfloor} \frac{n!}{(n-2k)!k!2^k} B_{n-2k}(X,Y) [Y, X]^k,
    \end{equation}
    where $[X,C]=[Y,C]=0$.
\end{corollary}

Hereafter, the continuous variables we are concerned with will be those associated with the unbounded conjugate operators $X$ and $P$, defined in \cref{sec:CVprocesses}.
Throughout this work, we restrict to a dense subspace of the Hilbert space known as the Schwartz space, on which exponentials of many polynomials in the position and momentum operators are well-defined through the spectral theorem \cite{schwartz1947theorie,hall2013quantum,schmudgen2020invitation}. We still denote that space as $\mathcal H$ for simplicity, and we also denote by by $\mathcal P$ the set of essentially self-adjoint operators over Schwartz space that are polynomials in the position and momentum bosonic operators, by $\mathcal P_k$ its subset of $k$-local operators, and by $\mathcal P_{k,d}$ its subset of $k$-local operators of degree $d$.

With this preliminary material introduced, we first study in \cref{sec:GDC} the complexity of bosonic computations involving only quadratic Hamiltonians, i.e., Gaussian computations, and we turn to the case of general polynomial Hamiltonians in \cref{sec:CVBQP}. Finally, we study bosonic ground state problems from a complexity0-theoretic point of view in \cref{sec:gsp}.

\section{The power of Gaussian bosonic computations}
\label{sec:GDC}

It is well-know that simulating Gaussian bosonic dynamics can be performed in polynomial time \cite{bartlett2002efficient}, which is a continuous-variable analog of Gottesman--Knill's theorem for simulating Clifford circuits. 
Moreover, recent work shows that simulating a specific class of Gaussian computations, which involves highly non-local beam-splitter gates together with input coherent states, over exponentially many modes is \textbf{BQP}-complete \cite{barthe2024gate}.
However, it remains unclear what the exact power of general Gaussian computations over polynomially many modes is.

In this section, we define a complexity class of \textit{Gaussian dynamical computations} based on problems that can be solved using Gaussian states, Gaussian evolution, and Gaussian measurements, with appropriately bounded energy, and show that it is equal to $\mathbf{BQL}$ (\cref{thm:GDC=BQL}). Recall that $\mathbf{BQL}$ refers to bounded-error quantum logspace computations, which are problems decidable with a logspace uniform family of quantum circuits acting on logarithmically many qubits.

\subsection{Gaussian dynamical computations}
\label{sec:GDCdef}

In what follows, to express quadratic Hamiltonians in a short form, we employ the following convention
\begin{align}
\widehat r := \begin{pmatrix}
X_1&
\cdots&
X_m&
P_1&
\cdots&
P_m
\end{pmatrix}^T,
\end{align}
and we also use $\mathbf d. \widehat r := \sum_{i=1}^{2m} d^{(x)}_{i} X_i + d^{(p)}_{i} P_i$ for $\mathbf d\in\mathbb R^{2m}$, and $\widehat r^T \mathbf M \widehat r:= \sum_{i,j} M_{ij}^{(xx)} X_i X_j + M_{ij}^{(pp)} P_i P_j + M_{ij}^{(xp)} X_iP_j$, for $\mathbf M\in\mathbb R^{2m\times 2m}$. Note that we can always pick $\mathbf M$ to be a real symmetric matrix. We refer to $g=(\Bf M,\Bf d, t)$ as a `Gaussian gate,' which represents the Gaussian unitary $\exp(-it(\widehat{r}^T \Bf M\widehat{r}+\Bf d\cdot \widehat{r}))$. In what follows, we introduce a complexity class based on Gaussian bosonic evolutions, which we name \textit{Gaussian dynamical computation} ($\mathbf{GDC}$). First, in the following definition, we fix some language to clarify our definition of $\mathbf{GDC}$.

\begin{definition}\label{def:Gaussian-circuit}
A `Gaussian circuit' $\mathcal C$ is determined by a set of $m\in\mathbb N$ ordered triples $\mathcal C:=(\Bf M_i, \Bf d_i, t_i)_{i=1}^m$, where each $(\mathbf M_i,\Bf d_i)$ defines a Hamiltonian according to $\widehat{r}^T \Bf M_i \widehat{r}+\Bf d_i\cdot \widehat{r}$ over $n\in\mathbb N$ modes applied for time $t_i>0$. We say that the `total evolution time' is $\tau = \sum_{i=1}^n t_i$, and `the output distribution' is the distribution obtained by measuring the first mode in the position basis. Formally, we have that for the solution $\psi$ to the following
\begin{align}
\frac{\mathrm d}{\mathrm dt}\ket{\psi(t)} = -iH_j\ket{\psi(t)}, \quad \text{for } t\in\left[\sum_{k<j} t_k, \sum_{k\leq j} t_k\right], \quad \forall j\in[n]
\end{align}
initialized at vacuum, $\ket{\psi(0)} =\ket{0^m}$, the output distribution $\mathcal P$ is determined by the probability density function $\left|(\langle x|\otimes \mathbb I^{n-1})|\psi(\tau)\rangle\right|^2$ for $x\in\mathbb R$. We say $Z\sim\mathcal P$ is a `sample' from $\mathcal C$. Moreover, we say that the `energy bound' of the evolution is 
\begin{align}\label{eq:energybounddef}
E^\ast = \sup_{t\in[0,\tau]} \langle \psi(t) | \sum_{i=1}^n N_i | \psi(t)\rangle.
\end{align}
\end{definition}

A schematic of a Gaussian circuit is provided in \cref{fig:Gaussian-circuit}. Note that the output distribution $\mathcal P$ is always a Gaussian (this can be readily seen from the Wigner representation). Note that the choice to measure the first mode is arbitrary as the swap gate is a passive Gaussian unitary (i.e., one can change the definition to allow for measuring an arbitrary mode without changing the results we obtain in this section). We are now ready to define a complexity class based on Gaussian computations.

\begin{definition}\label{def:GBC}
A language $L\subseteq\{0,1\}^*$ is in $\mathbf{GDC}$ (Gaussian Dynamical Computation) if there exists a log-space Turing machine $M$ that on input $x$, prints out a description of a Gaussian circuit $(\Bf M_i, \Bf d_i, t_i)_{i=1}^m$ defined over $n=O(\mathsf{poly}(|x|))$ modes with a total number of $m=O(\mathsf{poly}(|x|))$ evolutions, where each entry of $\Bf M_i, \Bf d_i$ is bounded by $1$ and $\sum_{i=1}^n t_i = O(\mathsf{poly}(|x|))$, the Gaussian circuit has an energy bound $E^\ast = O( \mathsf{poly}(|x|))$, and considering $\mathcal P$ as the output distribution of the first mode of the circuit in position basis we have
\begin{itemize}
    \item (YES case): $x\in L$ if $\underset{Z\sim \mathcal P}{\mathbb P}[Z\geq b]\geq \frac23$,
    \item (NO case): $x\notin L$ if $\underset{Z\sim \mathcal P}{\mathbb P}[Z\leq a] \geq \frac23$,
\end{itemize}
for given $a,b$ with the promise $b-a\geq\Omega(\frac{1}{\mathsf{poly(|x|)}})$. 
\end{definition}

We also consider a canonical problem related to simulation of Gaussian dynamics defined below.

\begin{definition}[Gaussian simulation]\label{def:Gaussian-simulation}
Consider a Gaussian circuit $(\Bf M_i,\Bf d_i, t_i)_{i=1}^m$ over $n$ modes with the promise that entries of $\Bf M_i,\Bf d_i$ are upper bounded by $1$, total evolution time is $\tau = O(\mathsf{poly}(n,m))$, and $E^\ast = O( \mathsf{poly}(n,m))$. The problem of $\mathsf{GausSim}$ is to decide between the following cases
\begin{itemize}
    \item (YES case): $\underset{Z\sim \mathcal P}{\mathbb E}[Z] \geq b$
    \item (NO case): $\underset{Z\sim \mathcal P}{\mathbb E}[Z] \leq a$
\end{itemize}
given $a,b$ with the promise $b-a\geq\frac{1}{\mathsf{poly}(n,m)}$. 
\end{definition}

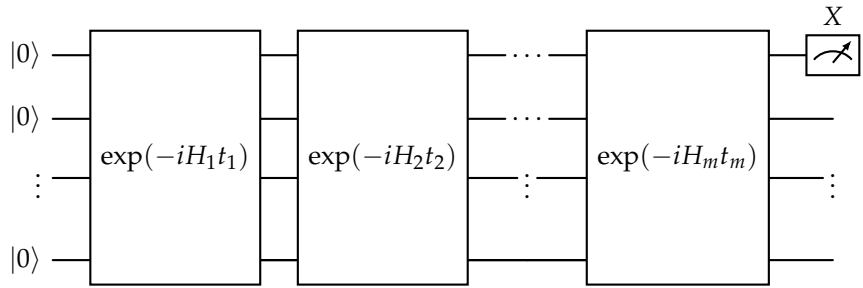
\begin{figure}[h]
\centering
\begin{quantikz}
    \lstick{$\ket{0}$} & \gate[wires=4]{\exp(-i H_1 t_1)} & \gate[wires=4]{\exp(-i H_2 t_2)} & \ \dots\ & \gate[wires=4]{\exp(-i H_m t_m)} & \meter{X}\\
    \lstick{$\ket{0}$} & & & \ \dots \ & &  \\
    \lstick{$\vdots$} &  &\ \vdots \ & \ \vdots\ & &\ \vdots \ \\
    \lstick{$\ket{0}$} & & \qw & \qw & \qw &
\end{quantikz}
\caption{A schematic of an instance of evolution for the problem of Gaussian evolution described in \cref{def:Gaussian-simulation}. We allow multiple Hamiltonians, and measurements in the end, starting from the vacuum state $\ket{0^m}$. Note that we require the energy of the computation to be bounded at all times by $E^\ast = O(\mathsf{poly}(m))$, and that the total evolution time is bounded by a polynomial, while Hamiltonian coefficients are held constant.}
\label{fig:Gaussian-circuit}
\end{figure}

Note that by Bloch-Messiah decomposition of pure states, a dynamics generated by a sequence of Hamiltonians $H_1, H_2, \cdots, H_n$ with $\ket{0^m}$ has an equivalent dynamics generated by a constant number of Hamiltonians. Nonetheless, we allow polynomially many Hamiltonians. 

Let us first introduce a gadget which will be very useful in our technical analysis. The approach shows how we can amplify the gap $b-a$ (both for problems in $\mathbf{GDC}$ and any instance of $\mathsf{GausSim}$) by any polynomial, at the cost of a polynomial overhead in the number of modes, gates, and energy bound.

\begin{proposition}\label{prop:sample-mean}
Let $(\mathcal C^{(k)})_{k=1}^{2^r}$ be $2^r$ Gaussian circuits. Let $Z_1,\cdots,Z_{2^r}$ be samples from $(C^{(k)})_{k=1}^{2^r}$. We can put together these circuits and combine them with $2^{r}-1$ many two-mode Gaussian gates that do not change the energy of the state, such that the combined Gaussian circuit computes the (scaled) sample mean
\begin{align}\label{eq:average-Gaussian}
\overline{Z} = \frac{1}{2^{r/2}}(Z_1+\cdots+Z_{2^r}).
\end{align}
In short, we can compute the sample mean of Gaussian circuits (up to an overall scaling) without increasing energy.
\end{proposition}
\begin{proof}
The proof relies on using a 50:50 beam splitter \eqref{eq:beam-splitter}. Recall that the beam splitter does the transformation $X_1 \mapsto \frac{1}{\sqrt 2} (X_1+X_2)$. Note that if we have two independent states (i.e., the reduced density over the two is a product state) at the input of the beam splitter, each producing samples $Z_1$ and $Z_2$, then, upon interacting through the beam splitter, we get a sample that has the same distribution as $\frac{1}{\sqrt 2}(Z_1+Z_2)$. We can apply this method recursively to achieve \eqref{eq:average-Gaussian}. This is simply achieved by grouping circuits in pairs and apply a $50:50$ beam splitter between each pair in recursion. Please refer to \cref{fig:Gaussian-average} for a diagramatic explanation of the approach.
\end{proof}

\begin{figure}
\begin{quantikz}
\lstick{$\ket{0}$} & \gate[wires=3]{\mathcal C^{(1)}} & & \ctrl{3} & & \ctrl{6} & \meter{X} & \setwiretype{c} \ Z = \frac{Z_1+\cdots +Z_4}{\sqrt{4}}\\
\lstick{$\vdots$} & &\ \vdots \\
\lstick{$\ket{0}$} & & \\
\lstick{$\ket{0}$} & \gate[wires=3]{\mathcal C^{(2)}} & & \targX & & \\
\lstick{$\vdots$} & &\ \vdots \\
    \lstick{$\ket{0}$} & &\\
\lstick{$\ket{0}$} & \gate[wires=3]{\mathcal C^{(3)}} & & \ctrl{3} & & \targX & &\\
\lstick{$\vdots$} & &\ \vdots \\
\lstick{$\ket{0}$} & & \\
\lstick{$\ket{0}$} & \gate[wires=3]{\mathcal C^{(4)}} & & \targX & & \\
    \lstick{$\vdots$} & &\ \vdots \ \\
    \lstick{$\ket{0}$} & &
\end{quantikz}
\begin{quantikz}
& \ctrl{1} &\\
& \targX & &
\end{quantikz}
= 
\text{50:50 beam splitter}
\caption{An example of the circuit for computing sample mean of many Gaussian circuits. Here we have shown an example where there are $4$ Gaussian circuits $\mathcal C^{(1)}, \cdots \mathcal C^{(4)}$, however, as explained in the proof of \cref{prop:sample-mean} this approach can be extended to many circuits. We highlight that in case of having $2^r$ initial circuits, we require $2^{r}-1$ many beam splitters. Note that the circuit composes of beam splitters only, and therefore, leaves the energy invariant. Moreover, the 50:50 beam splitter is defined as in \eqref{eq:beam-splitter}, where the transformation in the Heisenberg picture follows $X_1\mapsto \frac{1}{\sqrt 2}(X_1 + X_2)$ and $X_2 \mapsto \frac{1}{\sqrt 2}(-X_1 + X_2)$. Therefore, we note that the average (up to scaling) is computed on the top mode.}
\label{fig:Gaussian-average}
\end{figure}
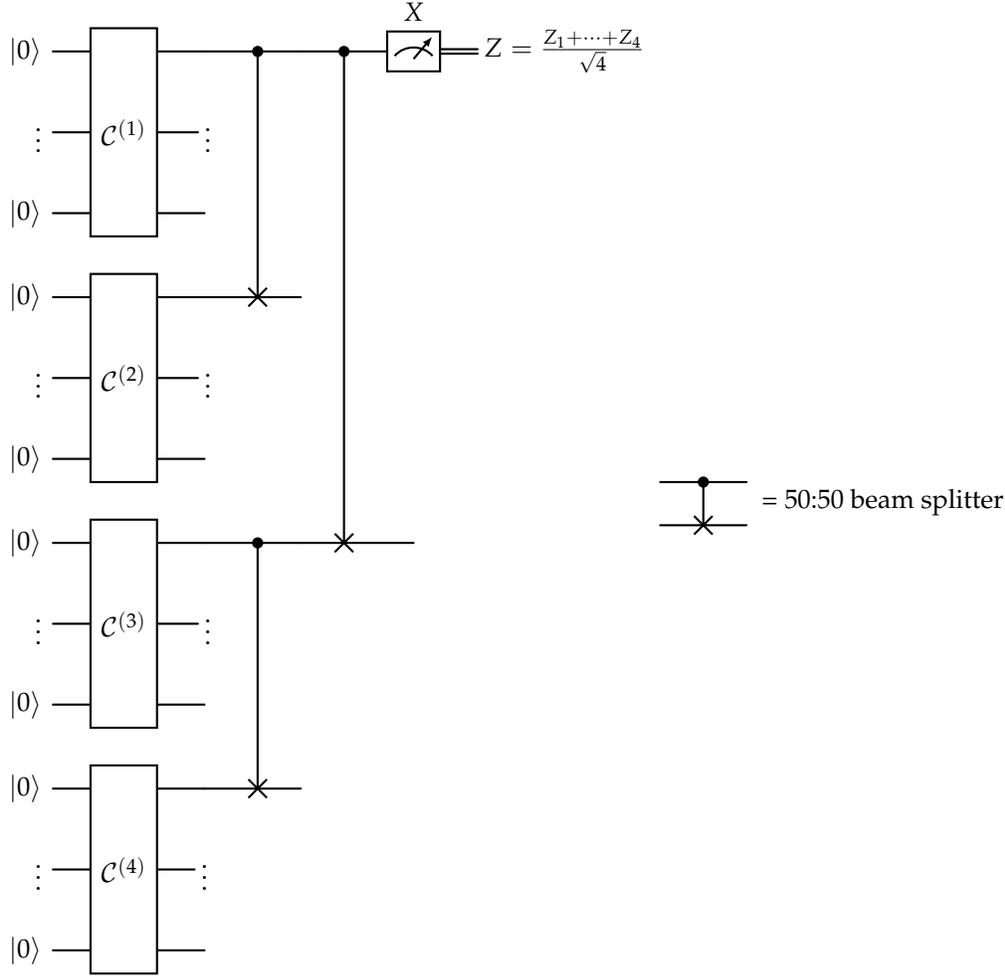

Since by \cref{prop:sample-mean} we can compute sample means in our $\mathbf{GDC}$ model, we can show an equivalence between $\mathbf{GDC}$ and the $\mathsf{GausSim}$ problem. This is formally established in the following proposition.

\begin{proposition}
$\mathsf{GausSim}$ is $\mathbf{GDC}$-complete.
\end{proposition}
\begin{proof}
Let us consider the two cases separately:
\begin{itemize}
    \item Solving $\mathsf{GausSim}$ in $\mathbf{GDC}$: Consider an instance of $\mathsf{GausSim}$ with parameters specified in \cref{def:Gaussian-simulation}. Note that we have
    \begin{align}
    \underset{Z\sim\mathcal P}{\mathrm{Var}}\left[Z\right] \leq \langle X_1^2 \rangle \leq \sum_{i=1}^n \langle X_i^2\rangle + \langle P_i^2\rangle = E^\ast.
    \end{align}
    Therefore, if we denote by $Z_1,\cdots,Z_{2^r}$ iid samples from $\mathcal C$ we have that we get samples from a distribution of mean $2^{r/2} \mu$ (with $\mu=\mathbb E[Z]$), but with the same variance i.e., at most $E^\ast$. Note that with our approach, we have amplified the gap $b-a$. Indeed, we can amplify it by any polynomial. Therefore, using a simple Hoeffding inequality, we obtain that collecting $N= 2^r = O(E^\ast \epsilon^{-2} \log(\delta^{-1}))$ many samples is sufficient for an $\epsilon$-accurate estimation of the mean with probability at least $1-\delta$. As explained in \cref{prop:sample-mean} this is can be simply done with a polynomial overhead (in terms of $n$ and $m$) in the number of modes and gates and the energy bound.
    \item Hardness of $\mathsf{GausSim}$ for $\mathbf{GDC}$: Note that if $Z$ follows a Gaussian distribution and $\mathbb P[Z\geq b]\geq \frac23$, then, we have $\mathbb E[Z] \geq b$ as median and mean coincide in Gaussian distributions. Therefore, for any problem $L$ in $\mathbf{GDC}$, the reduction is simply to consider its Gaussian circuit (with the corresponding paramters) generated by the logspace machine as an instance of $\mathsf{GausSim}$.
\end{itemize}
\end{proof}

The main result of this section is summarized in the following theorem.

\begin{theorem}\label{thm:GDC=BQL}
It is the case that
\begin{align}
\mathbf{GDC} = \mathbf{BQL}.
\end{align}
\end{theorem}

\paragraph{Discussing the definition} Here we discuss the reasons behind some of our assumptions in the definition of $\mathbf{GDC}$. Firstly, regarding the bounded energy assumption, from a practical point of view, it seems reasonable to consider a bound on the energy of the state under consideration. We cannot simply achieve arbitrarily high energies in a lab. For example, reaching high squeezing levels for a single mode with no displacement faces practical limits \cite{vahlbruch2016detection, meylahn2022squeezed}. Therefore, in our computational model, we put physical restrictions on top of common theoretical assumptions of the field (e.g., our bounded space uniformity). Please refer to \cref{sec:discussionE} for a thorough discussion about energy promises. Note that when considering the non-Gaussian computations later on, we would like to explore the limits of computability, but here, as we already know any Gaussian computation is simulable in $\mathbf{P}$, we have restricted ourselves  to the realistic case and we study its power. Moreover, we note that a model of Gaussian computation with no energy bound has more capabilities. As an example, one can apply the squeezing Hamiltonian $H=\frac12(a^2-a^\dagger{}^2)$, for time $t$ on a state displaced by the amount $\alpha\in\mathbb R$ in the $X$ direction, to obtain a state with mean $\alpha^t$ for a measurement in the $X$ direction. Letting $t$ grow polynomially large, we can obtain a number with polynomially many bits in output.
Lastly, we are using a logspace pre-processing as the simulation of a Gaussian dynamics can be performed in $\mathbf P$ (with ideas similar to \cite{bartlett2002efficient, chabaud2021classical}).

We turn to the proof of \cref{thm:GDC=BQL}.
Let $(\boldsymbol\Sigma_i,\boldsymbol\mu_i)$ denote the covariance matrix and mean of the overall state of $m$ modes right after the $i$-th evolution. We can track covariance matrices as
\begin{align}
\boldsymbol\Sigma_i = \frac12 \Bf S_i \Bf S_i^T,
\end{align}
where $\mathbf S_i = \exp(\boldsymbol\Omega \mathbf M_i t_i) \cdots \exp(\boldsymbol\Omega \mathbf M_1 t_1).$ 
With a slight abuse of notation, we use $\boldsymbol\mu(t)$ to denote the vector of means at time $t\in[0,\sum_{i=1}^n t_i]$. Going into the Heisenberg picture, we get that for any $t$ such that $\sum_{j=1}^{i} t_j <t< \sum_{j=1}^{i+1} t_j$ we have
\begin{align}\label{eq:gaussian-ode}
\frac{\mathrm d}{\mathrm dt} \boldsymbol\mu(t) = \boldsymbol\Omega \Bf M_i \boldsymbol\mu(t) + \mathbf d_i
\end{align}
as the system is being derived by the $i$-th Hamiltonian. The reason behind this equation, is that in the Heisenberg evolution picture, we have
\begin{align}
\frac{\mathrm d}{\mathrm dt} \langle \mathcal O \rangle_t = -i \langle [H,\mathcal O] \rangle_t,
\end{align}
and one simply verify that $ -i[H,\widehat{\mathbf r}] = \boldsymbol\Omega\mathbf M \widehat{\mathbf r} + \boldsymbol\Omega\mathbf d \cdot \mathbb I$, if $H$ is a Gaussian Hamiltonian\footnote{We use both the terms \textit{Gaussian} and \textit{quadratic} to refer to Hamiltonians of degree two.} and has $\mathbf M$ representing its quadratic part and $\mathbf d$ representing its linear part.

Note that a complete problem for $\mathbf{BQL}$ is the approximate inversion of well-conditioned matrices \cite{fefferman2021eliminating}. Indeed, we have the following:

\begin{definition}[Approximate inversion of well condition matrices]\label{def:app-matrix-inversion}
Let the problem of approximate inversion of well condition matrices (poly-conditioned-$\mathsf{MatInv}$), be to decide between
\begin{itemize}
    \item (YES case): $A^{-1}_{ij} > b$,
    \item (NO case): $A^{-1}_{ij} < a$,
\end{itemize}
given indices $(i,j)$ of matrix $A$. Note that, we are given $\mathbf A\in\mathbb R^{n\times n}$, and we have that the condition number of $\mathbf A$, denoted by $\kappa(\Bf A) := \frac{s_{\mathrm{max}}(\mathbf A)}{s_{\mathrm{min}}(\mathbf A)}$, is polynomially bounded i.e. $\kappa(\Bf A) = O(\mathsf{poly}(n))$, and that required the precision also satisfies $b-a = \Omega(\frac{1}{\mathsf{poly}(n)})$.
\end{definition}

Note that we use $s_{\min}$ and $s_{\max}$ to denote the smallest and largest singular value. As the problem defined above is complete for $\mathbf{BQL}$, we first need to give a logspace reduction from Gaussian computation to this problem. Then, we show that any such problem can be reduced to Gaussian simulation for hardness.

\subsection{Simulating Gaussian computations in $\mathbf{BQL}$}\label{sec:gdc-is-in-bql}
We first show that any Gaussian simulation problem can be turned into a matrix inversion problem.
Notice that the formulation in \eqref{eq:gaussian-ode} allows us to express the problem as a set of ODE simulation problems, and such problems can be handled by matrix inversion. Here, we consider the approach of \cite{krovi2023improved} as it turns out that the results therein can be immediately applied to our problem for bounding the condition number and the error of the computation. To do so, we first note that \eqref{eq:gaussian-ode} has the following closed-form solution
\begin{align}\label{eq:solution-to-gaussian-ode}
\boldsymbol\mu(t) = \exp(\boldsymbol\Omega \mathbf M (t-t_0)) \boldsymbol\mu(t_0) + (\boldsymbol\Omega \mathbf M)^{-1} \left(\exp(\boldsymbol\Omega \mathbf M (t-t_0)) - \mathbb I \right) \boldsymbol\Omega \mathbf d
\end{align}
As previously studied in the literature \cite{berry2017quantum, krovi2023improved}, we can encode \eqref{eq:solution-to-gaussian-ode} as a linear system problem. The idea is to Taylor-expand the exponential function. Note that we have
\begin{align}\label{eq:expanded-ode}
\boldsymbol\mu(t+h) = \sum_{l=0}^\infty \frac{(\boldsymbol\Omega \Bf M h)^l}{l!} \boldsymbol\mu(t_0) + h \sum_{l=1}^\infty \frac{(\boldsymbol\Omega \Bf M h)^{l-1}}{l!} \boldsymbol\Omega\Bf d.
\end{align}
Therefore, an algorithm for solving \eqref{eq:gaussian-ode} would be to discretize with $h$ spacing in time and truncate \eqref{eq:expanded-ode} up to order $k$ to get the update rule i.e.,
\begin{align}
\boldsymbol\mu_{i}((j+1)h) = T_{k}(\boldsymbol\Omega\Bf M_i) \boldsymbol\mu_i(jh) + h S_k(\boldsymbol\Omega\Bf M) \boldsymbol\Omega \Bf d,
\end{align}
where $T_k(z):= \sum_{l=0}^k \frac{z^l}{l!}$ and $S_k(z):= \sum_{l=1}^k \frac{z^{l-1}}{l!}$.
As noted in \cite[Section 4]{krovi2023improved}, this update rule can be reformulated as a linear system:
\begin{align}\label{eq:gaussian-linear-system-inclusion}
\mathbf A_i \mathbf x = \mathbf b,
\end{align}
where
\begin{align}
\begin{split}
\mathbf A_i &= \mathbb I - \mathbf L_i,\\
\mathbf L_i &= \sum_{j=0}^m \mathbf e_{j+1} \mathbf e_j^T \otimes \mathbf N_{2}(\mathbb I - \mathbf N_{1,i})^{-1}\\
\mathbf N_{1,i} &= \sum_{l=1}^{k-1} \mathbf e_{l+1} \mathbf e_l^T \otimes \frac{\boldsymbol\Omega \Bf M_i h}{l+1}\\
\mathbf N_{2} &= \sum_{l=0}^k \Bf e_0 \Bf e_l^T \otimes \mathbb I
\end{split}
\end{align}
where indices $j\in\{0,1,\cdots,m\}$ count time steps in the evolution (with time discretization $h$), and $l\in\{0,1,\cdots,k\}$ accounts for terms in the Taylor expansion up to order $k$. Note that $\mathbf N_{1,i}$ has a simple form, allowing us to write $(\mathbb I - \mathbf N_{1,i})^{-1} = \sum_{l=0}^{k}\sum_{l'=0}^{k-l} \mathbf e_{l+l'} \mathbf e_l^T \otimes \frac{l!(\boldsymbol\Omega\Bf M_i)^{l'}}{(l+l')!}$. Later on, we remark on how to prepare an oracle for the term $\Bf N_2(\mathbb I-\Bf N_{1,i})^{-1}$. In what follows, we provide a statement that summarizes what we aim to achieve. Note that the linear system of \eqref{eq:gaussian-linear-system-inclusion} has size $\lceil\frac{T}{h}\rceil \times k \times m$. We now borrow a result, which simplifies our analysis.

\begin{lemma}[Adaption of Theorem 3 and 4 of \cite{krovi2023improved}]\label{lem:adaption-from-krovi}
Assuming $h\norm{\mathbf M_i}\leq 1$, $\norm{\exp(\boldsymbol\Omega\Bf M_i)} \leq C$, and $(k+1)!\geq \frac{Te^3}{h\epsilon}(\norm{\boldsymbol\mu_{i+1}}+Te^2\norm{\Bf d})$, we have that the linear system \eqref{eq:gaussian-linear-system-inclusion} gets an $\epsilon$-approximate solution to \eqref{eq:gaussian-ode}, and that the condition number of the linear system is bounded as $\kappa(\Bf A_i) \leq O(C\cdot\frac{T}{h})$.
\end{lemma}
\begin{proof}
The statement is just rephrasing Theorems 3 and 4 from \cite{krovi2023improved} in the settings of our problem.
\end{proof}
Using the above result, we get the following.
\begin{lemma}\label{lem:parameter-choice-for-inclusion-in-BQL}
Choosing
\begin{align}
h = \Theta(\frac{1}m),\quad k = \Theta\left(\frac{\log(\frac{T}{\epsilon} (E^\ast + mT))}{\log \log (\frac{T}{\epsilon} (E^\ast + mT))} \right),
\end{align}
$\mathbf A_i$ has condition number 
\begin{align}
\kappa(\mathbf A_i) = O(E^\ast{} \cdot Tm),
\end{align}
and also, the linear system \eqref{eq:gaussian-linear-system-inclusion} simulates the ordinary differential equation \eqref{eq:gaussian-ode} with error $\epsilon$.
Note that the linear system has size $O(\frac{Tm}h \log(TE^\ast m/\epsilon))$.
\end{lemma}
In order to prove the above lemma, we first need to show a few favorable properties of the generator $\boldsymbol\Omega\Bf M_i$, which arise from the bounded energy assumption.

\begin{lemma}\label{lem:norm-of-exponentiation-bound}
For Hamiltonians $H_i$ and evolution times $t_i$ satisfying the energy bound condition of \cref{def:Gaussian-simulation}, we have
\begin{align}
\norm{\exp(\boldsymbol\Omega\mathbf M_i t_i )} = \norm{\exp(-\boldsymbol\Omega\mathbf M_i t_i)} \leq 2E^\ast
\end{align}
\end{lemma}
\begin{proof}
Note that $\exp(\boldsymbol\Omega\mathbf M_i t)$ is a symplectic matrix. For a symplectic matrix $\mathbf S$, we have $\mathbf S^{-1} = -\boldsymbol\Omega \mathbf S^T \boldsymbol\Omega$, and therefore, the singular values of $\mathbf S^{-1}$ and $\mathbf S$ coincide, which gives the equality in the statement of the lemma. To prove the upper bound, let $\langle X_j^{(i)}\rangle$ denote the $X_j$ operator expectation value right after the application of $i$-th Hamiltonian. Note that we have
\begin{align}
\langle X_i^{(j)} {}^2\rangle \geq \langle (X_i^{(j)})^2 \rangle - \langle X_i^{(j)} \rangle^2.
\end{align}
This yields
\begin{align}
\langle \sum_{j=1}^m N_j^{(i)} \rangle \geq \tr(\boldsymbol\Sigma_i)\geq \frac12 \norm{\mathbf S_i}^2
\end{align}
as the energy is bounded throughout by $E^{\ast}$, we get
\begin{align}\label{eq:norm-bound-on-S_i}
\norm{\mathbf S_i}^2 \leq 2E^{\ast}.
\end{align}
as $\mathbf S_i = \exp(\boldsymbol\Omega\Bf M_i t_i) \Bf S_{i-1}$, we have
\begin{align}
\begin{split}
\norm{\exp(\boldsymbol\Omega \Bf M_i t_i)}&=\norm{\exp(\boldsymbol\Omega \Bf M_i t_i) \Bf S_{i-1} \Bf S_{i-1}^{-1}}\\
&\leq \norm{\exp(\boldsymbol\Omega \Bf M_i t_i) \Bf S_{i-1}} \norm{\Bf S_{i-1}^{-1}}\\
&= \norm{\mathbf S_i} \norm{\mathbf S_{i-1}^{-1}}\\
&{\leq} 2E^{\ast}
\end{split}
\end{align}
where the last inequality follows from the fact that $\norm{\mathbf S_{i-1}^{-1}} = \norm{\mathbf S_{i-1}}$ (as mentioned earlier, the singular values of a symplectic matrix coincide with those of its inverse), and that the bound \eqref{eq:norm-bound-on-S_i} applies to all $i\in\{1,\cdots,n\}$.
\end{proof}

We are now ready to prove \cref{lem:parameter-choice-for-inclusion-in-BQL}.
\begin{proof}[Proof of \cref{lem:parameter-choice-for-inclusion-in-BQL}]
Note that $\langle X_i^{(j)} {}^2\rangle \geq \langle X_i^{(j)}\rangle^2 = (\boldsymbol{\mu}_i)_j$, and similarly for the momentum operator. Therefore, as the energy is bounded, we get $\norm{\boldsymbol{\mu_i}}\leq E^\ast$ for all $i\in\{1,\cdots,n\}$. Furthermore, \cref{lem:norm-of-exponentiation-bound} puts a bound on the $C$ parameter in \cref{lem:adaption-from-krovi} as $C\leq 2E^\ast$. Furthermore, as $\norm{\mathbf M}\leq O(m)$, we can choose $h=\Theta(1/m)$ to ensure $h\norm{\mathbf M}\leq 1$. Therefore, we can use \cref{lem:adaption-from-krovi} directly to conclude the result.
\end{proof}
Note that it remains to show that we can have a logspace oracle computing entries of $\Bf N_2(\mathbb I - \mathbf N_{1,i})$. Given the form of $\mathbf N_2$, we merely need to compute entries of $(\mathbb I-\mathbf N_{1,i})^{-1}$. Note that we have $\norm{\mathbb I-\mathbf N_{1,i}} \leq 2$, as our choice of $h$ guarantees $\norm{\Bf N_{1,i}}\leq 1$. Moreover, we have $(\mathbb I-\mathbf N_{1,i})^{-1} = \sum_{j=0}^{k-1} \mathbf N_{1,i}^j$, and hence, $\norm{(\mathbb I-\mathbf N_{1,i})^{-1}} \leq k$, which results in $\kappa({\mathbb I-\mathbf N_{1,i}} )\leq 2k$. Therefore, we can use a subroutine and compute the $(\mathbb I-\mathbf N_{1,i})^{-1}$ up to a desired error. We summarize this error analysis in the following statement.

\begin{lemma}
Given an oracle that computes entries of $(\mathbb I - \mathbf N_1)^{-1}$ up to error $\widetilde\epsilon = O(\frac{\epsilon}{(k+1)E^\ast T m})$, and using it for solving \eqref{eq:gaussian-linear-system-inclusion}, with the parameter choices of \cref{lem:parameter-choice-for-inclusion-in-BQL}, we can get a $2\epsilon$-accurate solution to \eqref{eq:gaussian-ode}.
\end{lemma}
\begin{proof}
Note that if we can find $(\mathbb I-\mathbf N_1)^{-1}$ with precision $\widetilde\epsilon$, then, we have evaluated $\mathbf A_i$ with precision $\widetilde\epsilon \cdot (k+1)$. Note that
\begin{align}
\begin{split}
\norm{\mathbf A^{-1} - \mathbf B^{-1}} &= \norm{\mathbf A^{-1}(\mathbf B - \Bf A )\Bf B^{-1}}\\
&\leq \norm{\Bf A^{-1}} \norm{\mathbf B^{-1}} \cdot \norm{\Bf B-\Bf A}\\
&\leq \frac{\norm{\Bf A - \Bf B}}{(s_{\mathrm{min}}(\Bf A) - \norm{\Bf A - \Bf B})^2}
\end{split}
\end{align}
for any pair of invertible matrices $\mathbf A,\Bf B$, where $s_{\mathrm{min}}$ denotes the smallest singular value. We replace $\mathbf A$ with $\mathbf A_i$, and $\mathbf B$ with its perturbed version, which we get by utilizing the subroutine computing inverse of $\mathbb I-\mathbf N_1$ with error $\widetilde\epsilon = O(\frac{1}{k})$ (which is implied by the choice of $\widetilde\epsilon = O(\epsilon/(kE^\ast T))$ in the statement of the lemma). Note that with this choice, we can compute the . Hence, we conclude that an error of $O(\widetilde\epsilon \cdot (k+1))$ in $\mathbf A_i$ would result in at most an error of $\kappa_{\mathbf A_i}^2 (k+1) \widetilde\epsilon = O(E^\ast T m (k+1) \widetilde\epsilon) = O(\epsilon)$ where the last equality follows from the choice of $\widetilde\epsilon$ in the statement.
\end{proof}

Summarizing the above results, we get the following statement.

\begin{theorem}\label{thm:GDC-is-in-BQL}
It is the case that
\begin{align}
\mathsf{GausSim} \leq_{\mathbf L} \text{poly-conditioned-}\mathsf{MatInv},
\end{align}
which also implies $\mathbf{GDC}\subseteq \mathbf{BQL}$.
\end{theorem}

\subsection{Simulating $\mathbf{BQL}$ with Gaussian computations} 

Here, we show how one can solve poly-conditioned-$\mathsf{MatInv}$ via $\mathbf{GDC}$. To this end, we first prove an intermediary lemma.

\begin{lemma}\label{lem:gaussain-bounding-lemma}
Let $\mathbf A$ be an anti-Hermitian matrix with singular values larger than $s_{\mathrm{min}}(\mathbf A)$ and smaller than $s_{\mathrm{max}}(\mathbf A)$. Then, we have that
\begin{align}\label{eq:matrix-sum-bound}
\norm{\frac{1}{\sqrt{\lceil T/\epsilon\rceil}}\sum_{k=0}^{\lceil T/\epsilon \rceil-1} e^{\epsilon\mathbf A k}} \leq \delta
\end{align}
is guaranteed with any $\epsilon \leq s_{\max}(\mathbf A)^{-1}$ and $T \geq \frac{16\, s_{\max}(\mathbf A)}{\delta^2 s^2_{\min}(\mathbf A)}$. Note that with this choice, the number of terms in the summation is $\lceil T/\epsilon \rceil = \Theta(\frac{\kappa^2(\mathbf A)}{\delta^2})$.
\end{lemma}

\begin{proof}
Let $x\in\mathbb R$ with $|x|\in[s_{\min}(\mathbf A), s_{\max}(\Bf A)]$, and note that
\begin{align}\label{eq:number-sum-bound}
\begin{split}
\left|\frac{1}{\sqrt{\lceil T/\epsilon\rceil}}\sum_{k=0}^{\lceil T/\epsilon \rceil-1} e^{\epsilon x k}\right| &= \frac{1}{\sqrt{\lceil T/\epsilon \rceil}} \left|\frac{e^{ix\lceil T/\epsilon\rceil \epsilon}-1}{e^{i\epsilon x}-1} \right|\\
&\leq \frac{4}{\sqrt{T\epsilon}} \frac{1}{x}\\
&\leq \frac{4}{\sqrt{T\epsilon} \, s_{\min}(\mathbf A)}\\
&\leq \delta
\end{split}
\end{align}
where the second line follows from the fact that $|e^{i\theta}-1|\leq 2$ for all real $\theta$, and that $|e^{i\theta}-1|\geq \frac12 |\theta|$ for $-1\leq \theta\leq 1$. As for an anti-Hermitian $\Bf A$, we have $\exp(\mathbf A t)= \sum_{j} e^{i\lambda_jt} \ket{v_j}\bra{v_j}$ where $i\lambda_j$ are eigen-values of $\mathbf A$ satisfying $|\lambda_j|\in[s_{\min}(\mathbf A), s_{\max}(\mathbf A)]$. Using the statement above, and the aforementioned spectral decomposition for $e^{\Bf At}$, we can readily get \eqref{eq:matrix-sum-bound}.
\end{proof}

Now, recall the solution to a Gaussian dynamics \eqref{eq:solution-to-gaussian-ode} rewritten below
\begin{align}
\boldsymbol\mu(t) = \exp(\boldsymbol\Omega \mathbf M (t-t_0)) \boldsymbol\mu(t_0) + (\boldsymbol\Omega \mathbf M)^{-1} \left(\exp(\boldsymbol\Omega \mathbf M (t-t_0)) - \mathbb I \right) \boldsymbol\Omega \mathbf d.
\end{align}

Let $\mathbf A$ be the matrix of interest in our poly-conditioned-$\mathsf{MatInv}$ problem. Note that it is sufficient to invert anti-Hermitian matrices. This is due to the fact that for any given matrix $\mathbf B$, we can write
\begin{align}
\begin{pmatrix}
0 & \Bf B\\
-\Bf B^T & 0
\end{pmatrix}^{-1}
=
\begin{pmatrix}
0 & -(\Bf B^T)^{-1}\\
\Bf B^{-1} & 0
\end{pmatrix},
\end{align}
and hence, by inverting the anti-symmetric matrix $\begin{pmatrix}
0 & \Bf B\\
-\Bf B^T & 0
\end{pmatrix}$, we are able to find $\Bf B^{-1}$ (note that our transformation leaves the condition number unchanged).
Hence, from hereon, we consider the input matrix $\mathbf A$ is anti-Hermitian.
The idea of the reduction then is to consider
\begin{align}
\Bf M^{(xx)} = \Bf M^{(pp)} = 0,\quad \Bf M^{(xp)} = \mathbf A.
\end{align}
With this choice, we have
\begin{align}\label{eq:reduction-M}
\boldsymbol\Omega \mathbf M = \begin{pmatrix}
\mathbf A & 0\\
0 & \mathbf A
\end{pmatrix}
= \mathbf A \oplus \mathbf A
\end{align}
Note that such a Hamiltonian, applies a beamsplitter on all modes, where the unitary $U$ describing the evolution in the annihilation operator's basis i.e., $a\mapsto Ua$, is a real-valued matrix. Moreover, in the position and momentum operator basis, we have that the evolution of the position and momentum parts decouple due to \eqref{eq:reduction-M}. 

Consider the following set of equations
\begin{align}
\begin{split}
\boldsymbol\mu_1 &= \begin{pmatrix}
\mathbf{A}^{-1}\left(\exp(\epsilon\mathbf A) - \mathbb I\right) & 0\\
0 & \mathbf{A}^{-1}\left(\exp(\epsilon\mathbf A) - \mathbb I\right) 
\end{pmatrix}\boldsymbol\Omega\mathbf d,\\
\boldsymbol\mu_2 &= \begin{pmatrix}
\mathbf{A}^{-1}\left(\exp(2\epsilon\mathbf A) - \mathbb I\right) & 0\\
0 & \mathbf{A}^{-1}\left(\exp(2\epsilon\mathbf A) - \mathbb I\right) 
\end{pmatrix}
\boldsymbol\Omega\mathbf d\\
&\; \vdots\\
\boldsymbol\mu_{N-1} &= \begin{pmatrix}
\mathbf{A}^{-1}\left(\exp((N-1)\epsilon\mathbf A) - \mathbb I\right) & 0\\
0 & \mathbf{A}^{-1}\left(\exp((N-1)\epsilon\mathbf A) - \mathbb I\right) 
\end{pmatrix} \boldsymbol\Omega\mathbf d
\end{split}
\end{align}
Also, let $\boldsymbol\mu_0=0$ for convenience. Note that taking $\epsilon = s_{\max}(\Bf A)^{-1}$ with $N=\lceil \frac{4\kappa^2(\Bf A)}{\delta^2} \rceil$ gives
\begin{align}
\norm{\sum_{i=0}^{N-1} \frac{\boldsymbol\mu_i}{\sqrt N} + \boldsymbol{A}^{-1}\boldsymbol\Omega\Bf d} \leq \delta \norm{\Bf A^{-1}}\norm{\Bf d}
\end{align}
from \cref{lem:gaussain-bounding-lemma}. This argument provides us with the following. In order to compute $\boldsymbol\mu_i$, we merely need to run the Hamiltonian described by $\Bf M := -\boldsymbol\Omega \mathbf A$ for time $t_i = i \epsilon = is_{\max}^{-1}(\Bf A)$. Therefore, say we want to compute the $(i,j)$ entry of $\mathbf A^{-1}$ with precision $\tilde\delta$. It suffices to set the displacements to $\mathbf d = (0, -\Bf e_j)$, $\tilde\delta = s^{-1}_{\min}(\mathbf A)\delta$, and take the average of the measurement of the $i$-th location in each of the $N-1$ chunks which can be performed via the gadget of \cref{prop:sample-mean}. Overall, we have used $T = O(\frac{\kappa(\mathbf A)^4}{\tilde\delta^2})$ amount of time and $m\cdot N = O(m\frac{\kappa(\Bf A)^2}{\tilde\delta^2})$ many modes, where $m$ is the size of $\Bf A$ (i.e., $\Bf A\in\mathbb R^{m\times m}$). Finally, we need to ensure that the energy bound is also satisfied. Note that the states we obtain have covariance matrices and means as
\begin{align}
\boldsymbol\mu_i(t) = \Bf A^{-1} \left(U_i(t) - \mathbb I \right) \boldsymbol{\Omega} \Bf d, \quad \boldsymbol\Sigma_i = \frac12 \mathbb I,
\end{align}
where $U_i(t)$ are unitaries. Therefore, we have that the energy at any time $t$ is bounded by $ E^\ast = \sup_t \norm{\boldsymbol\mu_i(t)}^2 + \frac12 \times (\# \mathrm{modes})\leq \norm{\mathbf A^{-1}} + O(m\frac{\kappa(\Bf A)^2}{\tilde\delta^2}) = \mathsf{poly}(\kappa(\Bf A), \frac1\delta,m)$.
We summarize this result in the following theorem.

\begin{theorem}\label{thm:BQL-is-in-GDC}
We have $\mathbf{BQL}\subseteq\mathbf{GDC}$. It is the case that
\begin{align}
\text{poly-conditioned-}\mathsf{MatInv} \leq_{\mathbf L} \mathsf{GausSim},
\end{align}
which also implies $\mathbf{BQL}\subseteq\mathbf{GDC}$.
\end{theorem}

\noindent Together with \cref{thm:GDC-is-in-BQL}, this completes the proof of \cref{thm:GDC=BQL}.

\section{The power of non-Gaussian bosonic computations}
\label{sec:CVBQP}

As discussed before, a formulation for the theory of computation using continuous quantum variables was first proposed by Lloyd and Braunstein \cite{lloyd1999quantum}. In this model, a continuous variable corresponds to a possibly unbounded operator $O$ acting on an infinite-dimensional Hilbert space $\mathcal H$. 
In the original paper by Lloyd and Braunstein \cite{lloyd1999quantum} and subsequent works, such as that by Sefi and van Loock \cite{sefi2011decompose}, it is argued by an approximation argument that the Gaussian gate set $\set{e^{it_1X}, e^{it_2X^2}, e^{i\frac{\pi}{4}(X^2 + P^2)}}$ along with a single non-Gaussian gate of degree 3 or more in the position and momentum operators, for example $e^{it_3X^3}$, is sufficient to simulate the evolution generated by Hamiltonians that are arbitrary polynomial operators in terms of $X$ and $P$. This statement was made precise in the context of (smooth) controllability theory by Wu, Tarn and Li in \cite{wu2006smooth}, however the computational efficiency of the simulation is unknown to the best of our knowledge. In successive works due to Bartlett et al.\ \cite{bartlett2002efficient} and Ukai et al.\ \cite{ukai2010universal}, the necessity of Gaussian gates for simulating all degree-2 operators is also shown.

In what follows, we define and study complexity classes of bounded-error continuous-variable quantum polynomial time computations, based on circuits generated by Gaussian and non-Gaussian Hamiltonians that are polynomials in the position and momentum bosonic operators, in \cref{sec:CVBQPdef}. We show that non-Gaussian bosonic computations can simulate $\mathbf{BQP}$ in \cref{sec:BQP-in-CVBQP}. Then, we focus on bosonic computations based on Gaussian and cubic phase gates and show that they can be strongly simulated in $\mathbf{EXPSPACE}$ in \cref{sec:expspace}. Finally, we study the problem of computing expectation values for the output of these bosonic circuits in \cref{sec:PARalgorithm} and show that it is contained in $\mathbf{PSPACE}$.

\subsection{\textbf{CVBQP}: bounded-error continuous-variable quantum polynomial time computations}
\label{sec:CVBQPdef}

In order to define bosonic computational complexity classes based on polynomial Hamiltonians, we first define the corresponding gates and circuits:

\begin{definition}[Bosonic gate of constant degree] A bosonic gate of constant degree over $n$ modes applied for time $t$ is a unitary operator $U$ over $n$ modes which may be written as $U=e^{-itH}$. Here, $H\in\mathcal P_n$ is a polynomial of constant degree in the position and momentum operators of the modes with constant coefficients.
\end{definition}

\noindent For instance, Gaussian unitary gates such as displacement or squeezing fall into that class. Similarly, cubic phase gates are also bosonic gates of constant degree. 

In general, a bosonic circuit with gates of constant degree is described by a collection of Hamiltonians that are polynomials of constant degree in the position and momentum operators of the modes with constant coefficients, together with the times for which these Hamiltonian evolutions are being applied. We define the output of such a circuit to be a classical measurement outcome of the single-mode number operator $N$, which is a non-negative integer in general. More formally:

\begin{definition}[Bosonic circuit with gates of constant degree with number measurement] A bosonic circuit $\mathcal C$ over $n$ modes with $m$ gates is described by constant-degree Hamiltonians $H_1,\dots,H_m\in\mathcal P_n$, gate times $t_1,\dots,t_m\in\mathbb R$. It takes as input $x\in\{0,1\}^n$ and returns an outcome $\mathcal C(x)\in\mathbb N$ which is the result of the measurement of $N$ on the first mode of the output state
\begin{equation}
    e^{-it_mH_m}\cdots e^{-it_1H_1}|x\rangle,
\end{equation}
where $|x\rangle$ denote a coherent state over $n$ modes.
\end{definition}

%
%

This definition leads to the following notion of bosonic circuit complexity:

\begin{definition}[Bosonic circuit complexity] The circuit complexity of a bosonic circuit $\mathcal C$ over $n$ modes with $m$ gates described by constant-degree Hamiltonians $H_1,\dots,H_m$ and gate times $t_1,\dots,t_m$ is given by $t_1+\dots+t_m$.
\end{definition}

In practice, the set of gates and measurements that can be implemented on a quantum system is restricted, and the computational equivalence between different choices of gate sets and measurements is not at all straightforward. 
In what follows, we employ the notation $\mathcal C_n[H_1,\dots,H_m]$ to denote a bosonic circuit with specific non-Gaussian gates generated by Hamiltonians $H_1,\dots,H_m$ together with Gaussian gates, and single-mode particle-number measurement $N$.
Note that we only specify the non-Gaussian gates for brevity, as one can always efficiently compile any Gaussian unitary generated by a quadratic polynomial using any Gaussian gate set that is universal for Gaussian computations, such as $\set{e^{it_1X}, e^{it_2X^2}, e^{i\frac{\pi}{4}(X^2 + P^2)}}$ \cite{becker2021energy}.  

Given a choice of bosonic gates, we define the corresponding complexity class of bounded-error continuous-variable quantum polynomial time computations as follows:

\begin{definition}[CVBQP]\label{def:CVBQPfixednG}
    $\mathbf{CVBQP}[H_1,\dots,H_m]$ is the class of languages $L \subseteq \{0,1\}^*$ for which there exists a polynomial-time uniform family $\{\mathcal C_n[H_1,\dots,H_m]\}_{n\in\mathbb N}$ of bosonic circuits of polynomial circuit complexity with non-Gaussian gates $H_1,\dots,H_m\in\mathcal P_{O(1),d}$ with $d>2$ and single-mode number measurement, such that: 
    \begin{itemize}
        \item for all $x\in L$, $\mathcal C_{|x|}[H_1,\dots,H_m](x)\in[b,b+E]$ (accepts) with probability greater than $\frac23$,
        \item for all $x \notin L$,  $\mathcal C_{|x|}[H_1,\dots,H_m](x)\in[a-E,a]$ (rejects) with probability greater than $\frac23$,
    \end{itemize} 
    given constants $a,b,E=O(\mathsf{poly}(n))$ such that $b-a\ge\Omega(\frac1{\mathsf{poly}(n)})$ and $a-E\ge0$.
\end{definition}

\noindent Here and after, we omit the $a,b,E$ dependencies for brevity. From a computational perspective, the acceptance intervals in the definition ensure that the measurement outcome (which in principle could be any non-negative integer) may be processed by a Turing machine up to sufficient precision to assess whether it is smaller than $a$ or greater than $b$. From a physical standpoint, it correspond to the practical limitation of a real-world scenario in which the energy range of a computation is bounded on average (see \cref{sec:discussionE} for a discussion).

We write $\mathbf{CVBQP}$ when no restriction on the set of constant-degree gates is made. We have $\mathbf{CVBQP}[H_1,\dots,H_m]\subseteq\mathbf{CVBQP}$ for any $H_1,\dots,H_m\in\mathcal P_{O(1),d}$ with $d>2$.
As previously mentioned, the computational equivalence between different non-Gaussian gate sets is a highly non-trivial problem, and only known in select cases where exact gate decompositions are available \cite{kalajdzievski2019exact}. We leave the reverse inclusion $\mathbf{CVBQP}\subseteq\mathbf{CVBQP}[H_1,\dots,H_m]$ for specific $H_1,\dots,H_m\in\mathcal P_{O(1)}$ of degree greater than $2$ as an open question.

\cref{def:CVBQPfixednG} leads to the following $\mathbf{CVBQP}[H_1,\dots,H_m]$-complete problem.

\begin{definition}[Probability estimation for bosonic circuits]\label{def:CVBQPfixednGcompleteprob}
    Given a polynomial-time uniform family of bosonic circuits of polynomial circuit complexity $\{\mathcal C_n[H_1,\dots,H_m]\}_{n\in\mathbb N}$ with non-Gaussian gates $H_1,\dots,H_m\in\mathcal P_{O(1),d}$ with $d>2$, $\mathsf{ProbaEstBosonCircuit}[H_1,\dots,H_m]$ is the problem of deciding whether $\mathrm{Pr}[C_n[H_1,\dots,H_m](0)\in[b,b+E]]$ is greater that $\frac23$ or $\mathrm{Pr}[C_n[H_1,\dots,H_m](0)\in[a-E,a]]$ is greater than $\frac23$, promised that one of the two cases hold, given $a,b,E$ such that $b-a\ge\Omega(\frac1{\mathsf{poly}(n)})$ and $a-E\ge0$.
\end{definition}

\begin{proposition}\label{prop:CVBQPfixednGcompleteprob}
    $\mathsf{ProbaEstBosonCircuit}[H_1,\dots,H_m]$ is $\mathbf{CVBQP}[H_1,\dots,H_m]$-complete.
\end{proposition}

\begin{proof}
    Without loss of generality, we can always hardwire the description of $x=(x_1,\dots,x_n)$ into a new circuit $\mathcal C_n^{x}$ by adding a layer of displacement operators with amplitudes $x_1\dots x_n$ and assume that the input is the vacuum state.
\end{proof}

In what follows, we provide upper and lower bounds on specific instances of $\mathbf{CVBQP}[H_1,\dots,H_m]$.

\subsection{Simulating \textbf{BQP} with non-Gaussian computations}
\label{sec:BQP-in-CVBQP}

It is well-known that discrete-variable quantum computations may be embedded in bosonic systems. Some of the most prominent ways are due to Knill, Laflamme and Milburn \cite{knill2001scheme} and to Gottesman, Kitaev and Preskill (GKP) \cite{gottesman2001encoding}. The former encodes a qubit using a single-particle Fock state in two modes (\textit{dual-rail encoding}), and performs the computation using Gaussian (passive linear) unitary operations together with non-Gaussian (particle-number) measurements and feed-forward. The latter encodes a qubit in a complex quasi-periodic non-Gaussian state known as a GKP state (\textit{GKP encoding}), and performs the computation using Gaussian and cubic phase gates together with Gaussian (position) measurements and feed-forward. Both schemes are capable of simulating \textbf{BQP} computations efficiently, but require feed-forward of measurement outcomes and a supply of non-Gaussian states.

We show that a simple instance of bosonic computations using constant-degree gates on coherent state input and number measurement, as in \cref{def:CVBQPfixednG}, can simulate universal discrete-variable quantum computations, without the need for feed-forward of measurement outcomes.

First, we define a pair of operators analogous to the Pauli unitaries for a qubit.
\begin{definition}\label{def:PauliLikeOperators}
    Let 
    \begin{align}
    \overline{\sigma_x} := \left(\mathbb I - a^\dagger a  \right) a + h.c., \quad \overline{\sigma_z} := \mathbb I -2a^\dagger a.
    \end{align}
\end{definition}

\begin{theorem}\label{th:CVBQPcontainsBQP} 
It is the case that
\begin{equation}
\mathbf{BQP}\subseteq \mathbf{CVBQP}[\overline{\sigma_x},\overline{\sigma_z}\otimes\overline{\sigma_z}].
\end{equation}
\end{theorem}

\begin{proof}
Note the form of these operators in Fock basis
\begin{align}\label{eq:ZX}
\overline{\sigma_x} &= \left(\mathbb I - a^\dagger a  \right) a + h.c. = 
\begin{pmatrix}
0 & 1 & 0\\
1 & 0 & 0\\
0 & 0 & \ast
\end{pmatrix},\\
\overline{\sigma_z} &= \mathbb I -2a^\dagger a = 
\begin{pmatrix}
1 & 0 & 0\\
0 & -1 & 0\\
0 & 0 & \ast
\end{pmatrix},
\end{align}
where the matrices are in the Fock basis and each $\ast$ denotes an operator acting on the subspace of two or more particles. Now if we start from the vacuum state on all modes, i.e., $\ket{0^n}$, and we apply Hamiltonians of the form $e^{i\overline P_1\otimes\overline P_2\otimes\cdots\otimes\overline P_n}$ for $P_1,\dots,P_n\in\{\sigma_x,\sigma_z\}$ being the usual Pauli matrices, we remain in the subspace spanned by the Fock states $\ket{m_1,\dots,m_n}$, where $m_i\in\{0,1\}$. Since unitaries generated by Pauli Hamiltonians $\sigma_z\otimes \sigma_z$, $\sigma_x$, and $\sigma_z$ form a dense subset of all unitaries, we obtain that the Hamiltonians $\overline{\sigma_x}\in\mathcal P_{1,3}$ and $\overline{\sigma_z}\otimes\overline{\sigma_z}\in\mathcal P_{2,4}$ are capable of preparing the output states of discrete-variable quantum circuits, together with quadratic Hamiltonians. Indeed, recall that from the Solovay--Kitaev theorem this dense subset also allows efficient decompositions \cite{dawson2005solovay}, e.g., one may prepare elements of the Clifford + $T$ gateset by composing constantly many such unitaries. Finally, a measurement of the first qubit in the computational basis can be simulated by measuring the number operator on the first mode.
\end{proof}

\subsection{Simulating non-Gaussian computations in \textbf{EXPSPACE}}
\label{sec:expspace}

In this section, we study the problem of \textit{strongly} simulating bosonic computations.

In the discrete-variable theory of quantum complexity, strong simulation of polynomial-size quantum circuits, i.e., compute individual amplitudes within exponentially small precision, is complete for the complexity class $\#$\textbf{P} corresponding to counting the number of solutions to an \textbf{NP}-complete problem. Roughly speaking, the proof is based on viewing each amplitude as a weighted enumeration over different paths from the input to output configurations. In the discrete-variable formalism, local quantum gates acting on a few qubits can be represented using local matrices (represented by a constant number of parameters). This locality structure is crucial to the proof since, using these local representations, we can describe each amplitude as a weighted enumeration over, at most, exponentially many paths, each of which can be evaluated in polynomial time. The exponential upper bound on the number of paths is due to the finite (constant) dimension of each subsystem in the tensor product.

Can we prove a similar result in the continuous-variable domain? In what follows, we provide evidence for a negative answer, for continuous-variable quantum circuits based on cubic and Gaussian gates, together with number measurements. The class of decision problems that can be solved using these circuits and efficient classical post-processing is denoted as $\mathbf{CVBQP}[X^3]$ (see \cref{def:CVBQPfixednG}). Because exact gate decompositions are available, we remark that $\mathbf{CVBQP}[X^k]\subseteq\mathbf{CVBQP}[X^3]$, when $k$ is divisible by 2 or 3 \cite{kalajdzievski2019exact}, so all of our upper bounds hereafter also apply to those cases.

Due to the infinite-dimensional Hilbert space even for a single mode, it is unclear how one would formulate each amplitude as a weighted enumeration of a finite number of paths. In this section, we prove that continuous-variable quantum circuits that use $T$ cubic phase gates and $T^{O(1)}$ Gaussian gates over $T^{O(1)}$ modes can be strongly simulated (i.e., each amplitude can be estimated within additive error $2^{-2^{T^{O(1)}}}$) in \textbf{SPACE} $(2^{T^{O(1)}})$. More precisely:

\begin{theorem} [Strong simulation of bosonic circuits in \textbf{EXPSPACE}]
    Given a circuit $C$ of $T$ cubic phase gates and $\poly(T)$ Gaussian gates on $\poly(T)$ modes, then for $n, m \leq 2^{T^{O(1)}}$, the amplitudes $\bra{n} C \ket{m}$ can be computed within $2^{-2^{T^{O(1)}}}$ error in \textbf{SPACE} $(2^{T O(\log T)})$. 
    \label{thm:schrod-in-expspace}
\end{theorem}

An immediate corollary of this Theorem is that: 
\begin{corollary}
    $\mathbf{CVBQP}[X^3]$ can be simulated in $\mathbf{EXPSPACE}$.
    \label{th:CVBQPinEXPSPACE}
\end{corollary}

 \begin{proof}[Proof sketch of \cref{thm:schrod-in-expspace}]
 The proof strategy is as follows: first, we show that the output state of a $\mathbf{CVBQP}[X^3]$ circuit with polynomial bosonic circuit complexity can be well approximated by the output state of a bounded-energy version of the circuit (\cref{prop:error}). Then, we show that this energy bound can be taken as doubly exponential (\cref{prop:energyboundcubic}). This implies that $\mathbf{CVBQP}[X^3]$ is included in $\mathbf{EEXP}$ (\cref{thm:eexp-bound}). We then use a Feynman path technique to compute the output probability of $\mathbf{CVBQP}[X^3]$ circuits in an appropriate form. To bring the complexity down to \textbf{EXPSPACE}, we show that up to rescaling by efficiently computable numbers, the Feynman path form of the amplitudes are polynomials of degree at most doubly exponential in the number of cubic phase gates. We then apply a well-known depth reduction technique of Valiant, Skyum, Berkowitz, and Rackoff \cite{valiant1981fast}, which states: if a polynomial of degree $d$ can be computed using a boolean circuit of size $s$, then it can be calculated using a boolean circuit of size $(sd)^{O(1)}$ and depth $O(\log (s) \log (d))$. This implies strong simulation is possible in exponential depth and doubly exponential size. Using standard techniques, we conclude that this model can be simulated in \textbf{EXPSPACE}.    
 \end{proof}
 
In the following sections, we develop each of the tools outlined in the proof sketch above. We then put everything together to prove \cref{thm:schrod-in-expspace}. 

\subsubsection{Energy cutoff and doubly exponential upper bound}

The main result of this section is the following:
\begin{theorem}[Strong simulation of bosonic circuits in \textbf{EEXP}]
    Given a continuous variable quantum circuit $C$ of $T$ cubic phase gates and $\poly(T)$ Gaussian gates on $\poly(T)$ modes there exists a classical arithmetic circuit of size $2^{2^{O(T \log T)}}$ circuit which strongly simulates the quantum circuit, i.e., estimate $\bra{n} C \ket{m}$ for integers $m, n \leq 2^{2^{\poly(T)}}$ within $2^{-2^{\poly (T)}}$ error.
    \label{thm:eexp-bound}
\end{theorem}

\begin{proof}
    We first show that if the average energy of a computation is $E$ (corresponding to an effective number of particles that contribute to the computation), then we can effectively truncate the Hilbert space to energy $E$. More precisely, we can replace each gate $U$ (cubic or Gaussian) to $U_E$ where $U_E$ is an $E \times E$ matrix with entries $U_{E, m,n} = \bra{m} U \ket{n}$. This is established in \cref{prop:error}. Next, we bound $E$ for the computation by $2^{2^{T \log T}}$. This is established in \cref{prop:energyboundcubic}. Finally, we use a result of \cite{miatto2020fast} to evaluate each matrix entry within doubly exponential time (outlined at the end of this Section).
\end{proof}

To establish this result, we develop a framework to truncate the energy spectrum in a quantum computation. 

 For an operator $O$, let $O_E$ denote the truncation of $O$ up to particle number $E$.

\begin{proposition}
\label{prop:error}
   Let $U = U_T \ldots U_1$ be a single-mode quantum circuit of size $T$ consisting of Gaussian and cubic phase gates. Let $E^* = \max_{i \in \{0,1, \ldots , T\}} \bra{\psi_i} N \ket{\psi_i}$, where $\ket{\psi_i} = U_i \ldots U_1 \ket {\psi_0}$. Then 
   $$
   \|U_E \ket{\psi_0} - U \ket{\psi_0}\| \leq T\cdot \sqrt{\frac{2E^*}{E}} 
   $$
   Here $U_E := U_{T, E} \ldots U_{1, E}$. 
\end{proposition}

In order to prove this result, we need the following Quantum Markov's inequality lemma:
\begin{lemma}[Quantum Markov's inequality]
    Let $ O$ be a PSD operator $: \mathbb{C}^\infty \rightarrow \mathbb{C}^\infty$, and let $\Pi_E$ be the projector onto the eigenspaces of $O$ with eigenvalues less than $E$. Let $\ket{\psi}, \ket{\phi} \in \mathbb{C}^\infty$ be such that $\bra{\psi}  O \ket{\psi} = E_\psi, \bra{\phi}  O \ket{\phi} = E_\phi$, then 
    $$
    |\braket{\psi}{\phi} - \bra{\psi}\Pi_E\ket{\phi}| \leq \frac{\sqrt{E_\psi E_\phi}}{E}.
    $$
    \label{lem:overlap}
\end{lemma}

\begin{proof}
    We observe that $E (I -  \Pi_E) \preceq  O$, therefore, for any density matrix $ \rho$, $\Tr ( \rho (I-  \Pi_E)) \leq \frac{\Tr ( \rho  O)}{E}$, which is reminiscent of Markov's inequality in a quantum setting. Now we apply Cauchy Schwartz $|\bra{\psi}(I-\Pi_E)\ket{\phi}| \leq \sqrt{|\bra{\phi}(I-\Pi_E)\ket{\phi}| \cdot |\bra{\psi}(I-\Pi_E)\ket{\psi}|} \leq \frac{\sqrt{E_{\psi} E_{\phi}}}{E}$.
\end{proof}

Now, we use this lemma for $O=N$ to prove \cref{prop:error}. 

\begin{proof}[Proof of \cref{prop:error}]
The idea is very similar to proving that error grows linearly in discrete-variable quantum circuits. For $i = 1, \ldots, T$, let $\ket{\phi_i} := U_{i,E} \ldots U_{1,E}\ket{\psi_0}$, and define $\delta_i := \|\ket{\psi_i} - \ket{\phi_i}\|$. We set $\ket{\phi_0} = \ket{\psi_0}$ and $\delta_0 = 0$. Our goal is to bound $\delta_T$. We derive a recursive relationship $\delta_i$ and $\delta_{i-1}$ for any $i = 1, \ldots, T$:
\begin{align}
\begin{split}
    \delta_i &= \|U_i \ldots U_1 \ket{\psi_0} - U_{i,E} \ldots U_{1,E} \ket{\psi_0}\| \\ 
    &\leq \|(U_i- U_{i,E}) \ket{\psi_{i-1}}\| +\| U_{i,E} (\ket{\psi_{i-1}}- \ket{\phi_{i-1}})\| \\
    &\leq \|(U_i- U_{i,E}) \ket{\psi_{i-1}}\| +\|  \ket{\psi_{i-1}}- \ket{\phi_{i-1}}\| \\
    &= \|(U_i- U_{i,E}) \ket{\psi_{i-1}}\| + \delta_{i-1}.
\end{split}
\end{align}    
The second line is by triangle inequality, and the third line is by observing that $U_{i,E}$ does not change the Euclidian norm of a vector. Next, we use \cref{lem:overlap} for $O=N$ to bound 
    
\begin{align}
    \begin{split}
        1/2 \|(U_i- U_{i,E}) \ket{\psi_{i-1}}\|^2 &=  \Re{1 -\bra{\psi_{i-1}} U^{-1}_i \Pi_E U_{i} \Pi_E\ket{\psi_{i-1}}}\\
        &\leq  |1 -\bra{\psi_{i-1}} U^{-1}_i \Pi_E U_{i} \Pi_E\ket{\psi_{i-1}}|\\
        & \leq \frac{\sqrt{E_1 E_2}}{E},
    \end{split}
\end{align}
where $E_1 = \bra {\psi_{i-1}} \Pi_E U_i^{-1} O U_i \Pi_E \ket{\psi_{i-1}}$ and $E_2 = \bra{\psi_i} O \ket{\psi_i}$. We have $E_2 \leq E^*$ by definition. To see why $E_1 \leq E^*$, note that $\Pi_i U_i^{-1} O U_i \Pi_E \preceq U_i^{-1} O U_i$, and thus $E_2 \leq \bra{\psi_i} O \ket{\psi_i} \leq E^*$. As a result we obtain the recursion $\delta_i \leq \delta_{i-1} + \sqrt{2E^*/E}$, and hence $\delta_T \leq T \sqrt{\frac{2E^*}{E}}$.
\end{proof}

Next, we consider the number operator $N$ (or any low-degree operator in $X$ and $P$) and find a doubly exponential upper bound on the expectation value of this operator on the output state of a $\mathbf{CVBQP}[X^3]$ circuit of polynomial bosonic circuit complexity. 

\begin{proposition} [Energy in the circuit model]\label{prop:energyboundcubic}
    Let $U = U_T \ldots U_1$ be an $n$-mode unitary circuit composed of $T=\mathsf{poly}(n)$ Gaussian and cubic phase gates, each gate time polynomially bounded, and let $\ket{\psi} = U \ket{0^n}$. Then, $\bra{\psi}  {N} \ket {\psi} = O(2^{2^{\mathsf{poly}(n)}})$. 
\end{proposition}
\begin{proof}
From the Bloch--Messiah decomposition of Gaussian unitary operations (see \cref{thm:EBMdecomp}), a $\mathbf{CVBQP}[X^3]$ circuit may only use passive multimode gates without loss of generality, which does not change the energy, so we only need to track how the energy evolves under single-mode displacement gates, single-mode squeezing gates, and cubic phase gates.

After $T$ gates, the multimode number operator $N$ transforms into a polynomial in $X$ and $P$ of degree at most $2^{T+1}$. To see this, let $O$ be an observable in the form $\sum_{\mu,\nu: |\mu + \nu| \leq d} \alpha_{\mu,\nu} X^{\mu} P^{\nu}$, where we used multi-index notations. Here, the degree of $O$ is at most $d$. If $H$ is a quadratic Hamiltonian, then the degree of $e^{iH}Oe^{-iH}$ does not change. However, when $H$ is a cubic Hamiltonian $sX^3$, the degree of $O$ at most doubles, since $e^{iH}Xe^{-iH}=X$ and $e^{iH}Pe^{-iH}=P+isX^2$. As a result, the degree of $ N(t) := U^\dagger_{T-t+1}\ldots U^\dagger_{T}NU_{T}\ldots U_{T-t+1}$ is at most $2^{t+1}$.

Moreover, after $T$ gates, the magnitude of the coefficients in the polynomial expansion of $N(T)$ is bounded by $O(2^{2^{\mathsf{poly}(n)}})$. To see that, consider a squeezing gate generated by $\frac12r(XP+PX)$, with $r=\mathsf{poly}(n)$ specified using $O(\log(n))$ bits. It maps $X$ to $e^{-r}X$ and $P$ to $e^{r}P$, so at step $t$, when the degree of $N(t-1)$ is bounded by $2^{t}$, it can at most multiply the biggest coefficients in the polynomial expansion of $N(t-1)$ by $e^{|r|2^{t}}=O(2^{2^T\mathsf{poly}(n)})$. The effect of a displacement or a cubic phase gate is similar and more limited. Overall, the magnitude is at most $O(2^{T2^T\mathsf{poly}(n)})=O(2^{2^{\mathsf{poly}(n)}})$.

As a result, we can write the final energy as $\bra{0}N(T)\ket{0} = \sum_{|\mu + \nu| \leq 2^{T+1}} \alpha_{\mu, \nu} \bra{0} X^{\mu} P^{\nu} \ket{0}$, with $\|\boldsymbol{\alpha}\|_1 = \sum_{|\mu + \nu| \leq 2^{T+1}} |\alpha_{\mu, \nu}|=O(2^{2^{\mathsf{poly}(n)}})$, where we used $\sum_{|\mu + \nu| \leq 2^{T+1}}1=\binom{2n+2^{T+1}}{2^{T+1}}=O(2^{2^{\mathsf{poly}(n)}})$. We now bound:
\begin{align}
        |\bra{0}{N}(T)\ket{0}| &\leq \|\boldsymbol{\alpha}\|_1 \cdot \max_{|\mu + \nu| \leq 2^{\mathsf{poly}(n)}} |\bra{0} X^{\mu} P^{\nu} \ket{0}|=O(2^{2^{\mathsf{poly}(n)}}),
\end{align}
where we used a crude upper bound on Eq.~(\ref{eq:genexpvacmon}) for the last step (see \cref{app:XP-vaccuum-exp} for a detailed calculation of position and momentum moments for the vacuum).
\end{proof}

\noindent The logarithmic bits of precision for the gate specification in the above lemma corresponds to gate time that is polynomially-bounded, i.e., bosonic circuits with polynomial circuit complexity, since $T=\mathsf{poly}(n)$.

Combining \cref{prop:error} and \cref{prop:energyboundcubic} shows that a $\mathbf{CVBQP}[X^3]$ computation may be simulated by restricting all gates to the subspace of states over $n$ modes with energy $E=O(2^{2^{\mathsf{poly}(n)}})$. The dimension of this subspace is $\binom{n+E}{E}=O(2^{2^{\mathsf{poly}(n)}})$.

It remains to show that the matrix of each gate in this subspace may be computed in \textbf{EEXP}. We first consider the case of a cubic phase gate $U = e^{i s/3 X^3}$ and use the recursive approach in \cite[Equation (148)]{miatto2020fast}:
\begin{align}
\begin{split}
    \bra{m} U \ket{n} &= \frac{2\sqrt{\pi}ye^{2/(3s^2)}}{\sqrt{2^{m+n} m! n!}} \sum_{k = 0}^m \sum_{l = 0}^n (-\frac{2i}{s^{1/3}})^{m + n - k - l} \binom nk \binom ml H_k (-i/s)H_l (-i/s) Ai^{(n + m - k - l)} (s^{-4/3})\\
    &=: \sum_{k = 0}^m \sum_{l = 0}^n U^{k,l}(m,n)
\end{split}
    \label{eq:cubic-sum2}
\end{align}
where $Ai : \mathbb{C} \rightarrow \mathbb{C}$ is the Airy function $Ai (x) = \frac{1}{2\pi} \int_{-\infty}^\infty dk e^{i k^3/3 + ikx}$, and $Ai^{(l)} (x) = \frac{1}{2\pi} \int_{-\infty}^\infty dk (ik)^l e^{i k^3/3 + ikx}$ is its $l$'th derivative. $H_k$ is the $k$'th Hermite polynomial of degree $k$. We can evaluate the Airy function and its (up to) exponentially many derivatives and Hermite polynomial of up to exponential degree within exponential precision using \textbf{EEXP} computation (this error analysis is carefully implemented in Lemma \ref{lem:cubic+gaussian-degree}). As a result, each coefficient $U^{k,l}(m,n)$ in Eq.~(\ref{eq:cubic-sum2}) can be evaluated within exponential bits of precision in \textbf{EEXP}. Other Gaussian gates can be evaluated in \textbf{EEXP} (and in possibly a much smaller complexity class) based on the recursive approach in \cite{miatto2020fast}.

More general types of measurement can be simulated using the above approach. The problem of estimating the probability of a positive outcome is complete for $\mathbf{CVBQP}[X^3]$ (see \cref{prop:CVBQPfixednGcompleteprob}), so by cutting off the polynomial measurement operator, its eigenstates for positive eigenvalues can be computed in \textbf{EEXP}. Putting things together, the probability of a positive outcome can be estimated up to constant precision (exponential, in fact) in \textbf{EEXP}, which completes the proof of \cref{thm:eexp-bound}.


\subsubsection{Bounding the degree of amplitudes}

The main result of this section is:
\begin{theorem} [Bounding the degree of the amplitudes]
    Given a continuous variable quantum circuit $C$ of $T$ cubic phase gates and $\poly(T)$ Gaussian gates on $\poly(T)$ modes, up to $\poly(T)$ time computable rescaling, there exists a polynomial of degree $2^{2^{\poly(T)}}$ in the parameters of the circuit which estimates $\bra{m} C\ket{n}$ to within $2^{-2^{T^{O(1)}}}$ precision. 
    \label{thm:degree-bound}
\end{theorem}


\begin{proof}
    Using \cref{prop:energyboundcubic} the energy of computation is $E^\ast = 2^{2^{O(T\log T)}}$. Therefore using \cref{prop:error} we can impose a cutoff $E = 2^{2^{O(T)}}$ such that $\|U_E \ket{0} - U \ket{0}\| = 1/2^{2^{T^{\Omega(1)}}}$. For any POVM $0 \leq \Pi \leq I$. Therefore
    \begin{align}
    \begin{split}
    |\bra{\psi_E} \Pi \ket{\psi_E} - \bra{\psi} \Pi \ket{\psi}| &= |\Tr (\Pi (\ket{\psi}\bra{\psi} - \ket{\psi_E}\bra{\psi_E}))|\\ 
    &\leq \| \ket{\psi}\bra{\psi} - \ket{\psi_E}\bra{\psi_E}\|_1\\
    &= 2 \sqrt{1 - |\braket{\psi_E}{\psi}|^2}= 1/2^{2^{T^{\Omega(1)}}}.
    \end{split}
    \end{align}
    As a result, it is sufficient to compute $\bra{m}U_E\ket{n}$. We use Feynman's path integral to evaluate this quantity. In particular, let $n_0 = n$ and $n_T = m$, then
\begin{align}
            \alpha = \sum_{0 \leq n_1, \ldots n_{T-1} \leq E} \prod_{i=1}^{T} \bra{n_i}U_{i}\ket{n_{i-1}}
            \label{eq:feynman-path}
\end{align}
By definition, each $U_i$ is either cubic or Gaussian. To prove the main result, it is sufficient to show that each Gaussian or cubic phase gate, up to rescaling, can be effectively approximated using a polynomial of at most doubly exponentially large degree. This is achieved in Lemma \ref{lem:cubic+gaussian-degree} below.
\end{proof}

\begin{lemma} [Bounding the degree of Gausian and cubic phase gates]
    Let $m,n \in \mathbb {N}$ and $U$ be a Gaussian gate (rotation, displacement, or squeezing) or a cubic phase gate. Then up to ($O(1)$-computable) normalization  $\bra{m}U\ket{n}$ is computable within high precision by a polynomial of degree $O(m+n)$ in circuit parameters. In particular,
    \begin{enumerate}
        \item $\bra{m}R(\theta)\ket{n}$ is zero if $m \neq n$ and is a polynomial of degree $m$ in $e^{i \theta}$ otherwise.
        \item $e^{|\alpha|^2/2}\bra{n+d}D(\alpha)\ket{n}$ is a polynomial of degree $d + n$ in $\alpha$ and degree $n$ in $\alpha^*$. 
        \item $\sqrt{\cosh(r)} \bra{m} S(r) \ket{n}$ is a polynomial of degree $O(m + n)$ in $\sinh(r), \cosh(r), \tanh (r)$.
        \item For the cubic phase gate $U = e^{i s/3 X^3}$, there exist constants $c, c'$ such that there is a polynomial of degree $d > c(m+n)$ in $s^{-1/3}$ which approximates $e^{-\frac 2{3s^2}} \bra{m}U\ket{n}$ to within additive error at most $1/d^{c' d}$, assuming $|\log s| = O(m+n)$. 
    \end{enumerate}
    \label{lem:cubic+gaussian-degree}
\end{lemma}

\begin{proof}
    The first item is relatively straightforward: $\bra{m} R(\theta) \ket{n} = e^{i\theta n} \delta_{m,n}$ since the Fock basis is an eigenstate of $R(\theta)$. For the displacement operator, we use Equation 1.46a of \cite{ferraro2005gaussian}: $e^{|\alpha|^2/2} \bra{n+d}D(\alpha)\ket{n} = \sqrt{\frac{n!}{(n+d)!}} \alpha^d L^{d}_n (|\alpha|^2)$, where $L^d_n$ is the generalized Laguerre polynomial where $L^d_n (x) = \sum_{i=0}^n (-1)^i \binom{n+d}{n-i} \frac{x^i}{i!}$ is a polynomial of degree $n$. For the squeezing operator, we note:
    \begin{align}
        \begin{split}
            \bra{m} S(r) \ket{n} &= \frac{1}{\sqrt{ n!}} \bra{m} S(r) a^{\dagger n} \ket{0}\\
            &= \frac{1}{\sqrt{ n!}} \bra{m} (S(r) a^{\dagger } S^{\dagger}(r))^n S(r) \ket{0}\\
            &= \frac{1}{\sqrt{ n!}} \bra{m} (\cosh{(r)} a^
            \dagger - \sinh (r) a)^n \ket{S(r)}.
        \end{split}
    \end{align}
Here $\ket{S(r)} := S(r) \ket{0}$ is the squeezed state and is equal to (see e.g. \cite{ferraro2005gaussian}) 
\begin{align}
    \ket{S(r)} = \frac{1}{\sqrt{\cosh(r)}} \sum_{k = 0}^{\infty} (-\tanh(r))^k \frac{\sqrt{(2k)!}}{2^k k} \ket{2k}.
\end{align}
    To evaluate $(\cosh{(r)} a^
            \dagger - \sinh (r) a)^n$ we use the Zassenhaus formula
\begin{align}
    \begin{split}
    \sum_{n\geq 0} \frac{t^n}{n!}(\cosh(r) a^\dagger - \sinh(r) a)^n &=e^{t(\cosh(r) a^\dagger - \sinh(r) a)}\\
    &= e^{t(- \sinh(r) a)}e^{t(\cosh(r) a^\dagger)} e^{\frac{t^2}{2} \sinh(r) \cosh(r)}\\
    &= \sum_{k \geq 0} \frac{t^k}{k!}(- \sinh(r) a)^k \sum_{l \geq 0} \frac{t^l}{l!}( \cosh(r) a^\dagger)^l \sum_{p \geq 0} \frac{t^{2p} (\sinh(r) \cosh(r))^{2p}}{p!}\\
    \end{split}
\end{align}
Equating the $n$'th terms from both sides, we get
\begin{align}
    (\cosh(r) a^\dagger - \sinh(r) a)^n = \sum_{k + l + 2p \leq n} \frac{n!}{k! l! p!} (-1)^k (\sinh(r))^{k+2p} (\cosh(r))^{l+2p} a^k a^{\dagger l} 
\end{align}
    therefore
\begin{align}
    \begin{split}
    \sqrt{\cosh(r)}\bra{m} S(r) \ket{n} =  \sum_{\substack{0\leq k + l + 2p \leq n\\ k+m -l = 2q\\q \geq 0}}A_{k,l,p, n}(-1)^{m-l} (\sinh(r))^{2k+2p + m -l} (\cosh(r))^{2l+2p-m-k} 
    \end{split}
\end{align}
is a polynomial of degree $O(m+n)$ in $\cosh (r)$ and $\sinh(r)$;here 
$$
A_{k,l,p, n} := \frac{\sqrt{(2(m+k-l))!n!/\sqrt{(m)!}}}{2^{m+k-l} (m+k-l)!^{3/2}k! l! p!}
$$

To analyze the cubic phase gate, we truncate Equation 148 of \cite{miatto2020fast} up to $d$ terms. We first show
\begin{lemma} [Truncation of the Airy function]
    Let $Ai : \mathbb{C} \rightarrow \mathbb{C}$ be the Airy function with $Ai (x) = \frac{1}{2\pi} \int_{-\infty}^\infty dk e^{i k^3/3 x + ikx}$, and $Ai^{(l)} (x) = \frac{1}{2\pi} \int_{-\infty}^\infty dk (ik)^l e^{i k^3/3 x + ikx}$ be its $l$'th derivative. Let $c = (|x|^3 + |x|^2)$ and $d \geq 100 (|x|^3 + 1) + l$. Then, there exists a polynomial $\tilde Ai^{(l)}(x)$ of degree $d$ in $x$ such that 
    $$
    |Ai^{(l)}(x) - \tilde Ai^{(l)}(x)| \leq (\frac{c}{d})^d. 
    $$
    \label{lem:airy-truncation}
\end{lemma}

We will apply this lemma to truncate the amplitudes of the cubic phase gate 
\begin{align}
\begin{split}
    \bra{m} U \ket{n} &= \frac{2\sqrt{\pi}s^{-1/3}e^{2/(3s^2)}}{\sqrt{2^{m+n} m! n!}} \sum_{k = 0}^m \sum_{l = 0}^n (-\frac{2i}{s^{1/3}})^{m + n - k - l} \binom nk \binom ml H_k (-i/s)H_l (-i/s) Ai^{(n + m - k - l)} (s^{-4/3})
\end{split}
    \label{eq:cubic-sum}
\end{align}

 We claim that by truncating every Airy function $Ai^(l)$ in $e^{-2/(3s^2)} \bra{m} U \ket{n}$ in Equation \ref{eq:cubic-sum} up to $d = O(m + n)$ according to the above lemma, we obtain a polynomial of degree $O(m + n)$ which approximates the amplitude up to degree $1/d^{O(d)}$. We first observe that the Hermite function $H_k (x)$ is a polynomial of degree $k$ in $x$. Hence, for $d = O(m+n)$, the function resulting from truncating Airy functions is a polynomial of degree $O(m + n)$. To compute the error bound, we need to bound: 
\begin{align}
\begin{split} \delta := \frac{|s|^{-1/3}}{\sqrt{2^{m+n} m! n!}} \sum_{k = 0}^m \sum_{l = 0}^n |\frac{2}{s^{1/3}}|^{m + n - k - l} \binom nk \binom ml |H_k (-i/s)| \cdot |H_l (-i/s)| (\frac{c}{d})^d
\end{split}
    \label{eq:cubic-truncation-error}
\end{align}
To bound the expression above, we first need to bound the magnitude of the complex Hermite polynomials. To do this, we use the following orthonormality relationship from \cite{van1990new, karp2001holomorphic} for any $b > 0$.
$$
\int |H_k (x + i y)|^2 e^{-\frac{b}{b+1} x^2 - by^2} dx dy = \pi\sqrt{\frac{b+1}{b^2}} n! 2^n (1 + 2/b)^n
$$
As a result, for any $b > 0$, the following upper bound holds
$$
|H_k (i y)|^2 \leq e^{by^2} \pi\sqrt{\frac{b+1}{b^2}} n! 2^n (1 + 2/b)^n
$$
In particular, for $b = 2$, the upper bound $|H_n(iy)| \leq 2 e^{y^2} \sqrt{n!} 2^n$ holds. Using this bound 
\begin{align}
    \begin{split}
    \delta &\leq \frac{4|s|^{-1/3}}{\sqrt{2^{m+n} m! n!}} \sum_{k = 0}^m \sum_{l = 0}^n |\frac{2}{s^{1/3}}|^{m + n - k - l} \binom nk \binom ml  e^{2/s^2} \sqrt{l! k!} 2^{l+k} (\frac{c}{d})^d\\       &\leq 2^{O(n+m)\cdot \max\{\log (1/s),1\} +O(1/s^2)} (\frac{c}{d})^d\\        
    \end{split}
\end{align}

Therefore if $|\log s| = O(n+m)$, and $d > \Omega (n+m)$, $\delta < 1/d^{O(d)}$.

\begin{proof}[Proof of Lemma \ref{lem:airy-truncation}]
Using Equation 2.7. \cite{corless1992numerical} the Airy function can be computed as $Ai (x) = Ai(0) f(x) + Ai'(0) g(x)$, where $Ai(0) = \frac{3^{-2/3}}{\Gamma(2/3)}$ and $Ai'(0) = \frac{-3^{-1/3}}{\Gamma(1/3)}$ and
$$
f(x) = \sum_{n=0}^{\infty} \frac{3^n \Gamma (n + 1/3)}{\Gamma(1/3)(3n)!} x^{3n}
$$
and
$$
g(x) = \sum_{n=0}^{\infty} \frac{3^n \Gamma (n + 2/3)}{\Gamma(2/3)(3n+1)!} x^{3n+1}.
$$
As a result the $l$'th derivative of $Ai(x)$ becomes
\begin{align}
    \begin{split}
    Ai^{(l)}(x) &= Ai(0) f^{(l)}(x) + Ai'(0) g^{(l)}(x)\\
    &= \frac{Ai(0)}{\Gamma(1/3)} \sum_{n=\lceil l/3 \rceil}^{\infty} \frac{3^n \Gamma (n + 1/3)}{(3n-l)!} x^{3n-l} + \frac{Ai'(0)}{\Gamma(2/3)} \sum_{n=\lceil l/3 \rceil-1}^{\infty} \frac{3^n \Gamma (n + 2/3)}{(3n+1-l)!} x^{3n+1-l}\\
    \end{split}
\end{align}
We let $\tilde {Ai}^{(l)}$ be equal to the first $d$ terms in the above expression. Therefore for $d \geq l$:
\begin{align}
    \begin{split}
    |Ai^{(l)} (x) - \tilde{Ai}^{(l)}(x)| &\leq (1 + \frac{|x|}{3d-l})   \sum_{n=d}^{\infty} \frac{3^n (n + 1)!}{(3n-l)!} (|x|)^{3n-l}\\
    &\leq \frac{(1 + \frac{|x|}{2d})}{1- \frac{3 |x|^3}{2d}}   \frac{(|x|)^{3d-l}}{(2d/3)^d} 
    \end{split}
\end{align}
In the first line we used the identity $\Gamma(1/3) \Gamma (2/3) = \frac{2\pi}{\sqrt{3}}$ to deduce the coefficients have magnitude $\leq 1$, and the fact that $\Gamma (n + 1/3) \leq \Gamma (n + 1/3) \leq (n+1)!$. In the second line, we used the fact that the summand in the first line decays faster than a geometric series with ratio $\frac{3 |x|^3}{2d}$ and that $d \geq l$. Assuming $d \geq l$ and $d \geq 100 (|x|^3 +1)$ the overall constant factor can be made as small as $1.1$ (and $1 + O(1/d))$ to be more precise); for simplicity, we absorb this constant in the exponent of the exponential following it. If $|x| \geq 1$, we can use the upper bound $(\frac{3|x|^3}{d})^d$ holds and if $|x| < 1$ we can use the upper bound $(\frac{3|x|^2}{d})^d$. Hence, the upper bound expression is in the lemma holds.
\end{proof}

\end{proof}

We now complete the proof of our main \cref{thm:schrod-in-expspace} using the following well-known depth reduction result below.
\begin{theorem}[Depth reduction \cite{valiant1981fast}] Let $p$ be a polynomial of degree $d$ which can be computed using a Boolean circuit of size $s$, then there exists a Boolean circuit of size $(sd)^{O(1)}$ and depth $\log (d) \cdot (\log (s) + \log (d))$ that computes this function.
\end{theorem}

We apply this theorem to \cref{thm:degree-bound} and \cref{thm:eexp-bound} to prove that each amplitude can be strongly simulated in doubly exponential size and exponential depth. In other words, the amplitude can be evaluated using doubly exponentially many processors (see \cref{def:pram}) in exponential time. The \textbf{EXPSPACE} upper bound is implied by \cref{thm:par-pspace}.

\subsubsection{Comparing Schr\"odinger and Heisenberg pictures in discrete and continuous variables}
\label{sec:schrod-vs-heis}

One feature of quantum computation over discrete variables is that, in terms of computational complexity, the Heisenberg and Schr\"odinger pictures lead to equivalent definitions for quantum computations. In particular, in the Schr\"odinger picture, we start with $n$ qubits initiated in $\ket{0}^{\otimes n}$, apply unitary gates one by one to obtain a quantum state like $U |0\rangle^{\otimes n} =:|\psi\rangle = \sum_x \alpha_x |x\rangle$, and then measure each amplitude according to the Borne rule, i.e., obtain string $y$ with probability $|\alpha_y|^2$. Strong simulation corresponds to computing $\alpha_y$ within an exponentially small error. In the Heisenberg picture, we start with a simple observable, such as a Pauli string $P$, and we evolve the observable to obtain $P_U = U P U^\dagger$. Strong simulation corresponds to computing the coefficients of $P_U$ in a suitable such as the Pauli basis, i.e., compute $\gamma^P_Q = Tr (Q P_U)$ given the description of $P$ and $Q$. It is well-known that special sub-universal computations, such as the Clifford computation on stabilizer states, are succinctly and efficiently representable within the Heisenberg picture (leading to efficient classical simulations). The intuition is that under Clifford dynamics, Pauli operators are mapped to Pauli operators, i.e., the $\gamma$ coefficients are $0$ or $1$. For general dynamics, $\gamma$ coefficients take arbitrary values. Here, we first show that in terms of the definition of strong simulations outlined above, the Heisenberg and Schrodinger pictures yield the same computational complexity classes. In particular, let $\mathcal{O}_{\text{Heis}}$ be an oracle $\{0,1\}^* \rightarrow \{0,1\}^*$ which takes a pair of Pauli strings $(P,Q)$ along with the description of a polynomial-size (in the number of qubits) quantum circuit $U$, and outputs $\gamma^P_{Q} := Tr (Q U P U^\dagger)$ within polynomial bits of precision in $n$. Also, let $\mathcal{O}_{\text{Schr\"od}}$ be an oracle $\{0,1\}^* \rightarrow \{0,1\}^*$ which takes the description of $(x,y, U)$, where $x, y$ are $n$ bit strings and $U$ is the description of a polynomial-size quantum circuit and outputs $\alpha_{x,y} := \bra{x}U\ket{y}$ within polynomial bits of precision.

\begin{theorem}[Computational equivalence of Schr\"odinger and Heisenberg pictures] For a discrete variable formulation of quantum computing, a $\#$\textbf{P} simulation of $\mathcal{O}_{\text{Heis}}$ implies a $\#$\textbf{P} simulation of $\mathcal{O}_{\text{Schr\"od}}$, and vice versa.
\label{thm:schor-heis-dv}
\end{theorem}

\begin{proof} We first show that $\mathcal{O}_{\text{Schr\"od}} \subseteq \#\textbf{P}^{\mathcal{O}_{\text{Heis}}}$. This implies that a $\# \textbf{P}$ simulation of ${\mathcal{O}_{\text{Heis}}}$ implies a $\# \textbf{P}$ simulation of $\mathcal{O}_{\text{Schr\"od}}$. Without loss of generality, we simulate the instance $(0,0,U)$, since for any unitary $V$ the instance $(x,y, V)$ is equivalent to $(0,0, X^x V X^y)$. The projector onto $|\psi\rangle = U \ket {0}$ can be computed as $U (|0\rangle \langle 0|)^{\otimes n}U^\dagger = \sum_{P \in \{I, Z\}^n} U P U^\dagger$. To see this, we observe that the amplitude $|\alpha_{0,0}|^2 = \sum_{P,Q \in \{I,Z\}^n} \gamma^P_{Q} (-1)^{|Q|}$, $|Q|$ is the Hamming weight of the string describing $Q$ (by assigning $I$ to $0$ and $Z$ to $1$). A $\#$\textbf{P} computation can evaluate this exponential sum. The phase component of $\alpha_{0,0}$ can be computed within the same computational model by e.g. performing the Hadamard test on real and imaginary parts of the Unitary.  Next, we show that $\mathcal{O}_{\text{Heis}} \subseteq \#\textbf{P}^{\mathcal{O}_{\text{Schr\"od}}}$. Consider the instance $(P,Q, U)$. For a Pauli string $P$, let $P_{x,y} = \bra{x}P\ket{y}$. Then $\gamma^P_Q = \sum_{x,y,w,z} Q_{w,z} P_{x,y} \alpha_{w,x} \alpha^*_{z,y}$ which is an exponential sum that can be evaluated using $\#\textbf{P}$ having access to $\mathcal{O}_{\text{Schr\"od}}$.
\end{proof}

Does this equivalence result hold in the continuous variable domain? In the remainder of this Section, we give evidence that the answer is probably no. First, let us clarify the definitions for Schr\"odinger and Heisenberg pictures in the continuous quantum computational domains. We consider the cubic and Gaussian gate set. In the Schr\"odinger picture, we start by $m$ modes initialized in vacuum, i.e., $\ket{0}^{\otimes m}$. We apply a polynomial-size quantum circuit $U$ consisting of cubic and Gaussian gates to obtain $\ket{\psi} = I \ket{0}^{\otimes m}$. Strong simulation corresponds to computing $\braket{\mathbf{n}}{\psi}$, where $\mathbf{n} \in \mathbb{N}^m$. Let $\mathcal{O}_{\text{Schr\"od}}: ( \mathbf{n}, U) \mapsto \alpha_{\mathbf{n}}$ (within polynomial bits of precision) be the corresponding oracle. In the Heisenberg picture we start with an observable $X^{\boldsymbol{\mu}} P^{\boldsymbol{\nu}}$, evolve to $U X^{\boldsymbol{\mu}} P^{\boldsymbol{\nu}} U^\dagger = \sum_{\boldsymbol{\alpha}, \boldsymbol{\beta}} \gamma^{\boldsymbol{\mu}, \boldsymbol{\nu}}_{\boldsymbol{\alpha}, \boldsymbol{\beta}} X^{\boldsymbol{\mu}} P^{\boldsymbol{\beta}}$ (these coefficients are unique due to linear independence of the Heisenberg---polynomials in $X$ and P---basis). Strong simulation corresponds to computing $\gamma^{\boldsymbol{\mu}, \boldsymbol{\nu}}_{\boldsymbol{\alpha}, \boldsymbol{\beta}}$ up to polynomial bits of precision. Let $\mathcal{O}_{\text{Schr\"od}}: (\boldsymbol{\mu, \nu, \alpha, \beta}, U) \mapsto \gamma^{\boldsymbol{\mu}, \boldsymbol{\nu}}_{\boldsymbol{\alpha}, \boldsymbol{\beta}}$ (within polynomial bits of precision) be the corresponding oracle. We conjecture
\begin{conjecture} [Inequivalence between Heisenberg and Schr\"odinger pictures in continuous domains] In terms of strong simulation of quantum circuits consisting of cubic and Gaussian gates, there are computational problems that are solvable by the Schr\"odinger picture and not Heisenberg even with access to a $\#$\textbf{P} machine. In other words
    $\mathcal{O}_{\text{Schr\"od}} \subsetneq \#\textbf{P}^{\mathcal{O}_{\text{Heis}}}$.
    \label{conj:schrod-heis}
\end{conjecture}
In the remainder of this section, we justify why this conjecture is plausible. First of all, while in \cref{sec:expspace}, we show that to simulate cubic and Gaussian circuits in the Schr\"odinger picture, one needs to simulate a ``doubly exponentially'' large dimensional Hilbert space  (in the number of cubic phase gates), in \cref{sec:expspace}, we show that for the same task in the Heisenberg picture, we need to keep track of vectors in vector spaces that are only ``exponentially larger.'' One basis for \cref{conj:schrod-heis} is this discrepancy between ``exponential'' and ``doubly exponential'' dimensionalities. Using standard depth reduction techniques, we prove that the simulation in the Schr\"odinger and Heisenberg pictures can be done in \textbf{PSPACE} and \textbf{EXPSPACE}, respectively. Of course, to justify this \cref{conj:schrod-heis} more rigorously, we need to prove \textbf{EXPSPACE} lower bound on strong simulation of the Schr\"odinger picture which we leave for future work. 

Our second justification for \cref{conj:schrod-heis} is by observing that the approach used in the proof of \cref{thm:schor-heis-dv} fails in proving equivalence between the two pictures due to infinite dimension Hilbert spaces. Consider the following decomposition (See Equation 18 of \cite{paris1996quantum})
\begin{equation} 
\ket{0}\!\bra{0} = \lim_{\epsilon \rightarrow 1^-} \sum_{p \geq 0} \frac{1}{p!} (-\epsilon)^p a^{\dagger p} a^p.
\end{equation} 
As a result 
\begin{align}
\begin{split}
    \alpha &= \lim_{\epsilon \rightarrow 1^-} \sum_{p \geq 0} \frac{1}{p!} (-\epsilon)^p \bra{0}U a^{\dagger p} a^p U^\dagger\ket{0}.\\
    &=: \lim_{\epsilon\rightarrow 1^-} \sum_{p \geq 0} A_p(\epsilon)
\end{split}
\end{align}
We ask, ``can we somehow truncate the sum $\sum_{p \geq 0} A_p(\epsilon)$ to exponentially many terms to approximate $\alpha$?'' If so, then if the rate of convergence is quick enough, we can use $\mathcal{O}_{\text{Heis}}$ to estimate $\alpha$ in \textbf{PSPACE}. In the following, we explain this series diverges for the cubic phase gate, giving evidence that the reduction approach discussed in the proof of \cref{thm:schor-heis-dv} fails.

Before describing the cubic phase gate situation, we prove that finite truncation is possible for Gaussian gates. We need to show this for the three elementary Gaussian operations: rotation, squeezing, and displacement. Rotation is easy to analyze because $ R  a R^\dagger = e^{i\theta}  a$ for $R = e^{i \theta  a^\dagger  a}$. Therefore, we only need to keep the first term ($\alpha = 1$). For displacement, $D = e^{\delta  a - \delta^* a^\dagger}$ and $D a^{\dagger p} D^\dagger = (a^\dagger - \delta^*)^p$, as a result, $A_p (\epsilon) = \frac{(-\epsilon |\delta^2|)^p}{p!}$ which tends to zero. Lastly, we study squeezing $ S = e^{1/2 r (a^2 - a^{\dagger 2})}$. We can evaluate $S \ket{0} = \frac{1}{\sqrt{\cosh{r}}} \sum_{n \geq 0} \frac{\sqrt{(2n)!}}{2^n n!} \tanh^n{r} \ket{2n}$ (see for example (33) in \cite{weedbrook2012gaussian}). Therefore $a^p S \ket{0} = \frac{1}{\sqrt{\cosh{r}}} \sum_{2n \geq p} \frac{\sqrt{(2n)!}}{2^n n!} \tanh^n{r} \sqrt{\frac{(2n)!}{(2n-p)!}}\ket{2n-p}$. Therefore
\begin{equation} 
A_p(\epsilon) = \frac{1}{\cosh{r}} 
\sum_{2n \geq p}
\binom{2n}{n} \frac{1}{4^n} \tanh^{2n}{r} 
\binom{2n}{p} (-\epsilon)^p
\end{equation} 
Therefore
\begin{align}
\begin{split}
    \sum_{p \geq 0} A_p(\epsilon) &= \sum_{n \geq 0} \frac{1}{\cosh{r}} \binom{2n}{n} \frac{1}{4^n} \tanh^{2n}{r}  \sum_{p \leq 2n} \binom{2n}{p} (-\epsilon)^p\\
    &= \sum_{n \geq 0} \frac{1}{\cosh{r}} \binom{2n}{n} \frac{1}{4^n} \tanh^{2n}{r}  (1-\epsilon)^n\\
\end{split}
\end{align}
Therefore $\lim_{\epsilon \rightarrow 1^-} \sum_{p \geq 0} A_p(\epsilon) = \frac{1}{\cosh{r}}$, i.e., only the first term survives.

Now, let us study the cubic phase gate. We explain that the above expansion does not converge for the cubic phase gate. For simplicity, assume we want to compute $\alpha_1 := \bra{0}e^{i s X^3/3} \ket{0}$. We consider the expansion $\alpha_1 (l) := \sum_{j = 0}^{2l} \frac{(is/3)^j}{j!}\bra{0} X^{3j} \ket{0}$. Basic calculation indicates for even $l$, $\bra{0} X^{3j} \ket{0} = 2^{-3j+1} \frac{(3j-1)!}{(3/2j-1)!}$ (for odd $l$ we get zero). Therefore $\alpha_1 (l) = \sum_{j = 0}^{l} \frac{(-1)^{j} s^{2j}}{(2j)!} 2^{-6j+1} \frac{(6j-1)!}{(3j-1)!}$. We can show that the magnitude of this sum is dominated by the magnitude of the last term, which diverges like $l^l$ and hence $\infty =\lim_{l \rightarrow \infty} |\alpha_1 (l)| \neq |\alpha_1|$. To understand the root source of the difficulty, we write $\alpha_1$ in the position eigenbasis:
$$
\alpha_1 (s) = \int_{-\infty}^{+\infty} \frac{e^{-x^2 + is x^3/3}}{\sqrt{\pi}} dx
$$
performing Taylor expansion for $e^{isX^3}$ is similar to evaluating 
$$
\alpha_1 (s) = \int_{-\infty}^{+\infty} \frac{e^{-x^2}}{\sqrt{\pi}} \sum_{j=0}^\infty \frac{(isx^3/3)^j}{j!} dx
$$
The above expression is convergent and correctly takes a value between $-1$ and $+1$; however, we cannot change the order of the sum and integral (because, e.g., it does not satisfy the conditions of the Fubini-Tonelli theorem). It turns out that the solution to the integral above is given by the Airy function (See, for example, C.1. of \cite{miatto2020fast}): 
$$
\alpha_1(s) = 2 \sqrt{\pi} e^{2/3 y^6} |y| Ai (y^4), \text{ where } y^3 = s.
$$

\subsection{A parallel algorithm for computing expectation values}
\label{sec:PARalgorithm}

In this section, we focus on the problem of computing expectation values of position and momentum operators at the output of $\mathbf{CVBQP}[X^3]$ circuits. 

\begin{definition}[Expectation value computation for bosonic circuits with fixed non-Gaussian gates]\label{def:ExpValfixednGprob}
    Given a polynomial-time uniform family $\{\mathcal C_n[H_1,\dots,H_m]\}_{n\in\mathbb N}$ of bosonic circuits of polynomial circuit complexity with non-Gaussian gates $H_1,\dots,H_m\in\mathcal P_{O(1),d}$ with $d>2$, a measurement operator $H\in\mathcal P_1$, $\mathsf{ExpValBosonCircuit}[H_1,\dots,H_m;H]$ is the problem of computing the sign of $\langle x|C_n^\dag HC_n|x\rangle$ on input $x\in\{0,1\}^n$.
\end{definition}
 
For cubic phase gates, we show that this problem of computing the sign of an expectation value is in $\mathbf{PSPACE}$. This is in contrast with the problem of estimating the probability of the sign of an outcome, which by \cref{prop:CVBQPfixednGcompleteprob} and \cref{thm:schrod-in-expspace} is in $\mathbf{EEXP}$.

\begin{theorem}\label{th:expvalinPSPACE} We have:
\begin{equation}
    \mathsf{ExpValBosonCircuit}[X^3;X]\in\mathbf{PSPACE}.
\end{equation}
\end{theorem}

\noindent The proof of this result is based on a parallel algorithm for computing the evolution of position and momentum operators under bosonic computations (see \ref{sec:parallelComplexity} for a brief review of parallel computations).
 
Before presenting the proof, we recall a useful technical result.
The Zassenhaus formula (see \cref{lem:zassenhaus} or \cite{zaimi2011binomial}) allows us to make some assertions about the normal form of operators that result from using the standard gate set. 
In particular,
\begin{proposition}\label{prop:normalOrdering}
    For all $n \in \mathbb N$ the following product of binomial operators in the position and momentum variables $X$ and $P$, has a normal form (as in \cref{def:normalForm}) expansion:
    \begin{equation}
        (\alpha X + \beta P)^n =  \sum_{k=0}^{\lfloor n/2 \rfloor} \sum_{j=0}^{n-2k} c_{j,(n-2k)-j} X^j P^{(n-2k)-j},
    \end{equation}
    where
        \begin{align*}
            c_{j,(n-2k)-j}  &= \frac{n!}{k!j!(n-2k-j)!} \left(\frac{-i}{2}\right)^k (\alpha^{k+j} \beta^{n-k-j}).
        \end{align*}
\end{proposition}
\begin{proof}
    Clearly $[\beta P, \alpha X] = (-i\alpha\beta)I$ satisfies the condition $[\beta P, (-i\alpha\beta)I] = 0 = [\alpha X, (-i \alpha\beta)I]$. Applying \cref{cor:closedFormBinomialExpansion} we get that,
    \begin{align*}
        (\alpha X + \beta P)^n &= \sum_{k=0}^{\lfloor \frac{n}{2} \rfloor} \frac{n!}{(n-2k)!k!2^k} B_{n-2k}(\alpha X, \beta P) (-i\alpha \beta)^k \\
        &= \sum_{k=0}^{\lfloor \frac{n}{2} \rfloor} \frac{n!}{(n-2k)!k!} \left(\frac{-i\alpha\beta}{2} \right)^k \sum_{j=0}^{n-2k} \binom{n-2k}{j} (\alpha X)^j (\beta P)^{n-2k-j} \\
        &= \sum_{k=0}^{\lfloor \frac{n}{2} \rfloor} \sum_{j=0}^{n-2k} \frac{n!}{k!j!(n-2k-j)!} \left(\frac{-i}{2}\right)^k (\alpha^{k+j} \beta^{n-k-j}) X^j P^{n-2k-j},
    \end{align*}
    which is in normal form. Moreover, since $B_{n-2k}$ is homogeneous of degree $n-2k$ for each fixed $k$ we get the closed form formulas above for each coefficient. Notice that whenever $n$ is even, the only non-zero terms have even total degree, and when $n$ is odd, the non-zero terms must have odd total degree.
\end{proof}

Observe that the coefficients for the monomials in the result above can be obtained from the coefficients of the polynomial $(X+P)^n$, by appropriately scaling with $\alpha$ and $\beta$. We should be able to obtain a similar expression for the product operator $(\alpha X + \beta P^2)$.

\begin{proposition}\label{prop:normalOrdering2}
    $(\alpha X + \beta P^2)^n$ has the following normal form expression
      \begin{align*}
         \frac{(\alpha X + \beta P^2)^n}{n !} 
        &= \sum_{j+k+2l+3m=n} \frac{(\alpha^{j + l + 2m} \beta^{k + l + m}) (-1)^l(i)^{m-l}}{3^m\cdot j! k! l! m!} X^j P^{2k+l}.
    \end{align*}
    Furthermore, the coefficients are computable in parallel polynomial depth in $n$.
\end{proposition}
\begin{proof}
    Consider that $[X,P^2]=2iP$ and so 
    \begin{equation}
        [[[X,P^2],X],X] = [[[X,P^2],X],P^2] = [[[X,P^2],P^2],P^2]=0.
    \end{equation}
    Thus, we can truncate the Zassenhaus formula to third-order commutators and we have,
    \begin{equation}
        e^{t(\alpha X + \beta P^2)} = e^{t(\alpha X)}e^{t(\beta P^2)}e^{-t^2(2\alpha\beta i)P/2} e^{t^3(2\alpha^2 \beta)iI/6} = e^{t(\alpha X)} e^{t(\beta P^2)} e^{-t^2(\alpha\beta i P)} e^{t^3(\alpha^2 \beta i I)/3},
    \end{equation}
    where the left-hand side is the generating function for $(\alpha X + \beta P^2)^n$ and the right-hand side contains terms of the form $X^j P^{2k + l}$. 
    \begin{equation}
        \sum_\lambda \frac{(\alpha X + \beta P^2)^\lambda t^\lambda}{\lambda !} 
        = \sum_{j,k,l,m} t^{j+k+2l+3m} \frac{(\alpha X)^j (\beta P^2)^k (-\alpha \beta i P)^l (\alpha^2 \beta i I/3)^m}{j!k!l!m!}.
    \end{equation}
    Comparing coefficients for the $t^n$ terms we have,
    \begin{align*}
         \frac{(\alpha X + \beta P^2)^n}{n !} 
        &= \sum_{j+k+2l+3m=n} \frac{(\alpha^{j + l + 2m} \beta^{k + l + m}) (-1)^l(i)^{m-l}}{3^m\cdot j! k! l! m!} X^j P^{2k+l}.
    \end{align*}
    To see why the coefficients are computable in parallel depth we note that factorials (and other arithmetic operations involved) are easy for parallel random access machines, and other models of parallel computation. 
\end{proof}

The goal is to show that the problem of determining the sign of a coefficient of the monomial $X$ in the normal form of the expansion of a continuous-variable quantum circuit with polynomially-sized coefficients can be solved by a classical parallel algorithm. Each gate in a continuous-variable quantum circuit transforms the position and momentum variables by means of a substitution rule. In particular, as we have seen in \cref{sec:elementary_calculations}, unitary evolution of the conjugate variables can produce polynomial operators of higher degrees, including product terms of the form $(\sum_{i,j}c_{i,j}X^iP^j)^k$. This suggests that after polynomially many gates, there may be exponentially many monomial terms, even if they are in normal form. The exponential width of the polynomial expansion, together with the polynomial number of gates, suggests that a parallel algorithm might succeed for such a task.

The idea of the algorithm is to attach a processor to each possible monomial term, keeping track separately of the real and imaginary parts of its coefficient in the shared memory of the PRAM. At each time step $t$, each processor corresponding to a non-zero coefficient is active and will perform the substitution rule corresponding to the gate $U_t$ in its local memory and then write send messages to all processors corresponding to the monomials produced by the substitution. In the next time step, all machines will read their messages and update the coefficients in the global memory. If a monomial term goes from zero to non-zero, then the corresponding processor is activated for the next time step. After $T$ time steps, the sign of (the real part of) the coefficient of the monomial $X$ in the normal form of the expansion is returned as the output.

\begin{proposition}\label{prop:CVOperatorEvolutionIsInPSPACE}
    There is a parallel polynomial-time algorithm using exponentially many processors, that takes the description of a continuous-variable quantum circuit $U = U_T U_{T-1} \cdots U_1$ composed of Gaussian and cubic phase gates (with coefficients of size at most $\mathsf{poly}(T)$) over one mode and outputs the coefficients corresponding to  monomials $X^mP^n$ in the normal form expansion of $X$ transformed by the unitary bosonic circuit $U$, 
    \begin{equation}U X U^\dagger = \sum_{m,n} c_{m,n} X^mP^n = c_{1,0} X + c_{0,1} P + \cdots \end{equation}
\end{proposition}
\begin{proof}
The algorithm follows the pseudocode given below (\cref{alg:parallelCoeffBosonCircuit}).
\begin{algorithm}
\caption{A parallel algorithm for calculating the normal form expansion of X} 
\begin{algorithmic}[1]
\State \textbf{Input:} Ordered list $x = (g_i, t_i)_{i \in [T]}$ encoding quantum circuit
\State \textbf{Output:} The non-zero coefficients $\eta_{k,l}=\eta_{k,l}^{(r)} + \eta_{k,l}^{(i)}$
\State \textbf{Initialization:}
\State Store in the global memory the elementary substitution operations from \cref{sec:elementary_calculations}
\State Precompute coefficients for the expressions from \cref{prop:normalOrdering} with $\alpha=\beta=1$ \label{step:precomputeCoeff}
\State Allocate processor $M_{k, l}$ to the monomial $X^k P^l$
\State Allocate registers in the global memory to the real part $\eta_{k, l}^{(r)}$ and imaginary part $\eta_{k, l}^{(i)}$ for the coefficient of each monomial $X^k P^l$
\State Activate machines $M_X \equiv M_{1,0}$ and $M_P \equiv M_{0,1}$
\State Set $\eta_{1, 0}^{(r)} \gets 1$
\For{$k,l$}
\State Set $\eta_{k,l}^{(r)}\gets 0$ and $\eta_{k,l}^{(i)} \gets 0$
\EndFor        
\State \textbf{Processing:}
\For{$t = 1$ to $T$}
\For{each active machine $M_{k, l}$}
\State $g_i$ transforms $X^k P^l$ according to \cref{sec:elementary_calculations}
\State $M_{k,l}$ computes the coefficients of the polynomial in normal form by looking up the precomputed values from step \ref{step:precomputeCoeff} and scaling with $\alpha$ and $\beta$ 
\State In local memory, 
\State Add the results to all affected global registers $\eta_{k',l'}$
\EndFor
\For{each machine $M_{k, l}$}
\If{coefficient is zero}
\State Deactivate machine $M_{k, l}$
\ElsIf{coefficient is non-zero}
\State Activate machine $M_{k, l}$
\EndIf
\EndFor
\EndFor
\end{algorithmic}\label{alg:parallelCoeffBosonCircuit}
\end{algorithm}
    The input to the algorithm is an ordered list encoding a quantum circuit on one bosonic mode using $T$ unitary gates from the set $S:=\set{e^{it_1X}, e^{it_2X^2}, e^{i\frac{\pi}{4}(X^2 + P^2)},e^{it_3X^3}}$. Let $x:=(g_i,t_i)_{i\in[T]}$ be this input where $g_i$ is an encoding for a gate in $S$ and $t_i \in \mathbb Q$ is the rational parameter with which it is applied. We will write $(g_i,t_i)=(X,t)$ if the unitary applied is $e^{itX}$, and so on. If $g_i$ is a Fourier gate, we exclusively use the parameter $\pi/4$ which is considered to be part of the specification of the gate and $t_i=1$. 

    From \cref{sec:elementary_calculations} we know that the evolution of the position and momentum operators under the action of the gates can be given by a set of replacement rules for the non-commuting formal variables $X$ and $P$. From \cref{prop:normalOrdering} and \cref{prop:normalOrdering2}, compositions of these rules result in \textit{higher-order} replacement rules that can be put in normal form. The general form of these higher-order rules are precomputed using the closed-form formulas for the coefficients and stored, so that they are available to the processors during simulation.

    Allocate to each non-commutative monomial of the form $X^k P^l$ where $k, l \in \mathbb N$, a pair of registers $\smash{\eta_{k, l}^{(r)}}$ and $\smash{\eta_{k, l}^{(i)}}$ in the global memory which store the real and imaginary parts, respectively, of its current coefficient at time $t$. At time $t=0$, set all coefficients to $0$, except the ones corresponding to $X$ which are set to $1$. Now, to each coefficient allocate a processor $M_{k, l}$ where $(k, l)$ is a label for a monomial in terms of its powers. Each processor has access to the input, and thus knows which unitary gate acts in each time-step. 
    
    In the first step, only the machines $M_X\equiv M_{1,0}$ and $M_P \equiv M_{0,1}$ are \textit{active}. For each step $t$, every active machine $M_{k,l}$ performs the substitutions corresponding to the gates specified in the input and stores the results in its local memory. It then sends a message to the machines $M_{k',l'}$ if the $(k',l')$-th monomial term is produced by substitution, along with the corresponding coefficient. If a monomial coefficient is equal to zero after all computations are complete in a given time-step, deactivate the respective machine. If it goes from zero to non-zero, then activate the respective machine for the next time-step.

    \paragraph{Complexity analysis:} First, consider that in the precomputation step the coefficients of the terms produced in \cref{prop:normalOrdering} and \cref{prop:normalOrdering2} contain factorials of size at most $n!$. On a PRAM[*] (with multiplication as a unit cost primitive operation), there is a $O(\log n)$ step procedure to calculate each such factorial: Each of $\lceil n/2 \rceil$ processors multiplies two consecutive numbers $a$ and $a+1$ in the first step and writes the answer to the global memory, repeat with half as many processors for each step until the last processor returns $n!$. Each coefficient contains a constant number of factorial terms and there are polynomially-many such coefficients for a single expansion, so there is a PRAM[*]-procedure using $\mathsf{poly}(n)$ processors to compute the list of coefficients for each such binomial expansion with cost $O(\log (n))$.

    Next, in each processing step the processor for $X^kP^l$ performs a substitution based on the elementary transformations in \cref{sec:elementary_calculations}, by looking up the transformations corresponding to the gate $g_i$ with parameter $t_i$. $M_{k,l}$ may then need to use \cref{prop:normalOrdering} and \cref{prop:normalOrdering2} to calculate the normal ordered resultant, and thus calculate coefficients as above. There are $T$ gates and thus, the algorithm takes at most $\mathsf{poly}(T)$ parallel time. 
\end{proof}
Notice that the ideas of the algorithm above work just as well for the complementary basis of annihilation and creation operators $\set{ \hat a,  \hat a^\dagger}$, rather than the position and momentum operators $\set{X,P}$. The replacement rules of \cref{sec:elementary_calculations} must be modified to operate on the annihilation and creation operators $\set{ \hat a,  \hat a^\dagger}$, so that we can arrive at equivalent expansions with terms of the (normal) form $ \hat a^m ( \hat a^\dagger)^n$. 
\begin{align}
    {a} &\xmapsto{(X,t)} {a} - \frac{it}{\sqrt{2}} I
    &{a}^\dagger &\xmapsto{(X,t)} {a}^\dagger + \frac{it}{\sqrt{2}} I \nonumber \\
    {a} &\xmapsto{(X^2,t)} {a} (I - it) - (it) a^\dagger
    &{a}^\dagger &\xmapsto{(X^2,t)} {a}^\dagger (I + it) + (it) a \nonumber \\
    {a} &\xmapsto{(X^2+P^2,\pi/4)} (-i)  a
    &{a}^\dagger &\xmapsto{(X^2+P^2,\pi/4)} i  a^\dagger \nonumber \\
    {a} &\xmapsto{(X^3,t)}  a - (2\sqrt 2 it)X^2
    &{a}^\dagger &\xmapsto{(X^3,t)}  a^\dagger + (2\sqrt 2 it)X^2
\end{align}
Another generalization we can make to the prior result, is to replace $X$ with a generic monomial operator $O$ of the form $a^m ( a^\dagger)^n$. A particularly nice property of monomial operators of this form is that their expectations with respect to the Fock basis have a simple closed form. For example,
\begin{equation}
    \bra{0} a^m ( a^\dagger)^n\ket{0} = \sqrt{m!n!}\delta_{m,n}.
\end{equation}
Or more generally for $k \in \mathbb{N}$,
\begin{equation}
    \bra{k} a^m ( a^\dagger)^n\ket{k} = \frac{\sqrt{(m+k)!(n+k)!}}{k!}\delta_{m,n}.
\end{equation}

Hence, \cref{alg:parallelCoeffBosonCircuit} can be extended to compute expectations of operators that are arbitrary monomials in $\set{a,a^\dagger}$, by also precomputing the expectations of monomial operators up to exponential degree (which we can do since factorials are easy for parallel random-access machines with unit-cost multiplication), and then multiplying with the amplitudes produced at the end of the prior algorithm.

\begin{proposition}\label{prop:ComputingVacExpIsInPSPACE}
    There is a parallel polynomial-time algorithm using exponentially many processors, that takes the description of a continuous-variable quantum circuit $U = U_T U_{T-1} \cdots U_1$ composed of Gaussian and cubic phase gates (with coefficients of size at most $\mathsf{poly}(T)$) over one mode, an input $x\in \set{0,1}$ and a constant-degree operator $O$ of the form $ a^\mu ( a^\dagger)^\nu$ and outputs the corresponding expectation of $O$, 
    \begin{equation}\bra{x} U O U^\dagger \ket{x} = \sum_{m,n} c_{m,n} \bra{x}  a^m ( a^\dagger)^n \ket{x} = \begin{cases}
    \sum_{m} c_{m,m} m! \text{, if $x=0$} \\
    \sum_{m} x_{m,m} (m+1)! \text{, if $x=1$.}
    \end{cases}
    \end{equation}
\end{proposition}

Obviously, we can determine the sign of the expectation $\bra{x} U O U^\dagger \ket{x}$ from the above. Altogether these observations allow us to conclude that \cref{prop:CVOperatorEvolutionIsInPSPACE} together with \cref{thm:PRAMtimeEqualsPSPACE} gives us an upper-bound of $\textbf{PSPACE}$ for the decision problem of determining the sign of the expectation value of an observable $O$, and further our algorithm allows us to calculate this expectation directly. The next step is to expand this result to $m$ bosonic modes.

\paragraph{The multimode case:}
Consider now the question of $m$-continuous modes where $m>1$. We will approach this more general problem by extending the tools we have previously used. Firstly, we will augment the standard Gaussian gate set $\set{e^{itX}, e^{itX^2}, e^{i\frac{\pi}{4}(X^2+P^2)}}$ we have used for the single-mode version of the problem to $m$ continuous modes, by adding in the multimode SUM gate, which is the continuous-variable analog to the discrete-variable CNOT gate. In particular we define\footnote{In the rest of this section, we will simplify our notation by using $X_j P_k$ to represent the tensor product of operators $X_j \otimes P_k$, where the subscripts $j,k$ are natural number indices for the modes.},
\begin{equation}
    SUM_{j,k} := e^{-iX_j\otimes P_k}.
\end{equation} 
The multimode Gaussian gate set including for all modes $1\leq j\leq m$ the position-displacement operator $e^{itX_j}$, quadratic phase gate $e^{itX_j^2}$, Fourier transform $e^{i\frac{\pi}{4}(X_j^2+P_j^2)}$ and the SUM gates $e^{iX_jP_{j+1}}$ was presented by Bartlett et al.\ \cite{bartlett2002efficient}, and shown to generate the multimode group of Clifford (Gaussian) unitaries.

Next, notice that the $SUM_{j,k}$ gate has the following commutator relations,
\begin{align}
    [X_jP_k, X_j I_k] &= 0, &[X_jP_k, I_j P_k] =& 0 \nonumber \\
    [X_jP_k, P_j I_k] &= i I_j P_k, &[X_jP_k, I_j X_k] =& -i X_j I_k \nonumber \\
    [X_jP_k, X_j P_k] &= 0, &[X_jP_k, X_j X_k] =& -i X_j^2 I_k \nonumber \\
    [X_jP_k, P_j P_k] &= i I_j P_k^2, &[X_jP_k, X_jP_j] =& i X_jP_k \nonumber \\
    [X_jP_k, X_kP_k] &= -i X_jP_k, &[X_jP_k, P_j X_k] &= i I_j (X_kP_k) - i (X_j P_j) I_k
\end{align}
where $I_j$ and $I_k$ are the identity operators on the $j$-th and $k$-th mode respectively. This gives rise to the following quadratic replacement rules,
\begin{align}\label{eqn:two-mode-replacement-rules}
X_j I_k &\xmapsto{(X_j P_k)} X_j I_k
& P_j I_k &\xmapsto{(X_j P_k)} P_j I_k - I_jP_k \nonumber\\
I_j X_k &\xmapsto{(X_j P_k)}I_j X_k + X_j I_k 
& I_j P_k &\xmapsto{(X_j P_k)} I_j P_k \nonumber\\
X_j P_k &\xmapsto{(X_j P_k)} X_j P_k 
& X_j X_k &\xmapsto{(X_j P_k)} X_j X_k + X_j^2 \nonumber\\
P_j P_k &\xmapsto{(X_j P_k)} P_j P_k - P_j^2 
& X_j P_j &\xmapsto{(X_j P_k)} X_j P_j - X_jP_k \nonumber\\
X_k P_k &\xmapsto{(X_j P_k)} X_k P_k + X_jP_k 
& P_j X_k &\xmapsto{(X_j P_k)} P_j X_k - X_k P_k + X_j P_j - X_jP_k.
\end{align}
Using the properties of tensor products, we can see that the first four rules are sufficient, but we provide a few more calculations as examples.

With the formal substitutions above in addition to those for the single-mode case, the idea for a parallel algorithm for the multimode case is straightforward. The single-mode gates are handled just as before, while the two-mode $SUM$ gates follow the rules computed above.

\begin{proposition}\label{prop:mCVOperatorEvolutionIsInPSPACE}
    There is a parallel polynomial-time algorithm using exponentially many processors, that takes the description of a continuous-variable quantum circuit $U = U_T U_{T-1} \cdots U_1$ composed of Gaussian and  gates (with coefficients of size at most $\mathsf{poly}(T)$) over $m$ modes and outputs the coefficients corresponding to the monomial $X_1$ according to the normal form of the expansion, 
    \begin{equation}U X_1 U^\dagger = \sum_{j,k} \sum_{m,n} c_{m,n} X_j^m P_k^n = c_{1,0}^1 X_1 + c_{0,1}^1 P_1 + \cdots \end{equation}
\end{proposition}
\begin{proof}
    We will allocate exponentially many processors to each of the $m$ modes as before. Whenever a single-mode unitary gate appears in the description of the continuous-variable quantum circuit, we can directly use \cref{alg:parallelCoeffBosonCircuit} to update the set of global variables corresponding to the coefficients of the transformed elementary variables for that mode. If a $SUM(j,k)$ gate is introduced, we will instead use the replacement rules described in Eq.~(\ref{eqn:two-mode-replacement-rules}) to update the coefficients corresponding to the operators on modes $j$ and $k$.

    Notice that the coefficients of each mode's observables can be maintained by exponentially many processors, and so we can deal with polynomially many modes. Moreover, the depth of the computation is not significantly affected by $SUM$ gates, so the total depth is still polynomial. This means that the entire algorithm works in parallel polynomial time.
\end{proof}
Combining these results shows that $\mathsf{ExpValBosonCircuit}[X^3;X] \in \mathbf{PSPACE}$, and in fact for any constant-degree monomial operator $O$, $\mathsf{ExpValBosonCircuit}[X^3;O] \in \mathbf{PSPACE}$.

\section{Ground energy of bosonic Hamiltonians}
\label{sec:gsp}

The problem of interest in this section is to understand the complexity of determining the ground energy of a local bosonic Hamiltonian that is a polynomial in the position and momentum operators.
The finite-dimensional version of this problem (the so-called Local Hamiltonian problem) is very well-studied, and is characterized by the complexity class $\mathbf{QMA}$, for which estimating the ground energy of a local Hamiltonian is a complete problem. We define a version of this class in the bosonic setting, based on the complexity class introduced in \cref{def:CVBQPfixednG}.

\begin{definition}[Continuous-variable quantum Merlin Arthur]\label{def:CVQMA}
    $\mathbf{CVQMA}[H_1,\dots,H_m](c,s)$ is the class of languages $L \subseteq \{0,1\}^*$ for which there exists a polynomial-time uniform family $\{\mathcal V_n\}_{n\in\mathbb N}$ of $\mathbf{CVBQP}[H_1,\dots,H_m]$ verifiers such that:
    \begin{itemize}
        \item for all $x\in L$, there exists a quantum state $\ket\psi$ over at most $\mathsf{poly(|x|)}$ modes such that $V_{|x|}(x,\ket\psi)$ accepts with probability greater than $c$;
        \item for all $x \notin L$, for all quantum states $\ket\psi$ over at most $\mathsf{poly(|x|)}$ modes, $V_{|x|}(x,\ket\psi)$ accepts with probability smaller than $s$.
    \end{itemize} 
\end{definition}

We denote by \textsf{CVLH} the continuous-variable counterpart of the local Hamiltonian problem, which is formally defined as follows:

\begin{definition}[Continuous-variable local Hamiltonian problem]\label{def:cvlh}
Let $\mathcal S$ be a set of states and suppose that we are given a Hamiltonian $H$ which is a polynomial in the position and momentum operators, presented in the normal form
\begin{align}
H = \sum_{\substack{\boldsymbol{\mu},\boldsymbol{\nu} \in \mathbb{Z}_+^n \\ |\boldsymbol{\mu}+\boldsymbol{\nu}|\leq d}} h_{\boldsymbol{\mu},\boldsymbol{\nu}} X^{\boldsymbol{\mu}} P^{\boldsymbol{\nu}},
\end{align}
of degree $d$, over $n$ modes, and with constant rational coefficients with $|h_{\boldsymbol{\mu},\boldsymbol{\nu}}|=O(1)$. An instance of $\mathsf{CVLH}_\mathcal{S}^{d}(a,b)$ consists of a Hamiltonian $H$ specified as above, along with two constants $a$ and $b$ with $b-a\geq \frac{1}{\mathsf{poly}(n)}$ and a promise that either
$
\inf_{\psi\in\mathcal{S}} \langle \psi | H | \psi\rangle \leq a,
$
or
$
\inf_{\psi\in\mathcal{S}} \langle \psi | H | \psi\rangle \geq b,
$
holds. The \textsc{yes} instances are those where the smallest eigenvalue of $H$ is at most $a$, and the \textsc{no} instances are those where it is larger than $b$. We should decide which is the case.
\end{definition}

\noindent As the exact values of the constants $a,b$ do not change the complexity of the problem, we omit that dependency in what follows.

As we shall see in the upcoming sections, when $\mathcal S$ is related to the set of states with bounded stellar rank $r$, the parameters $d$ and $r$ tune the complexity of this problem, letting it range from being solvable in polynomial time, all the way to being undecidable. In particular, the dependency on the set of states over which the energy is being optimised is crucial to obtain a decidable problem.

Note that since bosonic Hamiltonians are generally unbounded operators, it is not clear a priori whether their spectrum is bounded from below: we call this \textit{the boundedness problem}.

\begin{definition}[Hamiltonian boundedness problem]\label{def:boundednessprob}
Given a Hamiltonian of degree $d$ in the normal form, as in \cref{def:cvlh}, $\mathsf{HBound}_{\mathcal S}^{d}$ is the problem of deciding whether the Hamiltonian is bounded from below over the set of states $\mathcal S$, or not.
\end{definition}

We study the computational complexity of these two problems in terms of standard classes for Gaussian Hamiltonians in \cref{sec:Gaussian-opt} and for non-Gaussian ones in \cref{sec:non-Gaussian-optimization}, and we provide the basis for relating them to the class \textbf{CVQMA} in \cref{sec:history}.

\subsection{Gaussian Hamiltonians}
\label{sec:Gaussian-opt}

In this section, we focus on the case of Gaussian Hamiltonians, i.e., when the degree of $H$ is less than or equal to $2$.
Denoting the set of Gaussian states over $n$ modes as $\mathcal G$, we show that one can decide in polynomial time whether $\min_{\ket\psi\in\mathcal G} \bra{\psi} H \ket\psi$ exists (i.e., if the Hamiltonian is bounded from below on Gaussian states). In fact, this problem can be reformulated as checking the positive definiteness of a matrix of coefficients. Moreover, when the boundedness condition is satisfied, we present a semi-definite program (SDP) that minimizes the energy over $\mathcal G$, and hence, the problem can be put in the class $\mathbf{P}$ (\cref{th:gspGaussian}). Since ground states of quadratic Hamiltonians are Gaussian states \cite{serafini2017quantum}, this allows us to solve both the boundedness problem and the continuous-variable local Hamiltonian problem in polynomial time for Gaussian Hamiltonians. We further show that it is sufficient to decide boundedness over the set of coherent states $\mathcal C$ (\cref{prop:boundednessGcoh}). The following result summarises this:

\begin{theorem}[Complexity of ground state problems for Gaussian Hamiltonians]\label{th:gspGaussian}
Deciding boundedness of a Gaussian Hamiltonian is equivalent to deciding boundedness over coherent states and can be done in polynomial time:
\begin{equation}
    \mathsf{HBound}_{\mathcal C}^{2}=\mathsf{HBound}_\mathcal{H}^{2}\subseteq\mathbf P.
\end{equation}
Moreover, estimating the ground state energy can also be done in polynomial time:
\begin{equation}
\mathsf{CVLH}_\mathcal{H}^{2}(a,b)\subseteq\mathbf P,
\end{equation}
for all $(a,b)$ such that $b-a\geq \frac{1}{\mathsf{poly}(n)}$.
\end{theorem}

Additionally, we show that the single-mode version of the problem can be solved analytically (\cref{prop:single-mode-quadratic}).

We now turn to the proofs of the various propositions.
We recall that a Gaussian Hamiltonian is a degree-2 Hamiltonian and that a Gaussian state is one which can be obtained by application of some Gaussian dynamics on the vacuum state $\ket{0^n}$. More concretely, $\ket\psi$ is a Gaussian state if there exist degree-2 Hamiltonians $H_1,\cdots,H_k$ such that $\ket{\psi} = e^{-iH_k} \cdots e^{-iH_1} \ket{0^n}$. We also recall that (mixture of) Gaussian states are the states with zero stellar rank. Therefore, the stellar function of a Gaussian state is of the form
\begin{align}\label{eq:stellarfuncGaussianstate}
F^\star(\Bf z) = \frac{1}{\mathcal N} \exp(-\frac12 \Bf z^T \Bf A \Bf z + \Bf b^T \Bf z),
\end{align}
where $\Bf b\in\mathbb C^n$ can be any complex vector, and $\Bf A$ is any complex symmetric (though \textit{not} hermitian)\footnote{A complex symmetric matrix $\Bf S$ is one that satisfies $\Bf S^T=\Bf S$.} with the property that $|A|\leq \mathbb I$. Note that any complex symmetric matrix $\Bf A$ can be decomposed as $\Bf A = \Bf U \Bf D \Bf U^T$, where $\Bf U$ is a unitary and $\Bf D$ is a positive semi-definite diagonal matrix. This is due to the fact that we can always use the singular value decomposition for $\Bf A$ to write $\Bf A= \Bf U \Bf D \Bf V$, and then use the symmetric assumption, which enforces $\Bf V = \Bf U^T$. This factorization is sometimes referred to as a Takagi--Autonne decomposition \cite{houde2024matrix}. The factor $\mathcal N$ is for normalisation and we omit it in what comes next.

In \cref{app:innerstellar}, we develop techniques that allow us to compute expectation of polynomials as a function of $\Bf A$ and $\Bf b$. Most importantly, we show the following relations:

\begin{corollary}\label{corol:gaussian-deg2-exp}
Let $F^\star_\psi(z) = \exp(-\frac12\Bf z^T\Bf A \Bf z + \Bf b^T \Bf z)$, where $\Bf A = \Bf U^T \Bf D \Bf U$. It is the case that
\begin{enumerate}
\item $(\langle F^\star_\psi, a_i F^\star_\psi\rangle)_i =  \Bf c$,
\item $(\langle F^\star_\psi, a_i a_j F^\star_\psi\rangle)_{ij} = \Bf U^T \frac{-\Bf D}{1- \Bf D^2} \Bf U + \Bf c \Bf c^T$,
\item $(\langle F^\star_\psi, a_ja_i^\dagger F^\star_\psi\rangle)_{ij} = \Bf U^T \frac{1}{1-\Bf D^2} \conj{\Bf U} + \Bf c^T \conj{\Bf c}$,
\end{enumerate}
where $\Bf c = \Bf U^T \frac{1}{1-\Bf D^2} (\conj{\Bf U} \Bf b - \Bf D \Bf U \conj{\Bf b})$.
\end{corollary}

\noindent Hence, we are able to write the problem of energy minimization over the Gaussian states in a simple mathematical form. Hereafter, we give a detailed derivation of a yet simpler form for this problem. Note that in \cref{corol:gaussian-deg2-exp}, the expressions are more neatly expressed in terms of $(\Bf c, \Bf U, \Bf D)$ rather than $(\Bf b, \Bf U, \Bf D)$. Here we show that for the purpose of optimization, it does not matter which set we optimize over. This can be done by taking the mapping
\begin{align}\label{eq:def-f-u,d}
f_{\Bf U,\Bf D}(\Bf b):= \Bf U^T \frac{1}{1-\Bf D^2}(\conj{\Bf U} \Bf b - \Bf D \Bf U \conj{\Bf b}).
\end{align}
We show hereafter that this map is invertible for any choice of $\Bf U$ and $\Bf D$. Hence, for every $\Bf c$, there exists a (unique) $\Bf b$. To show the invertibility of $f_{\Bf U,\Bf D}$, we write vectors in terms of their real and imaginary parts. This means, for $\Bf x \in \mathbb C^n$, we employ the notation
\begin{align}\label{eq:complex-to-real-vectors}
\wt{\Bf x}:= \begin{pmatrix}
    \Re(\Bf x)\\
    \Im(\Bf x)
\end{pmatrix}
\in \mathbb R^{2n}
\end{align}
and for any complex matrix $\Bf A$, we let
\begin{align}\label{eq:tilde-def}
\wt{\Bf A}:= \begin{pmatrix}
    \Re(\Bf A) & -\Im(\Bf A)\\
    \Im(\Bf A) & \Re(\Bf A)
\end{pmatrix}
\end{align}
which gives that $\wt{\Bf A \Bf x} = \wt{\Bf A}\wt{\Bf x}$. Moreover, we have
\begin{align}\label{eq:complex-to-real-conjugation}
\wt{\conj{\Bf{A}}} = \Bf J\,  \wt{\Bf A} \, \Bf J,\quad \wt{\conj{\Bf x}} = \Bf J \, \wt{\Bf x},
\end{align}
where $\Bf J = \begin{pmatrix}
\mathbb I_n & 0\\
0 & -\mathbb I_n
\end{pmatrix}$, where $\mathbb I_n$ is the $n\times n$ identity matrix. Furthermore, we have 
\begin{align}\label{eq:tilde-trick}
\Re(\langle \Bf x, \Bf y\rangle) = \wt{\Bf x}^T \wt{\Bf y}, \qquad \Re(\Bf x^T \Bf y) = \wt{\Bf x}^{T} \, \Bf J\, 
\wt{\Bf y}.
\end{align}
Lastly, we note that $\wt{\Bf A^\dagger} = \wt{\Bf A}^T$.
These can all be checked by direct inspection. 

Going back to proving the invertibility of $f_{\Bf U, \Bf D}$ from Eq.~\eqref{eq:def-f-u,d}, by denoting the left hand side by $\Bf c$, and employing identities introduced above, we have
\begin{align}
\wt{\Bf c} = \wt{\Bf U} \left(\frac{1}{1-\wt{\Bf D}^2}  - \frac{\wt{\Bf D}\, \Bf J}{1-\wt{\Bf D}^2} \right) \wt{\Bf U}^T \, \wt{\Bf b}.
\end{align}
Note that $\Bf D$ is a diagonal matrix, and hence, we can write
\begin{align}
\frac{1}{1-\wt{\Bf D}^2}  - \frac{\wt{\Bf D}\, \Bf J}{1-\wt{\Bf D}^2} = \begin{pmatrix}
\frac{1-\Bf D}{1-\Bf D^2} & 0\\
0 & \frac{1+\Bf D}{1-\Bf D^2}
\end{pmatrix} >0,
\end{align}
which concludes that the map between $\Bf c$ and $\Bf b$ is invertible. As a result, for the rest of this section, we will be optimizing over $(\Bf c, \Bf U, \Bf D)$ for convenience.

We first solve the problem for a single mode, and then, we discuss the more general form of it.

\begin{proposition}\label{prop:single-mode-quadratic}
The Hamiltonian
\begin{equation}H = \alpha a^2 + \conj \alpha^2 (a^\dagger)^2 + \beta a^\dagger a,\end{equation} is bounded from below if and only if $\beta > 2|\alpha|$, and the minimum energy is achieved for the Gaussian state with stellar function as in Eq.~(\ref{eq:stellarfuncGaussianstate}) with $A=a^\ast$ and $b=b^\ast$ where
\begin{align}
a^\ast = e^{-i\theta_\alpha} \tanh(\frac12\tanh^{-1}\left( \frac{2|\alpha|}{\beta}\right)), \quad b^\ast = 0,
\end{align}
with ground energy $\frac{1}{2}\left( \sqrt{\beta^2 - 4|\alpha|^2} - \beta \right)$.
\end{proposition}

\begin{proof} 
We aim to optimize this Hamiltonian over Gaussian states. Indeed, we know that the ground state of degree-2 Hamiltonians is one of such states. Using \cref{corol:gaussian-deg2-exp} we get
\begin{align}
\langle F^\star_\psi, H F^\star_\psi\rangle = -\beta-2\Re(\alpha e^{i\theta_a}) \frac{|a|}{1-|a|^2} + \frac{\beta}{1-|a|^2} + 2 \frac{\Re(\alpha(b-a\conj b)^2)}{(1-|a|^2)^2} + \beta \frac{|b-a\conj b|^2}{(1-|a|^2)^2}
\end{align}
Let $\phi$ represent the phase of the complex number $\alpha (b-a\conj b)^2$. We have
\begin{align}
\langle F^\star_\psi, H F^\star_\psi\rangle = -\beta-2\Re(\alpha e^{i\theta_a})\frac{|a|}{1-|a|^2} + \frac{\beta}{1-|a|^2} + \frac{c}{(1-|a|^2)^2} \left( 2|\alpha|\cos(\phi) + \beta \right),
\end{align}
where $c := |b-a\conj b|^2\geq0$. Note that if $|\alpha|> \frac12 \beta$, then, we can make the above expression as small as we want, by choosing a large $c$ (which can be done by choosing a large $b$, and proper phases for $a,b$). Therefore, there is no ground state if $|\alpha|> \frac12 \beta$. Now, assuming $|\alpha|\leq \frac12 \beta$, we note that any non-zero $c$ will result in an energy increase. Hence, we must set $c=0$, which will result in $b=0$. This reduces the problem to the following
\begin{align}
\min_{\psi} \langle F^\star_\psi, H F^\star_\psi \rangle = -\beta + \min_{a\in D} \frac{-2\Re(\alpha e^{i\theta_a}) |a| + \beta}{1-|a|^2},
\end{align}
where $D := \{ z: |z|<1\}$ is the interior of the unit disc in the complex plane. One can readily optimize over the phase of $a$, by setting $\theta_a = - \theta_\alpha$, and get
\begin{align}\label{eq:energy-gaussian}
\min_{\psi} \langle F^\star_\psi, H F^\star_\psi \rangle = -\beta + \min_{r\in[0,1)} \frac{\beta - 2|\alpha|r }{1-r^2}
\end{align}
Letting $\zeta = \tanh^{-1}(r)$, we can rewrite the optimization problem in terms of $\zeta$, which gives
\begin{align}
\min_{\psi} \langle F^\star_\psi, H F^\star_\psi \rangle = -\beta + \beta\left(\min_\zeta \cosh^2\zeta - \frac{2|\alpha|}{\beta} \sinh\zeta\cosh\zeta\right)
\end{align}
Notice that
\begin{align}
\begin{split}
\cosh^2\zeta - \frac{2|\alpha|}{\beta} \sinh\zeta\cosh\zeta = \frac12\left( \cosh(2\zeta) - \sinh(2\zeta) \tanh(\rho)\right)+\frac12 &= \frac12 \frac{\cosh(2\zeta - \rho)}{\cosh\rho} + \frac12\\
&\geq \frac{1}{2}(\sech\rho+1)
\end{split}
\end{align}
where we have used elementary hyperbolic identities, and employed $\rho:=\tanh(\frac{2|\alpha|}{\beta})$.
This gives us the solution $\zeta^\ast = \frac12 \tanh^{-1}(2|\alpha|/\beta)$. Hence, the ground state of our Hamiltonian has the following parameters
\begin{align}
a^\ast = e^{-i\theta_\alpha} \tanh(\frac12\tanh^{-1}\left( \frac{2|\alpha|}{\beta}\right)), \quad b^\ast = 0.
\end{align}
Also, note that the ground state exists if and only if $|\alpha|< \frac12 \beta$. Finally, note that the ground energy is given by
\begin{align}
\text{ground energy} = \frac{1}{2}(\sqrt{\beta^2-4|\alpha|^2} - \beta).
\end{align}
\end{proof}

Next, we use a similar technique for the multimode case to prove the following.

\begin{proposition}\label{prop:deg2-optimization}
Let
\begin{align}
\begin{split}
H = &\sum_{i<j} h^{(1)}_{i,j} a_i a_j^\dagger + h^{(2)}_{i,j} a_ia_j +  \conj h_{i,j}^{(2)} a_i^\dagger a_j^\dagger\\
+ &\sum_{i} h^{(3)}_i a_i + \conj h^{(3)}_i a_i^\dagger
\end{split}
\end{align}
Then, $H$ is bounded from below if and only if
\begin{align}
\begin{pmatrix}
\Re(\Bf h^{(2)} - 2\Bf h^{(1)}) & \Im(\Bf h^{(2)} - 2\Bf h^{(1)})\\
\Im(-\Bf h^{(2)} - 2\Bf h^{(1)}) & \Re(\Bf h^{(2)} + 2\Bf h^{(1)})
\end{pmatrix}
\end{align}
is a positive definite matrix. Moreover, if this condition is satisfied, the ground energy can be found in polynomial time via reformulation of the problem as a semi-definite program.
\end{proposition}

Note that we have assumed a normal form in terms of $a$ and $a^\dagger$. However, one can always transform a normal form in $X,P$ form into a normal form in terms of annihilation and creation operators in polynomial time, if the degree is constant. In this section, we work with this particular form, as it would be more convenient. The proof is as follows.

\begin{proof}
Using the results from \cref{corol:gaussian-deg2-exp}, we get the following

\begin{align}\label{eq:Gaussian-optimiation-formula}
\begin{split}
\langle F^\star_\psi, H F^\star_\psi\rangle = & 2\Re\left(\tr(\Bf h^{(2)} \Bf U^T \frac{\Bf D}{1-\Bf D^2} \Bf U )\right) + \tr(\Bf h^{(1)} \Bf U^T \frac{1}{1-\Bf D^2} \conj{\Bf U})\\
&+ 2 \Re\left(  \mathbf c^T \Bf h^{(2)} \Bf c\right) + \Bf c^T \Bf h^{(1)} \conj{\Bf c} + 2\Re(\Bf c \cdot \Bf h^{(3)}).
\end{split}
\end{align}
Note that the problem has already decoupled into two optimization problems, that are independent of each other. One is an optimization over $(\Bf U, \Bf D)$, while the other is an optimization over $\Bf c$. This has been demonstrated by writing the corresponding terms in separate lines in Eq.~\eqref{eq:Gaussian-optimiation-formula}.
Consider the optimization over $\Bf c$ for now. The objective function is
\begin{align}
2 \Re\left(  \mathbf c^T \Bf h^{(2)} \Bf c\right) + \Bf c^T \Bf h^{(1)} \conj{\Bf c} + 2\Re(\Bf c \cdot \Bf h^{(3)}).
\end{align}
Note that if there exists $\Bf c_0$ such that $2 \Re\left(  \mathbf c_0^T \Bf h^{(2)} \Bf c_0\right) + \Bf c_0^T \Bf h^{(1)} \conj{\Bf c_0}=-\epsilon <0$, then, one can make the expression arbitrarily negative by choosing $\Bf c= K \Bf c_0$ for a large $K>0$. This is due to the fact that $2 \Re\left(  \mathbf c^T \Bf h^{(2)} \Bf c\right) + \Bf c^T \Bf h^{(1)} \conj{\Bf c} + 2\Re(\Bf c \cdot \Bf h^{(3)}) = -K^2 \epsilon + K \cdot 2\Re(\Bf c_0 \cdot \Bf h^{(3)})$, which is a concave quadratic polynomial in one variable, and hence, is unbounded from below. Therefore, our goal for now is to identify the conditions under which $\min_{\Bf c} 2 \Re\left(  \mathbf c^T \Bf h^{(2)} \Bf c\right) + \Bf c^T \Bf h^{(1)} \conj{\Bf c} \geq 0$. By choosing a proper phase for $\mathbf c$, we can reduce this expression to $-2\abs{\mathbf c^T \Bf h^{(2)} \Bf c} + \Bf c^T \Bf h^{(1)} \conj{\Bf c}$. This can get arbitrarily small, i.e., we can get to $-\infty$, unless
\begin{align}\label{eq:gaussian-mess-1}
\langle \conj{\Bf c}, \Bf h^{(1)} \conj{\Bf c} \rangle \geq 2 \abs{\Bf c^T \Bf h^{(2)} \Bf c}, \quad \forall \mathbf{c}\in\mathbb C^{n}.
\end{align}
Hence, Eq.~\eqref{eq:gaussian-mess-1} is a necessary condition for the Hamiltonian to be bounded from below. Note that our second minimization problem is
\begin{align}\label{eq:opt-multi-g}
\begin{split}
\min_{\psi} \langle F^\star_\psi, H F^\star_\psi\rangle = \min \quad &-2\abs{\tr(\Bf h^{(2)} \Bf U^T \frac{\Bf D}{1-\Bf D^2} \Bf U )} + \tr(\Bf h^{(1)} \Bf U^T \frac{1}{1-\Bf D^2} \conj{\Bf U})\\
\mathrm{s.t.} \quad &\Bf U\in \mathcal S\mathcal U(n), 0 \leq \Bf D < \mathbb I
\end{split}
\end{align}

In \cref{lem:symmetric-vs-hermitian} we examine the conditions under which the Hamiltonian $H$ is bounded from below. In particular, we prove that Eq.~\eqref{eq:gaussian-mess-1} is equivalent to the following condition
\begin{align}\label{eq:existance-for-a}
\begin{pmatrix}
\Re(\Bf h^{(1)} - 2\Bf h^{(2)}) & \Im(-\Bf h^{(1)} + 2\Bf h^{(2)})\\
\Im(\Bf h^{(1)} + 2\Bf h^{(2)}) & \Re(\Bf h^{(1)} + 2\Bf h^{(2)})
\end{pmatrix}\geq0,
\end{align}
and therefore is a necessary condition for having a ground state. Let us denoting the matrix on the left hand side of Eq.~\eqref{eq:existance-for-a} by $\Bf M$. As outlined in the proof of \cref{lem:symmetric-vs-hermitian}, we have that $2 \Re\left(  \mathbf c^T \Bf h^{(2)} \Bf c\right) + \Bf c^T \Bf h^{(1)} \conj{\Bf c} = \wt{\Bf c}\Bf M\wt{\Bf c}$, where we have adapted the tilde notation from Eq.~\eqref{eq:complex-to-real-vectors}. Further using the mapping of complex vectors to real vectors of double the size lets us to write our optimization problem over $\Bf c$ as $\min_{\Bf x\in\mathbb R^{2n}} \Bf x^T \Bf M \Bf x + 2\wt{\Bf h^{(3)}} {}^T\, \Bf J \, \Bf x $ (please refer to Eq.~\eqref{eq:tilde-def}, Eq.~\eqref{eq:complex-to-real-conjugation}, and Eq.~\eqref{eq:tilde-trick}). Assuming $\Bf h^{(3)} \neq 0$, we require $\Bf M$ to be positive definite for the optimization problem to be bounded from below. This problem can be analytically solved to obtain $\conj{\Bf x} = -\Bf M^{-1} \,\Bf J \,\wt{\Bf h}^{(3)}$.

Moving on to the second optimization problem, we assume Eq.~\eqref{eq:existance-for-a} is satisfied, so that we try to find the ground energy. Recall that we have
\begin{align}\label{eq:reformulated-ge}
\text{ground energy} = \min \quad 2\Re\left\{\tr(\Bf h^{(2)} \Bf U \frac{\Bf D}{1-\Bf D^2} \Bf U^T)\right\} + \tr(\Bf h^{(1)} \Bf U \frac{1}{1-\Bf D^2}\Bf U^\dagger).
\end{align}
Employing the notation from Eq.~\eqref{eq:tilde-def}, we have that $2\Re\{\tr(\Bf X)\} = \tr(\wt{\Bf X})$ for any square operator $\Bf X$. Moreover, a simple calculation reveals that $\wt{\Bf X^\dagger} = \wt{\Bf X}^T$ and $\wt{\conj{\Bf X}}= \Bf J \wt{\Bf X} \Bf J$. Therefore, we can reformulate Eq.~\eqref{eq:reformulated-ge} as
\begin{align}\label{eq:tilde-reformed}
\text{ground energy} = \min \quad \tr\left( \wt{\Bf U}^T  \Bf J \wt{\Bf h}^{(2)} \wt{\Bf U} \frac{\wt{\Bf D}}{1-\wt{\Bf D^2}} \Bf J \right) + \frac12 \tr\left( \wt{\Bf U}^T \wt{\Bf{h}}^{(1)} \wt{\Bf U} \frac{1}{1-\wt{\Bf D}^2}\right)
\end{align}
Note that the domain we are optimizing here is more complex. For instance, $\wt{\Bf U}$ cannot be any arbitrary orthogonal matrix as it should satisfy the block structure of Eq.~\eqref{eq:tilde-def}.\footnote{Another way to see this is through the strict inclusion $U(n)\subset O(2n)$ for all $n\geq 2$.} However, we note that $\wt{\Bf{h}}^{(2)}\Bf J$ and $\wt{\Bf{h}}^{(1)}$ have specific structures which will allow us to overcome this issue. In particular, we note that both $\wt{\Bf{h}}^{(2)}\Bf J$ and $\wt{\Bf{h}}^{(1)}$ are hermitian and that for any eigenvector of $\wt{\Bf{h}}^{(2)}\Bf J$, say $\Bf u = (\Bf u_r, \Bf u_i)^T$, with eigenvalue $\mu_i$, we have that $\Bf u':=(-\Bf u_i, \Bf u_r)$ is also an eigenvector but with eigenvalue $-\mu_i$. This can be observed from the block structure of $\wt{\Bf{h}}^{(2)}\Bf J$:
\begin{align}
\wt{\Bf{h}}^{(2)}\Bf J = \begin{pmatrix}
\Re(\Bf h^{(2)}) & \Im(\Bf h^{(2)})\\
\Im(\Bf h^{(2)}) & -\Re(\Bf h^{(2)})
\end{pmatrix}
\end{align}
and that any matrix of the form $\begin{pmatrix}
\Bf A & \Bf B\\
\Bf B & -\Bf A
\end{pmatrix}$ satisfies the condition mentioned above. On the other hand, any matrix of the form $\begin{pmatrix}
\Bf A & -\Bf B\\
\Bf B & \Bf A
\end{pmatrix}$, including $\wt{\Bf h}^{(1)}$ satisfies the property that if $\Bf v = (\Bf v_1, \Bf v_2)^T$ is an eigenvector with eigenvalue $\nu_i$, so is the new vector $\Bf v = (-\Bf v_2, \Bf v_1)$. Therefore, there is a basis in which, we can write
\begin{align}\label{eq:good-basis-def}
\wt{\Bf h}^{(1)} = \begin{pmatrix}
\Bf P & 0\\
0 & \Bf P
\end{pmatrix},\qquad
\wt{\Bf h}^{(2)}\Bf J = \begin{pmatrix}
-\Bf Q & 0\\
0 & \Bf Q
\end{pmatrix},
\end{align}
where $\Bf P,\Bf Q\geq 0$. Note that the boundedness condition Eq.~\eqref{eq:existance-for-a} is equivalent to $\Bf P\geq 2\Bf Q$. Is is straightforward to show that the matrix $\wt{\Bf U}$ that minimizes the right hand side of Eq.~\eqref{eq:tilde-reformed} is of the form $\wt{\Bf U} = \begin{pmatrix}
\Bf V & 0\\
0 & \Bf V
\end{pmatrix}$ in the basis introduced in Eq.~\eqref{eq:good-basis-def}. It is also straightforward to check that such a matrix satisfies the conditions on $\wt{\Bf U}$, as this procedure merely tells us what are the phases that should be applied to each column of $\Bf U$. As a result, our problem transforms into the following
\begin{align}
\text{ground energy} = \min_{0\leq \Bf D<I} \tr\left( \Bf P \frac{\Bf D}{1-\Bf D^2} \right) - 2\tr\left( \Bf Q \frac{1}{1-\Bf D^2} \right).
\end{align}
Note that in this formulation, the $n\times n$ matrices $\Bf Q, \Bf P$, and $\Bf D$ are all real-valued, and that we have changed the condition for $\Bf D$ from be diagonal to be merely positive and less than $I$. Finally, we let $\Bf X:= \frac{1+\Bf D}{1-\Bf D}$. This allows us to rewrite our optimization problem as
\begin{align}
\text{ground energy} = \min_{\Bf X> \Bf I} \tr\left((\Bf P-2\Bf Q) \Bf X\right) + \tr\left((\Bf P+2\Bf Q) \Bf X^{-1}\right)
\end{align}
which can be readily reformulated as the following semi-definite program
\begin{center}
\underline{{Semi-definite program for ground energy}}
\begin{align}\label{eq:Gaussian-Optimization-SPD}
\begin{split}
&\min \tr\left((\Bf P-2\Bf Q) \Bf X\right) + \tr\left((\Bf P+2\Bf Q) \Bf Y\right)\\
&\quad \text{s.t.} \qquad \begin{pmatrix}
\Bf X & \Bf I\\
\Bf I & \Bf Y
\end{pmatrix}
\geq 0,\\
& \quad \qquad\qquad \Bf X>\Bf I.
\end{split}
\end{align}
\end{center}
The condition $\begin{pmatrix}
\Bf X & \Bf I\\
\Bf I & \Bf Y
\end{pmatrix}$ is equivalent to $\Bf Y\geq \Bf X^{-1}$ and $\Bf X \geq 0$ (see \cite[Section A.5.5]{boyd2004convex}). Note that for achieving the minimum objective value, as $\Bf P+2\Bf Q> 0$, it must be the case that $\Bf Y=\Bf X^{-1}$ happens. As pointed out earlier, $\Bf P-2\Bf Q\geq 0$ is equivalent to the necessary condition for boundedness, but the optimization problem \cref{eq:Gaussian-Optimization-SPD} ensures that given $\Bf P\geq 2\Bf Q$, the ground energy does indeed exist (as it is lower bounded by $0$), and hence, the condition presented in \cref{eq:existance-for-a} is the necessary and sufficient condition for having a ground state (with the caveat that the inequality must be strict whenever $\Bf h^{(3)}\neq 0$).
\end{proof}

Below is the lemma that was used in our of \cref{prop:deg2-optimization} above to obtain the boundedness condition.

\begin{lemma}\label{lem:symmetric-vs-hermitian}
Let $\Bf S$ be a $n\times n$ complex symmetric matrix, and $\Bf P$ be a hermitian matrix of the same size. The following are equivalent
\begin{enumerate}
    \item $\abs{\Bf x^T \Bf S \Bf x} \leq \langle \Bf x, \Bf P \Bf x\rangle$, for all $\Bf x\in \mathbb{C}^n$.
    \item $\begin{pmatrix}
        \Re(\Bf P-\Bf S) & \Im(\Bf S-\Bf P)\\
        \Im(\Bf S+\Bf P) & \Re(\Bf P+\Bf S)
    \end{pmatrix}\geq 0$.
\end{enumerate}
Note that we are using the usual inner product notation $\langle \Bf x,y\rangle := \conj{\Bf x}^T \Bf y$ on $\mathbb C^n$.
\end{lemma}
\begin{proof} Note that condition (1) of the lemma, i.e., the condition $\abs{\Bf x^T \Bf S \Bf x} \leq \langle x, \Bf P \Bf x\rangle$ for all $x$, can be rephrased as follows
\begin{align}\label{eq:intermediate-symmetric-vs-hermitian}
\Re\left(\Bf x^T \Bf S \Bf x\right) \leq \langle \Bf x, \Bf P \Bf x\rangle\text{, for all }\Bf x\in \mathbb{C}^n.
\end{align}
The reason for this equality is the following: firstly note that as $\Re(z)\leq |z|$ we can conclude \eqref{eq:intermediate-symmetric-vs-hermitian} from the condition (1). Now, let $\phi$ be the phase of the complex number $\mathbf x^T \mathbf S \mathbf x$, and note that $\Bf x':= \Bf x e^{-i\phi/2}$ satisfies $|\Bf x^T \Bf S \Bf x| = \Bf x'^T \Bf S \Bf x' \leq \langle \Bf x',\Bf P\Bf x'\rangle = \langle \Bf x,\Bf P\Bf x\rangle$, where the inequality follows if we assume \eqref{eq:intermediate-symmetric-vs-hermitian}. Therefore, we have shown that assuming $\eqref{eq:intermediate-symmetric-vs-hermitian}$ one can get condition (1). Thus, these conditions are equivalent.

We find it useful to apply the convention introduced in Eq.~\eqref{eq:complex-to-real-vectors} and Eq.~\eqref{eq:tilde-def} and to write the above condition for real vectors. Using Eq.~\eqref{eq:tilde-trick}, we  rewrite the right hand side of Eq.~\eqref{eq:intermediate-symmetric-vs-hermitian} in the real domain as
\begin{align}
\Re\left(\Bf x^T \Bf S \Bf x\right) \leq \langle \Bf x, \Bf P \Bf x\rangle \Leftrightarrow \wt{\Bf x}^T \Bf J\wt{\Bf S}\, \wt{\Bf x} \leq \wt{\Bf x}^T \wt{\Bf P}\, \wt{\Bf x}
\end{align}
Note that $\Bf J\wt{\Bf S}$ and $\Bf P$ both are real symmetric matrices, and therefore, the condition $\wt{\Bf x}^T \Bf J\wt{\Bf S}\, \wt{\Bf x} \leq \wt{\Bf x}^T \wt{\Bf P}\, \wt{\Bf x}$ is equivalent to $\Bf J \wt{\Bf S}\leq \Bf P$. Expanding this condition, we get
\begin{align} 
\begin{pmatrix}
\Re(\Bf P-\Bf S) & \Im(\Bf S-\Bf P)\\
\Im(\Bf S+\Bf P) & \Re(\Bf P+\Bf S)
\end{pmatrix} \geq 0
\end{align}
which completes the proof.
\end{proof}

In what follows, we remark on a few facts that can be learned from the above arguments. We first, take a closer look at the boundedness condition in \cref{eq:existance-for-a}.

\begin{proposition}\label{prop:boundednessGcoh}
A quadratic Hamiltonian is bounded from below if and only if it is bounded on coherent states.
\end{proposition}

\begin{proof}
For a coherent state with displacement $\Bf z\in\mathbb C^n$, we have that $\Bf b=\Bf z$ and $\Bf A=0$, which results in $\Bf c=\Bf z$ as defined in \cref{corol:gaussian-deg2-exp}. Plugging into Eq.~\eqref{eq:Gaussian-optimiation-formula} gives
\begin{align}
\langle \Bf z, H \Bf z\rangle = 2 \Re\left(  \mathbf z^T \Bf h^{(2)} \Bf z\right) + \Bf z^T \Bf h^{(1)} \conj{\Bf z} + 2\Re(\Bf z \cdot \Bf h^{(3)}).
\end{align}
As discussed in the proof of \cref{prop:deg2-optimization}, the boundedness of the above expression turns out to be the necessary and sufficient condition for the Hamiltonian to be bounded over all Gaussian states. 
\end{proof}

\noindent As we point out later in \cref{sec:non-Gaussian-optimization}, this property does not hold for higher-degree Hamiltonians, i.e., boundedness over coherent states does not imply boundedness over all states. 

\subsection{Non-Gaussian Hamiltonians}\label{sec:non-Gaussian-optimization}

In this section, we turn to the more subtle case of non-Gaussian Hamiltonians that are polynomials of arbitrary degree in the position and momentum operators.

We first show that the boundedness problem is much harder for non-Gaussian Hamiltonians than it is for Gaussian ones. Indeed, already for degree-4 Hamiltonians we obtain the following result:

\begin{theorem}\label{th:boundedness}
The problem $\mathsf{HBound}_\mathcal{H}^{4}$ of deciding boundedness of a degree-4 bosonic Hamiltonian is $\mathbf{co}$-$\mathbf{NP}$-hard.
\end{theorem}

We prove this result through a reduction from the matrix copositivity (\textsf{McP}) problem:

\begin{definition}[Matrix copositivity]
Given an $n\times n$ matrix $M$, $\mathsf{McP}$ is the problem of deciding whether
\begin{align}
\sum_{i,j} M_{i,j} x_i x_j
\end{align}
is non-negative for all $\Bf x\in(\mathbb R_+)^n$.
\end{definition}

\noindent We have employed the notation $\mathbb R_+ = \{x\in \mathbb R: x\geq 0\}$. The \textsf{McP} problem is known to be $\mathbf{co}$-$\mathbf{NP}$-complete \cite{murty1987some}.

\begin{proof}
Consider the Hamiltonian
\begin{align}
H = \sum_{i,j\in[n]} M_{i,j} N_i N_j.
\end{align} 
It is transparent that if $\Bf M\in\mathbb R^{n\times n}$ is copositive, then $H$ is a positive semi-definite operator. This is due to the fact that $H$ is diagonal in the number basis, and hence, its eigenvalues are computed according to $\sum_{i,j} M_{i,j} m_i m_j$ for some $\Bf m\in\mathbb Z_+^n$. Therefore, $\lambda_{min}(H)\geq 0$. But, as $\ket0$ saturates the lower-bound, we get that $\lambda_{min}(H)=0$ if $\Bf M$ is copositive.

Next, we show that if $\Bf M$ is not copositive, then, $\lambda_{min}(H)=-\infty$. To this end, let $\Bf x\in\mathbb R^n$ be a vector such that $\Bf x^T \Bf M \Bf x =-\epsilon<0$. Then, consider the states over one mode $\ket{\psi_i}:= \sqrt{1-\frac{x_i}{\lceil x_i \rceil}}\ket{0} + \sqrt{\frac{x_i}{\lceil x_i \rceil}}\ket{\lceil x_i \rceil}$, where we have used the Fock basis representation. Now, define the state $\ket{\psi}:=\otimes_{i=1}^n\ket{\psi_i}$. It is evident that
\begin{align}
\bra{\psi}H\ket{\psi} = \Bf x^T \Bf M \Bf x = -\epsilon.
\end{align}
Note that one can make $\epsilon$ arbitrarily large by simply scaling the vector $\Bf x$. This proves that $\lambda_{min}(H)=-\infty$ if $\Bf M$ is not copositive. 
\end{proof}

\noindent Note that the proof of the theorem implies that the problem $\mathsf{HBound}_{\mathcal{F}}^{4}$ is $\mathbf{co}$-$\mathbf{NP}$-hard, where $\mathcal{F}$ is the set of Fock states.

Recall that degree-2 Hamiltonians are bounded from below if and only if they are bounded from below on coherent states (see \cref{prop:boundednessGcoh}). Hence, one can ask if, for a degree-4 Hamiltonian, the boundedness over coherent states is sufficient for it to be bounded over all states. The answer to this question turn out to be negative, as demonstrated by the Hamiltonian $H = N_1^2 + N_2^2 - 2N_1 N_2 - N_1 - N_2$. This Hamiltonian satisfies $\langle \Bf z| H |\Bf z \rangle = \sum_{i,j} M_{i,j} |z_i|^2 |z_j|^2$, with the $2\times2$ matrix $\Bf M = \mathbb I - \sigma_x \geq 0$, where $\sigma_x$ is the Pauli-X matrix. Therefore, $\bra{\Bf z} H \ket{\Bf z}\geq 0$. Nevertheless, for a product of Fock basis states of the form $\ket\psi=\ket{m,m}$, we have $\bra{\psi} H \ket\psi = -2m$, which shows that $H$ is not bounded from below.

Next, we discuss an approach based on Sum-of-Squares (SoS) decomposition of non-negative polynomials which provides a sufficient condition for a Hamiltonian to be bounded from below. This condition can be checked in classical polynomial time. To this end, we first introduce the following lemma (see \cref{app:commutation} for a proof).

\begin{lemma}\label{prop:SOS}
Any multimode Hamiltonian can be expanded as
\begin{align}\label{eq:XP-form}
H = \sum_{\boldsymbol\mu,\boldsymbol\nu} h_{\boldsymbol\mu,\boldsymbol\nu} \{X^{\boldsymbol\mu},P^{\boldsymbol\nu}\}.
\end{align}
Note that the coefficients $(h_{\boldsymbol\mu,\boldsymbol\nu})_{\boldsymbol\mu,\boldsymbol\nu}$ are real-valued.
\end{lemma}
We can now associate a real-valued polynomial over $2n$ real variables to any polynomial Hamiltonian. We denote this polynomial by $p_H$, which is formally defined as
\begin{align}
p_H(x_1,\cdots,x_n,p_1,\cdots,p_n) = 2 \sum_{\boldsymbol\mu,\boldsymbol\nu} h_{\boldsymbol\mu,\boldsymbol\nu} x^{\boldsymbol\mu} p^{\boldsymbol\nu} \in\mathbb R[x_1,\cdots x_n,p_1,\cdots,p_n].
\end{align}

We point out that it is crucial to have a connection with real-valued polynomials, as SoS decompositions are not naturally defined for complex-valued polynomials. Let $\Sigma$ denote the set of Hamiltonians that admit a SoS decomposition. Similarly, let $\sigma$ denote the set of real-valued polynomials over the reals that admit a SoS decomposition. 
We are now ready to present our SoS condition.

\begin{proposition}\label{prop:soswitness}
Let $H$ be a degree-$4$ Hamiltonian. We have that $H+\alpha\cdot \mathbb I\in \Sigma$ if $p_H\in\sigma$, where $\alpha\in\mathbb R$ is a parameter which given an SoS decomposition of $p_H$ can be found efficiently.
\end{proposition}
This proposition above helps  determine if a degree-$4$ Hamiltonian is bounded from below, and gives a lower bound $\alpha$ on the Hamiltonian.

\begin{proof}
If $p_H\in\sigma$, then $H\in\Sigma + \alpha \mathbb I$, with $\alpha$ computable efficiently given the SoS decomposition of $p_H$. Let $p_H = \sum_{i} g_i^2$. Each $g_i$ has degree at most $2$. Hence, can be written as $g_i = a_i x^2 + b_i p^2 + c_i xp + d_i$. We use the following claim: let $m_1 = \{X^{s_1}, P^{r_1}\}$ and $m_2 = \{X^{s_2}, P^{r_2}\}$. It is the case that $\{m_1,m_2\} = 2\{X^{s_1+s_2}, P^{r_1+r_2}\} + c\cdot \mathbb I$
\footnote{Its correctness is obvious whenever $(s_i,r_i)$ have a zero component. Hence, the only non-trivial cases are $(1,1)$ with $(1,0)$ and also $(1,1)$ with $(1,1)$.}. Therefore, promoting $g_i$'s to operators, will give a decomposition for $H$, up to some shift. More concretely, $H = \sum_{i} G_i^2 + \alpha\cdot \mathbb I$. To be more precise, let $p_H$ have an SoS decomposition as follows
\begin{align}
p_H = \sum_i g_i^2,
\end{align}
where $g_i=\sum_{m,n} \alpha^{(i)}_{m,n} x_mx_n + \beta^{(i)}_{m,n} x_mp_n + \gamma^{(i)}_{m,n} p_mp_n + \delta^{(i)}_m x_m + \eta^{(i)}_m p_m$. One can observe that $H = \sum_{i} G_i^2 + c\cdot \mathbb I$, where $G_i:=\sum_{i} \alpha^{(i)}_{m,n} X_mX_n + \frac{1}{2}\beta^{(i)}_{m,n} \{X_m,P_n\} + \gamma^{(i)}_{m,n} P_mP_n + \delta^{(i)}_m X_m + \eta^{(i)}_m P_m$.
\end{proof}

In the rest of this section, we focus on the problem of estimating the ground state energy of polynomial Hamiltonians. As we will show in \cref{thm:undecidability}, putting no restriction on the family of states over which the optimization takes place leads to an undecidable problem.
Hence, to obtain decidable instances, we introduce sets of states of increasing complexity, parametrized by their stellar rank and their particle number:

\begin{definition}\label{def:e-cons-stellar}
    For all $r,E\in\mathbb N$, the set of states of stellar rank bounded by $r$ with particle number bounded by $E$ over $n$ modes is denoted by $\mathcal S_r^E$.
\end{definition}

\noindent For instance, $\mathcal S_0^{\mathsf{exp}(n)}$ is the set of Gaussian states with exponentially bounded number of particles, while $\mathcal S_{\mathsf{poly}(n)}^{\mathsf{poly}(n)}$ is the set of states with polynomially-bounded stellar rank and particle number.

Recall that any state of finite stellar rank $r$ over $n$ modes can be decomposed as $G\ket c$, where $G$ is an $n$-mode Gaussian unitary and $\ket c$ is a state with support only on Fock state with particle number less or equal to $r$ (see \cref{thm:core}). The Euler--Bloch--Messiah decomposition (\cref{thm:EBMdecomp}) shows that the Gaussian unitary $G$ can be expressed as a sequence of passive linear gates, squeezing, and displacement. In other words, we have
\begin{align}\label{eq:gaussian-ebm}
G = V \bigotimes_{i=1}^n G_i U,
\end{align}
where $V,U$ are passive linear operators (which conserve the total particle number), and $G_i$ are single-mode Gaussian unitary operators, that can be written as a squeezing followed by a displacement, i.e., $G_i = S_i \cdot D_i$. We denote the squeezing parameter of a Gaussian by $\xi$ and the displacement parameter by $\delta$ as in Eq.~\eqref{eq:simple-gaussian}. Given this decomposition, we provide an operational characterization of the set $\mathcal S_r^E$ as follows:

\begin{proposition}\label{prop:boundG}
    With the above notations, let $\ket\psi\in\mathcal S_r^E$. Then, we have that for all $i=1,\dots,n$,
    \begin{align}
        |\xi_i|\le O(\log((r+1) E)),\qquad|\delta_i|\le O(\mathsf{poly}(r,E)).
    \end{align}
\end{proposition}

\noindent We refer to \cref{app:boundG} for a proof when $r=0$, and \cref{lem:last-lem} for $r>0$.

With these results in place, we study the computational complexity of the ground energy problem over $\mathcal S_r^E$, as $r$ increases from $0$ to $\infty$: we first show that it is $\mathbf{NP}$-complete over $\mathcal S_0^{\mathsf{exp}(n)}$ in \cref{thm:gaussian-optimization}, then that it is in $\mathbf{QMA}$ over $\mathcal S_{\mathsf{poly}(n)}^{\mathsf{poly}(n)}$ in \cref{thm:bounded-squeezing-and-displacement}, and finally that it is $\mathbf{RE}$-hard in \cref{thm:undecidability} when no restriction is made on the stellar rank.

\subsubsection{Optimization over Gaussian states}
\label{sec:deg-2-optimization}

We first consider the problem of estimating the minimum energy of non-Gaussian Hamiltonians over Gaussian states with at most exponential average particle number in $\mathcal S_0^{\mathsf{exp}(n)}$. We show:

\begin{theorem}\label{thm:gaussian-optimization}
$\mathsf{CVLH}_{\mathcal S_0^{\mathsf{exp}(n)}}^4$ is $\mathbf{NP}$-complete.
\end{theorem}

Note that the problem is indeed in \textbf{NP}, as in the YES cases, the prover can send the displacement and squeezing parameters. Then, the verifier can compute degree-4 expectations of each monomial, and hence, evaluate the overall energy. Following this procedure, the verifier cannot be tricked in the NO cases.
We next show $\mathbf{NP}$-hardness in the case of unbounded displacement, which follows from \cref{prop:deg4-optimization-over-gaussian} below.

\begin{proposition}\label{prop:deg4-optimization-over-gaussian}
Let $\mathcal S$ be any family of states which includes all coherent states over $n$ modes. Then, minimizing the energy of degree-4 Hamiltonians over $\mathcal S$ is $\mathbf{NP}$-hard.
\end{proposition}

Similar to \cref{th:boundedness}, the proof of this proposition relies on the fact that \textit{matrix non-copositivity} ($\mathsf{MncP}$) is an \textbf{NP}-hard problem. Indeed, our proof below introduces a reduction from $\mathsf{MncP}$ to the optimization of degree-4 Hamiltonians over Gaussian states.

\begin{proof}[Proof of \cref{prop:deg4-optimization-over-gaussian}]
Let $\Bf M\in\mathbb R^{n\times n}$. Consider the following Hamiltonian
\begin{align}
H = \sum_{i,j} M_{ij} N_i N_j
\end{align}
In what follow, we show that
\begin{enumerate}
\item If $\Bf M$ is copositive
\begin{align}
\min_{\psi\in\mathcal S} \langle \psi |H| \psi\rangle = 0
\end{align}
\item Otherwise, we have
\begin{align}
\inf_{\psi\in\mathcal S} \langle \psi |H| \psi\rangle = -\infty.
\end{align}
\end{enumerate}
Note that this is sufficient since the ability to answer whether the ground energy is below $-2$ or above $-1$ decides $\mathsf{MncP}$. In what follows, we prove each of the above claims, and hence, conclude the proof.

\begin{enumerate}
\item Moreover $H$ is diagonal in the number basis. Hence
\begin{align}\label{eq:pf-of-deg4gaus-1}
\langle \psi |H| \psi\rangle \geq \min_{\Bf m\in \mathbb N^n} \langle \Bf m |H| \Bf m\rangle,  
\end{align}
where $\ket{\Bf m}$ represents the tensor product of Fock basis vectors i.e., $\ket{\Bf m}:=\ket{m_1}\otimes\cdots\otimes\ket{m_n}$.
As $\Bf M$ is copositive, we get that $\langle \Bf m, H \Bf m\rangle = \sum_{i,j} M_{i,j} m_i m_j\geq 0$. Putting this together with Eq.~\eqref{eq:pf-of-deg4gaus-1} gives $\langle \psi|H| \psi\rangle \geq 0$ for any state $\ket\psi$. Also, note that $\ket\psi=\ket{0^n}$ achieves the zero energy.

\item Letting $\Bf M$ being non-copositive, there exists $\conj{\Bf x} \in \mathbb R_+^n$ such that $\sum_{ij} M_{ij} x_i x_j =-\epsilon$ for some $\epsilon>0$. Furthermore, let $a:= \sum_i M_{i,i} x_i$. Now, consider the coherent state $\ket{\Bf z}$, where $z_i = \sqrt{ K x_i}$ for some $K>0$. We have
\begin{align}\label{eq:long-coherent}
\begin{split}
\langle \Bf z| H|\Bf z\rangle &= \sum_{i,j} M_{i,j} \langle \Bf z| a_i^\dagger a_i a_j^\dagger a_j| \Bf z\rangle\\
&= \sum_{i,j} M_{i,j} \conj{z}_j z_i \bra{\Bf z} a_i a_j^\dagger \ket{\Bf z}\\
&= \sum_{i,j} M_{i,j} \conj{z}_j z_i \bra{\Bf z} a_j^\dagger a_i\ket{\Bf z} + \sum_{i} M_{i,i} |z_i|^2\\
&= \sum_{i,j} M_{i,j} |z_i|^2 |z_j|^2 + \sum_{i} M_{i,i} |z_i|^2\\
&= - K^2 \cdot \epsilon + K\cdot a
\end{split}
\end{align}
Note that choosing $K$ to be sufficiently large, gives arbitrarily large negative numbers.
\end{enumerate}
\end{proof}

We are now ready to present the proof of \cref{thm:gaussian-optimization}.

\begin{proof}[Proof of \cref{thm:gaussian-optimization}]
To show \textbf{NP}-hardness, we note that even the following reformulation of $\mathsf{MncP}$, which considers only constantly large vectors $\Bf x$, remains \textbf{NP}-hard
\begin{align}\label{eq:bounded-matrix-copositivity}
\begin{split}
\min \quad &\Bf x^T \Bf M \Bf x\\
\mathrm{s.t.} \quad &0\leq x_i \leq 1, \forall i\in[n]
\end{split}
\end{align}
The reason Eq.~\eqref{eq:bounded-matrix-copositivity} is still \textbf{NP}-hard is due to the fact that its optimal value is either $0$ or at most $-2^{-L}$ with $L$ being the size of $\Bf M$ (depending on $\Bf M$ being copositive or not) \cite{murty1987some}. Hence, following the proof of \cref{prop:deg4-optimization-over-gaussian}, one can set $\conj{\Bf x}$ to be solution of Eq.~\eqref{eq:bounded-matrix-copositivity}, giving $\epsilon>2^{-L}$ in Eq.~\eqref{eq:long-coherent}, and in case of non-copositivity, choose $K = 2^{L}\tr(\Bf M)+ 1$ to conclude that the energy is at most $-(2\tr(\Bf M) + 2^{-L})$. Finally, note that $\mathcal S_0^{\mathsf{exp}(n)}$ contains the coherent states of exponentially large amplitudes so we can indeed choose $z_i=\sqrt{Kx_i}$ as in the above proof.
\end{proof}

\begin{remark}\label{rem:any-stellar-rank-is-np-hard}
Note that the hardness result presented in the proof above can be extended to any family of states that contain exponentially displaced coherent states (this is a stronger statement than \cref{prop:deg4-optimization-over-gaussian}). Therefore, we have $\mathsf{CVLH}_{\mathcal S_{r}^{\mathsf{exp}}}$ is $\mathbf{NP}$-hard for any $r\geq 0$.
\end{remark}

\subsubsection{Optimization over states of bounded stellar rank}

In this section, we discuss the complexity of optimizing a Hamiltonian over states of bounded stellar rank. Recall that $\mathcal S_r^E$ denotes the set of states with stellar rank at most $r$ and energy (average particle number) at most $E$.

\begin{theorem}\label{thm:bounded-squeezing-and-displacement}

We have that $\mathsf{CVLH}_{\mathcal S_{\mathsf{poly}(n)}^{\mathsf{poly}(n)}}^{O(1)}\in\mathbf{QMA}$.
\end{theorem}
 
\begin{proof}

Our goal is to optimize the energy with respect a given Hamiltonian $H = \sum_{\boldsymbol\mu, \boldsymbol\nu \in \mathbb Z_+^n} \alpha_{{\boldsymbol\mu, \boldsymbol\nu}} a^{{\boldsymbol\mu}} a^{\dagger {\boldsymbol\nu}}$ over $n$ modes and of constant degree $d$ over the set of states with stellar rank $r=\mathsf{poly}(n)$ and average particle number at most $\mathsf{poly}(n)$. 

Let $E_r(H)$ be the minimum energy over $\mathcal S_r$. Using \cref{thm:core}, any such quantum state can be characterized by $\ket{\psi} = G \ket{c}$, where $\ket{c}$ is a core state of degree $r$ and $G$ is a Gaussian operator. Let $\mathsf{C}_{r,n}$ be the family of core states of degree $r$ over $n$ modes, and let $\mathcal{G}_n$ be the family of Gaussian operators over $n$ modes. Therefore, the minimum energy becomes 
\begin{equation} 
E_r (H) = \inf_{\ket{c} \in \mathsf{C}_{r,n}, G \in \mathcal{G}_n} \bra{c} G^\dagger H G \ket{c}
\end{equation}
Let $\Pi_r:=\sum_{\mathbf m: |\Bf m|\le r} \ket{\Bf m}\!\bra{\Bf m}$ is the projection onto the Fock basis states that have $r$ many particles. We obtain (see also \cite{fiuravsek2022efficient})
\begin{align}
E_r (H) = \inf_{G\in\mathcal G} \lambda_{\min}(\Pi_{r}  G^\dagger H G \Pi_{r}),
\end{align}
Note that the dimension of $\tilde H:= \Pi_{r}  G^\dagger H G \Pi_{r}$ is $\binom{n+r-1}r = n^{O(r)}$. Next, we use the Euler--Bloch--Messiah decomposition (\cref{thm:EBMdecomp}) to decompose $G = V \bigotimes_{i=1}^m G_i U$, where $U$ and $V$ are passive linear optical elements. They map $U a^\dagger_i U^\dagger = \sum_j U_{ij} a^\dagger_j$, where $U_{ij}$ are the complex matrix elements corresponding to $U$ (similarly for $V)$. $G_i$ are single-mode Gaussian operators. We note that these operators act as affine maps on creation and annihilation operators, i.e., $G_i a G_i^\dagger = A_i a + B_i a^\dagger + C_i$, for complex numbers $A_i, B_i, C_i$. Crucially, neither of these two operators changes the degree of a Hamiltonian under conjugation. For more details, see \cref{sec:CVprocesses}. We have, $[\Pi_r, U] = 0$, since these operators preserve the particle number, so 
 \begin{align}
E_r (H) = \inf_{G\in\mathcal G} \lambda_{\min}(\Pi_{r}  \bigotimes_i G_i^\dagger (V^\dagger H V) \bigotimes_i G_i \Pi_{r}).
\end{align}
Now, when the average particle number is polynomially bounded, \cref{prop:boundG} ensures that the displacement and squeezing parameters specifying the Gaussian unitary are polynomially-bounded and logarithmically-bounded, respectively.
 
We now give a $\mathbf{QMA}$ procedure to perform the optimization within inverse-polynomial error when the stellar rank is bounded as $r=\mathsf{poly}(n)$. The prover encodes the description of the Gaussian operator $G$, computes $G^\dagger H G$ and brings it into a normal form in terms of $a,a^\dagger$ in polynomial time. The procedure is based on the following observations. $H' = V^\dagger H V$ from $H$ and $H'' = \otimes_i G^\dagger_i H' \otimes_i G_i$ from $H'$ can be computed using linear operators.
Now, the conjugated Hamiltonian still has the same degree and can be expressed as
\begin{align}\label{eq:conjugated-Hamiltonian}
H'' = G^\dagger H G = \sum_{\substack{\boldsymbol{\mu},\boldsymbol{\nu} \in \mathbb{Z}_+^n \\ |\boldsymbol{\mu}+\boldsymbol{\nu}|\leq d}} 
\alpha''_{\boldsymbol\mu,\boldsymbol\nu} a^{\boldsymbol\mu} a^{\dag\boldsymbol\nu}.
\end{align}
 
$\tilde H$ can be computed from $H''$ by computing overlaps with number bases 
\begin{equation} 
 \tilde H = \Pi_r G^\dagger H G \Pi_r = 
 \sum_{\substack{\boldsymbol{\mu},\boldsymbol{\nu} \in \mathbb{Z}_+^n \\ |\boldsymbol{\mu}+\boldsymbol{\nu}|\leq d}} \sum_{\mathbf {m,n}: |\Bf m|= |\Bf n|=r}\alpha''_{\boldsymbol\mu,\boldsymbol\nu}\ket{\mathbf{m}} \bra{\mathbf{n}} \bra{\mathbf{m}} a^{\boldsymbol\mu} a^{\dagger \boldsymbol\nu}\ket{\Bf n}.
\end{equation} 
 
Furthermore, $\tilde H$ is polynomially sparse. That is due to the fact that each term $a^{\boldsymbol\mu} a^{\boldsymbol\nu}{}^\dagger$ is 1-sparse in the Fock basis, and that there are $\binom{n+d-1}{d}= n^{O(d)}\in\mathsf{poly}(n)$ many terms in the summation Eq.~\eqref{eq:conjugated-Hamiltonian}. Therefore, we are now left with finding the ground energy of an exponential size, efficiently row-computable, sparse Hamiltonian. By a standard Hamiltonian simulation algorithm \cite{berry2015hamiltonian}, if we are provided the ground state, we can run a Hamiltonian simulation protocol in polynomial time and find the ground energy.
By \cref{prop:boundG}, the displacement/squeezing parameters are polynomially/logarithmically bounded, and based on this argument, we can show that $\wt{H}$ will have polynomially bounded norm and hence standard \textbf{QMA} protocols apply.

We now detail the proof that $\tilde H$ has polynomially bounded norm. Note that $H$ can be written as
\begin{align}
H = \sum_{k=1}^d \boldsymbol\zeta^{\otimes k}{}^\dagger \Bf h_k \boldsymbol\zeta^{\otimes k},
\end{align}
where $\boldsymbol\zeta := (a_1,\cdots,a_n,a_1^\dagger,\cdots,a_n^\dagger)^T$. Using a displacement unitary $D(d_1)\otimes D(d_2)\otimes \cdots D(d_n)$ with $\Bf d\in\mathbb C^n$, followed by Gaussian witness of the form $V\cdot (S(\xi_1)\otimes\cdots\otimes S(\xi_n))$ where $V$ is an interferometer and $S$ are squeezing gates with parameters $\boldsymbol\xi\in\mathbb C^n$, we get the following transformation on $\zeta$
\begin{align}\label{eq:mapping-params-after-conjugation}
\boldsymbol\zeta \mapsto  (\Bf R \, \mathrm{diag}(\boldsymbol\xi) \Bf R^\dagger) \Bf U_V\boldsymbol\zeta + \Bf d,
\end{align}
where $\Bf R = \frac{1}{\sqrt2}\begin{pmatrix}
\mathbb I & i\mathbb I\\\
\mathbb I & -i\mathbb I
\end{pmatrix}$, is the matrix transforming from the position and momentum basis to the annihilation and creation basis.
Therefore, we get that the new Hamiltonian can be expressed as
\begin{align}
H' = \sum_{k=1}^d \boldsymbol\zeta^{\otimes k}{}^\dagger \Bf h'_k \boldsymbol\zeta^{\otimes k}
\end{align}
with
\begin{align}\label{eq:new-H}
\Bf h'_k = \sum_{j\geq k} \sum_{\pi\in \mathcal S_j} \pi\left( \Bf T^{\otimes k}{}^\dagger \otimes \mathbf d^{\otimes j-k} {}^\dagger  \right) \cdot \mathbf h_j \cdot \pi\left( \Bf T^k \otimes \mathbf d^{\otimes j-k} \right),
\end{align}
where $\mathcal S_j$ is the set of permutation matrices over $j$ tensor factors, and $\Bf T = \Bf R\, \mathrm{diag}(\boldsymbol\xi) \Bf R^\dagger \Bf U_V$. Notice that there are constantly many terms in summation formula Eq.~\eqref{eq:new-H}. Moreover, the norm of each term on the right-hand-side of Eq.~\eqref{eq:new-H} is at most $\max_{j\in[d]}(\norm{\mathbf h_j} \cdot \norm{\Bf d}^{2(j-k)}) \cdot \norm{\Bf T}^{2k}$. Since $\Bf h_k$ by assumption has entries bounded by a constant, we get that $\norm{\Bf h_j}\leq O(n)$ for all $j\in[d]$, and moreover, by assumption $\norm{\Bf d}\leq \mathsf{poly}(n)$. Also, the assumption $|\xi|\leq O(\log (n))$ implies that $\norm{\mathbf{T}}\leq \mathsf{poly}(n)$. These will together imply that $\norm{\Bf h'}\leq \mathsf{poly}(n)$. Finally, since $\norm{\boldsymbol\zeta^{\otimes k}\ket{\Bf m}}\leq r^{O(d)}\leq \mathsf{poly}(n)$, (with $r$ being the stellar rank and $d$ being the degree of $H$) we have that the entries of $\wt{H}$ are polynomially bounded, and hence, it can be simulated within $\mathbf{BQP}$.
\end{proof}

With a similar line of reasoning, we show that optimizing energy over states of constant stellar rank is complete for $\mathbf{NP}$ (which generalizes . 

\begin{theorem}\label{thm:bounded-squeezing-and-displacementNP}
We have that $\mathsf{CVLH}_{\mathcal S_r^{\mathsf{exp}}}$ is in $\mathbf{NTIME}(n^{O(r)})$.
\end{theorem}

\begin{proof}
To show that the problem is $\mathbf{NP}$, we let the prover provide Gaussian parameters to the verifier. Then, the verifier can efficiently compute the conjugated Hamiltonian as in \eqref{eq:new-H}. Note that the mapping \eqref{eq:mapping-params-after-conjugation} now allows $|\xi_i|\leq \mathsf{poly}(n)$ and $|d_i|\leq \mathsf{exp}(n)$, which yields in the new coefficients (i.e., elements of $\mathbf h'_k$ in \eqref{eq:new-H}) to be at most exponentially large. We are then left with computing the minimum eigenvalue of $\Pi_r H' \Pi_r$ with $r$ being the stellar rank. This matrix is of $n^{O(r)} = \mathsf{poly}(n)$ size, and has entries that are at most $\mathsf{exp}(n)$ large. Therefore, we can compute its eigenvalues with inverse polynomial precision in polynomial time.
\end{proof}

Note that the reason that we are fine considering exponential energy in the $\mathbf{NP}$ algorithm is that one can compute eigenvalues of polynomial size matrices with exponential accuracy in polynomial time. However, this is not the case for the local Hamiltonian problem, since requiring exponentially small precision increases the complexity from $\mathbf{QMA}$ to $\mathbf{PSPACE}$ \cite{deshpande2022importance}.

\subsubsection{Optimization over states of arbitrary stellar rank}

In the most general setting, we are interested to know the hardness of estimating the ground energy over all states, i.e., when $r=\infty$ with no bound on average particle number, i.e., over the states in $\mathcal H= \mathcal S_\infty^\infty$. As stated in the following theorem, this problem is undecidable.

\begin{theorem}\label{thm:undecidability}
$\mathsf{CVLH}^8_{\mathcal H}$ is $\mathbf{RE}$-hard, and therefore undecidable.
\end{theorem}
This result implies the uncomputability of the ground energy of a generic polynomial bosonic Hamiltonian of constant degree. The proof of this result relies on a reduction from Hilbert's Nullstellensatz problem over integers, which we will denote as $\mathsf{HN}(\mathbb Z)$. An instance of $\mathsf{HN}(\mathbb Z)$ is to decide whether or not a polynomial with integer coefficients has a solution over integers. We refer to \cref{sec:constraint-satisfaction} for a more detailed discussion on this problem.

\begin{proof}[Proof of \cref{thm:undecidability}]
Let $P\in\mathbb Z[x_1,\cdots,x_n]$ be a polynomial with integer coefficients. Then, $P$ has a solution over integers, if and only if $\min_{\mathbf{x}\in\mathbb Z^n}P^2(\mathbf{x})=0$. The latter problem is also the same as deciding $\min_{\mathbf{s}\in\{\pm1\}^n,\mathbf{x}\in\mathbb Z_+^n } P^2(s_1 x_1,\cdots, s_n x_n)=0$. Therefore, if one can decide `satisfiability' over non-negative integers, then, they can apply that algorithm for all possible $\mathbf s\in\{-1,+1\}^n$ and decide satisfiability over all integers. 

It now remains to reduce the problem of minimization of a polynomial over $\mathbb Z_+^n$ to $\mathsf{CVLH}$. Following \cite{kieu2003quantum}, for any polynomial $Q\in\mathbb Z_+[x_1,\cdots,x_n]$, the spectrum of the operator
\begin{align}
H_Q:= Q(N_1,\cdots,N_n),
\end{align}
with $N_i$ being number operators, is the same as the image of $Q$. Hence, $\lambda_{\min}(H_Q)=\min_{\mathbf{x}\in\mathbb Z_+^n} Q(\mathbf{x})$. Regarding the spectral gap of $H_Q$, one should note that the image of $Q$ in contained within $\mathbb Z$. Hence, the spectral gap of $H_Q$ is bounded below by $1$. Setting $Q=P^2$ concludes the proof for constant degree.

As it turns out, there exist undecidable polynomial equations over non-negative integers with degree $4$ and constant number of unknowns \cite{jones58three}, so the above construction also implies that $\mathsf{CVLH}^8_{\mathcal H}$ is $\mathbf{RE}$-hard.
\end{proof}

Note that this proof is still valid under constant energy bound, and also implies that $\mathsf{CVLH}^8_{\mathcal F}$ is $\mathbf{RE}$-hard, where $\mathcal F$ is the set of Fock states (indeed, we can enumerate over natural numbers to put the problem inside $\mathbf{RE}$, and hence, $\mathsf{CVLH}^8_{\mathcal F}$ is $\mathbf{RE}$-complete).

Now we extend the undecidability result to the boundedness problem:
\begin{corollary}\label{corol:boundedness}
The boundedness problem over the entire CV Hilbert space, denoted by $\mathsf{HBound}_{\mathcal H}^d$ for $d\geq10$ is undecidable. It is in fact $\mathbf{co}$-$\mathbf{RE}$-hard.
\end{corollary}

\begin{proof}
We reduce the ground energy problem to the boundedness problem, and apply \cref{thm:undecidability}. To do so, let $H$ be a Hamiltonian whose ground energy is either less than $a$ or larger than $b$. Let $c = (a+b)/2$, and consider the following Hamiltonian
\begin{align}
H' = N \otimes (H-c).
\end{align}
Note that:
\begin{align}
\begin{cases}
H' \,\text{is unbounded from below if} &\quad \lambda_{\min}(H) \leq a,\\
H' \,\text{is bounded from below if} &\quad \lambda_{\min}(H) \geq b,
\end{cases}
\end{align}
which simply reduces the ground energy problem to boundedness.
\end{proof}

\subsection{Continuous-variable history state construction}
\label{sec:history}

This section discusses prospects and challenges in defining a continuous-variable history state construction, with the aim of proving \textbf{CVQMA}-completeness of variants of the continuous-variable local Hamiltonian problem \textsf{CVLH}.

The celebrated \textbf{QMA}-completeness result for the local Hamiltonian problem (see, e.g., \cite{kempe20033,kempe2006complexity}) is based on the so-called history state construction: given a \textbf{BQP} verifier circuit, we design a Hamiltonian such that its ground state contains the history of the steps of the computation. We have an extra register to encode time, and the ground state is $\frac{1}{\sqrt{T+1}} \sum_{t = 0}^T \ket{t}\ket{\psi_t}$, where $\ket{\psi_t}$ is the quantum state at time step $t$ of the \textbf{QMA} protocol. The Hamiltonian $H = H_{\mathrm{in}} + H_{\mathrm{evol}} + H_{\mathrm{final}}$ consists of three main ingredients: $H_{\mathrm{in}}$ initializes the verification protocol at time $t = 0$ (e.g., setting up ancillas), $H_{\mathrm{evol}}$ simulates the verification circuit for time $0\le t \le T$, and $H_{\mathrm{final}}$ simulates the final measurement. We can view this construction as the quantum generalization of the Cook--Levin reduction for \textbf{NP}-completeness of 3-\textsf{SAT}. 

Let us discuss an attempt in constructing a history of computation for $\mathbf{CVBQP}$. Recall from \cref{def:CVBQPfixednG} that such a computation consists of a sequence of unitary gates generated by constant-degree polynomial Hamiltonians $H_1,\dots,H_T$ applied to the vacuum, followed by a measurement of a single-mode polynomial Hamiltonian $H_{\mathrm{meas}}$ of constant degree. Similar to \cite{kempe2006complexity} we define
\begin{equation} 
H = H_{\mathrm{in}} + H_{\mathrm{evol}} + H_{\mathrm{final}}.
\end{equation} 
We also define a clock register $c$ to store time steps in the Fock basis. For instance, the time step $t = 1$ is captured by a single boson in this mode and so on. We now reflect on ways to construct $H_{\mathrm{in}}$ and $H_{\mathrm{final}}$ and point out the first fundamental difference with discrete variables. In the discrete-variable construction we can simply project onto $\ket{t = 0}_c$ and set the initial ancillae to zeros by defining $H_{\mathrm{in}} := \ket{0}_c\!\bra{0} \otimes H_0$, where $H_0 = \ket{0\cdots 0}_{\mathrm{ancillae}}\bra{0\cdots 0}$. In discrete variables, we can simply use linear operators such as $(I + Z)/2$ to define these projectors, but in the continuous-variable case, the projector $\ket{0}\!\bra{0}$ is not a polynomial in the position and momentum operators. For the final Hamiltonian we have the same problem: it is tempting to define $H_{\mathrm{final}} := \ket{T}_c\!\bra{T} \otimes H_{\mathrm{meas}}$, but $\ket{T}_c\!\bra{T}$ is not a polynomial in the position and momentum operators.

Instead, we may set $H_0=\sum_{i \in \text{ancillae}} N_i$ in the continuous-variable case to set the ancillae to the vacuum, but it is challenging to derive an unambiguous counter using polynomial Hamiltonians, since no polynomial Hamiltonian apart from constant ones is a projector.

Assuming this hurdle can be overcome, for instance based on graph gadgets as in \cite{childs2014bose}, the next difficulty lies in defining $H_{\mathrm{evol}}$. Up to a rescaling, the construction in \cite{kempe2006complexity} uses $H_{\mathrm{evol}} = \sum_{t = 0}^{T-1} H_{\mathrm{evol}}^{(t)}$, where $H^{(t)}_{\mathrm{evol}} = \ket{t+1}_c\! \bra{t+1} + \ket{t}_c \!\bra{t} - \ket{t+1}_c\! \bra{t} \otimes u_t - \ket{t}_c \!\bra{t+1} \otimes u^\dagger_t$, and $u_t$ is the $t^{th}$ gate in the circuit. In the continuous-variable case, the terms $u_t$ would not be polynomial Hamiltonians, even for gates as simple as displacement. We leave a detailed construction to future work.

\bibliography{ref}

\clearpage

\appendix

\begin{center}
{\Large Appendix}
\end{center}

In this appendix, we provide proofs of technical results that are used throughout the manuscript.

\section{Moments of Gaussian states}
\label{app:innerstellar}

Below we provide analytic expressions for inner products involving Gaussian states based on their stellar function.

\begin{lemma}\label{lem:gaussian-generator}
Let $F^\star_\psi(z) = \mathcal N^{-1} \exp(-\frac{a}2z^2 + bz)$. It is the case that
\begin{align}
\langle F^\star_\psi, e^{ta} e^{s a^\dagger}  F^\star_\psi \rangle = \exp\left( \frac{-\frac12 a t^2 + st - \frac12 \conj a s^2 + (b-a\conj b) t + (\conj b - \conj a b)s}{1-|a|^2} \right).
\end{align}
\end{lemma}

\begin{proof}
Let $G^{(s)}_\psi := e^{sa^\dagger} F^\star_\psi$, which implies $G^{(s)}_\psi(z) = e^{sz} F^\star_\psi(z)$. Note that we have
\begin{align}
\langle F^\star_\psi, e^{ta} e^{sa^\dagger} F^\star_\psi\rangle = \langle G^{(t)}_\psi, G^{(s)}_\psi\rangle.
\end{align}
Using this identity, we can compute the quantity of our interest as follows
\begin{align}
\langle F^\star_\psi, e^{ta} e^{sa^\dagger} F^\star_\psi\rangle = \mathcal N^{-2}\int \exp\left( (t+\conj b)\conj z + (s+b) z -\frac{a}{2} z^2 - \frac{\conj a}2 \conj z^2 \right) \, \mathrm d\mu
\end{align}
Now, for the integration, let $\wt{z}:=e^{i\theta_a/2} z$ where $\theta_a$ is the complex argument satisfying $a = e^{i\theta_a}|a|$. This means that $\frac{a}2 z^2 = \frac{|a|}2 \wt{z}^2$, while keeping the measure $\mathrm d\mu$ unchanged i.e., $\mathrm d\mu = \frac{1}{\pi}e^{-|\wt{z}|^2} \mathrm d^2\wt{z}$. Consequently, we get
\begin{align}\label{eq:gaussian-generator}
\begin{split}
\langle F^\star_\psi, e^{ta} e^{sa^\dagger} F^\star_\psi\rangle &= \mathcal N^{-2} \int \exp\left( -\frac{|a|}2 (\wt z^2 + \conj{\wt z}^2) + (s+b)e^{-i\theta_a/2}\wt z + (t+\conj b)e^{i\theta_a/2} \conj{\wt z} \right) \mathrm d\mu\\
&= \mathcal N^{-2} \left(\frac{1}{\sqrt{\pi}} \int \exp\left( -(1+|a|)x^2 + \alpha_+ x \right) \, \mathrm dx\right) \\
& \qquad\quad\times\left(\frac{1}{\sqrt{\pi}} \int \exp\left( -(1-|a|)y^2+\alpha_- y \right) \, \mathrm dy\right)\\
&= \mathcal N^{-2}  \frac{1}{\sqrt{1-|a|^2}} \exp\left( \frac{\alpha_+^2}{4(1+|a|)} + \frac{\alpha_-^2}{4(1-|a|)}\right),
\end{split}
\end{align}
where $\alpha_+ = (b+s) e^{-i\theta_a/2} + (\conj{b} + t) e^{i\theta_a/2}$, and $\alpha_- = i(b+s) e^{- i\theta_a/2} - i (\conj b +t)e^{i \theta_a/2}$. Note that the expression above for $s=t=0$ must be zero, as $F^\star_\psi$ corresponds to a valid state (of unit norm). This gives
\begin{align}
1 = \mathcal N^{-2} \frac{1}{\sqrt{1-|a|^2}} \exp(\frac{|b|^2 - \frac12 (\conj a b^2 + a \conj b^2)}{1-|a|^2}).
\end{align}
Plugging the normalization constant $\mathcal N$ into Eq.~\eqref{eq:gaussian-generator} gives
\begin{align}
\langle F^\star_\psi, e^{ta} e^{sa^\dagger} F^\star_\psi\rangle = \exp\left( \frac{st + bt+\conj b s -\frac12(at^2 + \conj a s^2 + 2\conj a b s+ 2a\conj b t)}{1-|a|^2} \right).
\end{align}
Rearranging the terms gives the desired formula.
\end{proof}

\begin{remark}
Note that as a byproduct, the proof above gives
\begin{align}
\mathcal N = \frac{1}{(1-|a|^2)^{1/4}} \exp\left(  
\frac{|b|^2 - \frac12 (\conj a b^2 + a \conj b^2)}{2(1-|a|^2)}
\right)
\end{align}
for the normalized state $F^\star_\psi(z) = \mathcal N^{-1} \exp(-az^2+bz)$.
\end{remark}

Note that the above lemma allows us to compute expectation values of polynomials as 
\begin{align}
\langle F^\star_\psi, (a^\dagger)^k a^j F^\star_\psi \rangle &= \partial_s^k \partial_t^l \left(\langle F^\star_\psi, e^{ia^\dagger s} e^{ita} F^\star_\psi \rangle \right)|_{s,t=0}.
\end{align}

The following corollary summarizes a few of the results that can be obtained by this approach.

\begin{corollary}\label{corol:gaussian-expectations}
Let $F^\star_\psi(z) = \mathcal N^{-1} \exp(-\frac12 az^2+bz)$. We have the following
\begin{enumerate}
\item $\langle F^\star_\psi, a F^\star_\psi\rangle = \frac{b-a\conj b}{1-|a|^2}$,
\item $\langle F^\star_\psi, a^2 F^\star_\psi\rangle = \frac{-a}{1-|a|^2} + \frac{(b-a\conj b)^2}{(1-|a|^2)^2}$,
\item $\langle F^\star_\psi, aa^\dagger F^\star_\psi\rangle = \frac{1}{1-|a|^2} + \frac{|b-a\conj b|^2}{(1-|a|^2)^2}$.
\end{enumerate}
\end{corollary}

\begin{lemma}\label{lem:generatingFunction-multimode}
Let $F^\star_\psi(\Bf z) = \mathcal N^{-1} \exp(-\frac12 \Bf z^T \Bf A \Bf z + \Bf b^T \Bf z)$. We have
\begin{align}
\langle F^\star_\psi, e^{\sum_i t_i a_i} e^{\sum_i s_i a_i^\dagger} F^\star_\psi\rangle =  \exp(-\frac12
\begin{pmatrix}
\Bf s^T & \Bf t^T
\end{pmatrix}
\Sigma
\begin{pmatrix}
\Bf s\\
\Bf t
\end{pmatrix}
+ \Bf c^T \Bf t + \conj{\Bf c}^t \Bf s)
\end{align}
where
\begin{align}
\Sigma = \begin{pmatrix}
\Bf U^\dagger & 0\\
0 & \Bf U^T
\end{pmatrix}
\begin{pmatrix}
\frac{\Bf D}{1-\Bf D^2} & -\frac{1}{1-\Bf D^2}\\
-\frac{1}{1-\Bf D^2} & \frac{\Bf D}{1-\Bf D^2}
\end{pmatrix}
\begin{pmatrix}
\conj{\Bf U} & 0\\
0 & \Bf U
\end{pmatrix}, \Bf c = \Bf U^T \frac{1}{1-\Bf D^2} (\conj{\Bf U} \Bf b - \Bf D \Bf U \conj{\Bf b}),
\end{align}
and $\Bf A = \Bf U^T \Bf D \Bf U$ is the Takagi factorization of $\Bf A$ \cite{horn_johnson_1985}.
\end{lemma}

\begin{proof}
Similar to the proof of \cref{lem:gaussian-generator}, we define $G^{(\Bf s)}_\psi := e^{\sum_i s_i a_i^\dagger} F^\star_\psi$, and write
\begin{align}
\langle F^\star_\psi, e^{\sum_i t_i a_i} e^{\sum_i s_i a_i^\dagger} F^\star_\psi\rangle &= \langle G^{(\Bf t)}_\psi, G^{(\Bf s)}_\psi\rangle\\
&= \mathcal N^{-2} \int \exp(-\frac12\Bf z^T \Bf A \Bf z - \frac12\conj{\Bf z}^T \conj{\Bf A} \conj{\Bf z}+ (\Bf b + \Bf s)^T \Bf z + (\conj{\Bf b} + \Bf t)^T\conj{\Bf z} ) \mathrm d\mu.
\end{align}
Using Takagi's factorization theorem \cite{horn_johnson_1985}, we have $\Bf A = \Bf U^T \Bf D \Bf U$, for a unitary matrix $\Bf U$, and a diagonal non-negative real matrix $\Bf D$. Note that the fact that $F^\star_\psi\in L^2(\mu)$ implies that $\Bf D \leq \mathbb I$. Now, let $\wt{\Bf z} := \Bf U \Bf z$. This transformation preserves the measure. We get
\begin{align}
\langle F^\star_\psi, e^{\sum_i t_i a_i} e^{\sum_i s_i a_i^\dagger} F^\star_\psi\rangle = \mathcal N^{-2} \prod_{i=1}^n \int \exp(-\frac12 d_i (z_i^2 + \conj{z}_i^2) + (b'_i + s'_i) z + (\conj{b'}_i+ t'_i)\conj z_i) \mathrm d\mu,
\end{align}
where $\Bf b' = \conj{\Bf U} \Bf b$, $\Bf s' = \conj{\Bf U} \Bf s$, and $\Bf t' = \Bf U \Bf t$, and $d_i$ are the diagonal elements of $D$. Now, using the result of \cref{lem:gaussian-generator}, we get
\begin{align}
\langle F^\star_\psi, e^{\sum_i t_i a_i} e^{\sum_i s_i a_i^\dagger} F^\star_\psi\rangle = \mathcal{N}^{-2} \prod_{i=1}^n \frac{1}{\sqrt{1-d_i^2}} \exp( \frac{(\alpha^{(i)}_+)^2}{4(1+d_i)} + \frac{(\alpha_-^{(i)})^2}{4(1-d_i)} ),
\end{align}
where $\alpha^{(i)}_+ = 2\Re(b_i') + s'_i + t'_i$, and $\alpha^{(i)}_- = -2\Im(b_i') + is'_i - i t'_i$. Using the fact that the value of the expression above should be $1$ when $s=t=0$, we get
\begin{align}
1 =  \frac{\mathcal N^{-2}}{\sqrt{\det(1-\Bf{D}^2)}}  \exp(\sum_{i=1}^n \frac{|b'_i|^2 - \frac{d_i}2(b'_i{}^2 + \conj{b'}_i^2) }{1-d_i^2}),
\end{align}
which gives
\begin{align}
\mathcal N = \frac{\exp(\frac12 \conj{\Bf b}^T \Bf U^T \frac{1}{1-\Bf D^2} \conj{\Bf U} \Bf b - \frac14 \Bf b^T \Bf U^\dagger \frac{\Bf D}{1-\Bf D^2} \conj{\Bf U} \Bf b - \frac14 \conj{\Bf b}^T \Bf U^T \frac{\Bf D}{1-\Bf D^2} \Bf U \Bf b)} {(\det(1-\Bf D^2))^{1/4}}
\end{align}
Plugging this into the above formula results in
\begin{align}
\begin{split}
\langle F^\star_\psi, e^{\sum_i t_i a_i} e^{\sum_i s_i a_i^\dagger} F^\star_\psi\rangle = \exp\Bigg(&\Bf s^T \Bf U^\dagger \frac{1}{1-\Bf D^2} \Bf U \Bf t + \Bf b^T \Bf U^\dagger \frac{1}{1-\Bf D^2} \Bf U \Bf t + \conj{\Bf b}^T \Bf U^T \frac{1}{1-\Bf D^2} \conj{\Bf U} \Bf s  \\
&\,-\frac12\big(  \Bf t^T \Bf U^T \frac{\Bf D}{1-\Bf D^2} \Bf U \Bf t + \Bf s^T \Bf U^\dagger \frac{\Bf D}{1-\Bf D^2} \conj{\Bf U}\Bf s\\
&+ 2 \Bf b^T \Bf U^\dagger \frac{\Bf D}{1-\Bf D^2} \conj{\Bf U} \Bf s + 2 \conj{\Bf b}^T \Bf U^T \frac{\Bf D}{1-\Bf D^2} \Bf U \Bf t\big) \Bigg)
\end{split}
\end{align}
\end{proof}

A direct consequence of this formula is the following corollary, which implies \cref{corol:gaussian-deg2-exp} in the main text.

\begin{corollary}\label{corol:gaussian-deg2-exp2}
Let $F^\star_\psi(z) = \exp(-\frac12\Bf z^T\Bf A \Bf z + \Bf b^T \Bf z)$, where $\Bf A = \Bf U^T \Bf D \Bf U$. It is the case that
\begin{enumerate}
\item $(\langle F^\star_\psi, a_i F^\star_\psi\rangle)_i =  \Bf c$,
\item $(\langle F^\star_\psi, a_i a_j F^\star_\psi\rangle)_{ij} = \Bf U^T \frac{-\Bf D}{1- \Bf D^2} \Bf U + \Bf c \Bf c^T$,
\item $(\langle F^\star_\psi, a_ja_i^\dagger F^\star_\psi\rangle)_{ij} = \Bf U^T \frac{1}{1-\Bf D^2} \conj{\Bf U} + \Bf c^T \conj{\Bf c}$,
\item For degree-4 terms that conserve the number of particles we have
\begin{align}
\begin{split}
\langle F^\star_\psi, a_ja_ia_j^\dagger a_i^\dagger F^\star_\psi\rangle =& (\Bf U^\dagger \frac{1}{1-\Bf D^2} \Bf U)_{jj} (\Bf U^\dagger \frac{1}{1-\Bf D^2} \Bf U)_{ii} + \left| \left(  \Bf U^\dagger \frac{1}{1-\Bf D^2} \Bf U\right)_{ij} \right|^2 + \left| \left(  \Bf U^\dagger \frac{D}{1-\Bf D^2} \conj{\Bf U}\right)_{ij} \right|^2\\
& -2\Re\left( \Bf c_i \Bf c_j\left( \Bf U^\dagger \frac{\Bf D}{1-\Bf D^2}\conj{\Bf U} \right)_{ij} \right)
 + 2\Re\left( \Bf c_i  \conj{\Bf c_j} \left( \Bf U^T \frac{1}{1-\Bf D^2} \conj{\Bf U}\right)_{ji}\right)\\
 &+ |\Bf c_i|^2 \left( \Bf U^\dagger \frac{1}{1-\Bf D^2} \Bf U\right)_{jj}
 + |\Bf c_j|^2 \left( \Bf U^\dagger \frac{1}{1-\Bf D^2} \Bf U\right)_{ii} + |\Bf c_i|^2 \cdot |\Bf c_j|^2
\end{split}
\end{align}
\end{enumerate}
where $\Bf c = \Bf U^T \frac{1}{1-\Bf D^2} (\conj{\Bf U} \Bf b - \Bf D \Bf U \conj{\Bf b})$.
\end{corollary}

\section{Commutation relations for powers of position and momentum operators}
\label{app:commutation}

Below, we provide the general commutation relation between arbitrary powers of position and momentum operators $X$ and $P$, which leads to a proof of \cref{prop:SOS}.

\begin{lemma}\label{lem:app-commutation}
We have the following
\begin{align}
[X^\mu,P^\nu] = -\sum_{\lambda=1}^{\min(\mu,\nu)} (-i)^{\lambda}\frac{\mu! \nu!}{(\mu-\lambda)!(\nu-\lambda)!\lambda!}  \, X^{\mu-\lambda} P^{\nu-\lambda}.
\end{align}
\label{lem:recurs}
\end{lemma}

\begin{proof}
We prove this inductively.
For $\mu=1$ we have
\begin{align}
[X,P^\nu] = i\nu P^{\nu-1},
\end{align}
which is what we get from \cref{lem:recurs} as well (same for $\nu=1$ and arbitrary $\mu$). Now, assume the relation works for $\mu-1,\nu$ (where $m=\min(\mu,\nu)>2$), we have
\begin{align}
\begin{split}
[X^{\mu}, P^\nu] &= [X^{\mu-1}X, P^{\nu}]\\
&= [X^{\mu-1}, P^{\nu}] X + X^{\mu-1} [X,P^{\nu}]\\
&= [X^{\mu-1},P^\nu]X + i\nu X^{\mu-1} P^{\nu-1}
\end{split}
\end{align}
where $m:=\min(\mu,\nu)$. Now let $[X^{\mu},P^\nu] = \sum_{\lambda} \alpha_{\lambda}^{\mu,\nu} X^{\mu-\lambda} P^{\nu-\lambda}$. We have
\begin{align}
\begin{split}
[X^\mu,P^\nu] &= \sum_{\lambda=1}^{m-1}  \alpha_{\lambda}^{\mu-1,\nu} X^{\mu-\lambda-1} P^{\nu-\lambda} X + i\nu X^{\mu-1} P^{\nu-1}\\
&= \sum_{\lambda=1}^{m-1} \alpha_{\lambda}^{\mu-1,\nu} X^{\mu-\lambda} P^{\nu-\lambda} + \sum_{\lambda=1}^{m-1} \alpha_{\lambda}^{\mu-1,\nu} X^{\mu-\lambda-1} [P^{\nu-\lambda}, X] + i\nu X^{\mu-1} P^{\nu-1}\\
&= \sum_{\lambda=1}^{m-1} \alpha^{\mu-1,\nu}_\lambda X^{\mu-\lambda} P^{\nu-\lambda} - i \sum_{\lambda=2}^{m} (\nu-\lambda+1)\alpha_{\lambda-1}^{\mu-1,\nu} X^{\mu-\lambda} P^{\nu-\lambda-1} + i\mu X^{\mu-1} P^{\nu-1},
\end{split}
\end{align}
which gives us the following recursion formula
\begin{align}
\alpha_{\lambda}^{\mu,\nu} = \begin{cases}
\alpha_{\lambda}^{\mu-1,\nu} + i\nu & \text{if } \lambda=1\\
\alpha_{\lambda}^{\mu-1,\nu} -i (\nu-\lambda+1) \alpha_{\lambda-1}^{\mu-1,\nu} & \text{if }\lambda\geq 2.
\end{cases}
\end{align}
Using our hypothesis, we have $\alpha_{1}^{\mu,\nu} = i\mu\nu$, which clearly satisfies the recursion above. For $\lambda\geq 2$, we have
\begin{align}
\begin{split}
\alpha_{\lambda}^{\mu-1,\nu} -i(\nu-\lambda+1) \alpha_\lambda^{\mu-1,\nu} &= -(-i)^{\lambda}\frac{(\mu-1)!\nu!}{(\mu-\lambda-1)!(\nu-\lambda)!\lambda!}\\
&\qquad\qquad- (-i)^{\lambda} (\nu-\lambda+1)\frac{(\mu-1)!\nu!}{(\mu-\lambda)! (\nu-\lambda+1)!(\lambda-1)!} \\
&= -(-i)^{\lambda} \frac{(\mu-1)!\nu!}{(\mu-\lambda)!(\nu-\lambda)!\lambda!}\left( \mu-\lambda+\lambda \right)\\
&= -(-i)^{\lambda} \frac{\mu!\nu!}{(\mu-\lambda)!(\nu-\lambda)!\lambda!}\\
&= \alpha^{\mu,\nu}_\lambda,
\end{split}
\end{align}
which shows that our formula works.
\end{proof}

We now provide the proof of \cref{prop:SOS}. We first show the result for a single mode in the next lemma. Indeed, \cref{lem:coef-cannonical} shows that we can express commutations in terms of anti-commutations. Given that normal forms are sufficient for expressing polynomial Hamiltonians, we can conclude that anti-commutators also form a basis for polynomial Hamiltonians.

\begin{lemma}\label{lem:coef-cannonical}
It is the case that
\begin{align}
[X^s,P^r] = i\sum_{m<s, n<r} \alpha_{m,n}^{(r,s)} \{X^m,P^n\},
\end{align}
where $\alpha_{m,n}^{(r,s)}$ are all positive and can be computed in time $O(r^3\cdot s^3)$.
\end{lemma}
\begin{proof}
The proof is, of course, by induction. Let us assume it holds for $(r,s)$, we show that it holds for $(r+1,s)$. A similar argument can be used to raise $s$.
\begin{align}
\begin{split}
[X^{r+1},P^s] &= [X\cdot X^{r}, P^s] = isP^{s-1} X^r + X\cdot [X^r,P^s]\\
&= isP^{s-1} X^r + \sum_{m<r,n<s} i\alpha^{(r,s)}_{m,n} (X^{m+1} P^n + X P^n X^m)\\
&= isP^{s-1} X^r - \sum_{m,n} \alpha^{(r,s)}_{m,n} n P^{n-1} X^m + i\sum_{m<r,n<s} \alpha^{(r,s)}_{m,n} (X^{m+1} P^n + P^n X^{m+1}).
\end{split}
\end{align}
As the left hand side in the equation above is anti-Hermitian, we can subtract its adjoint from it and devide by two, and still have the same operator. Thus
\begin{align}
\begin{split}
[X^{r+1},P^s] &= \frac{i}2s \{P^{s-1}, X^r\} + \frac12\sum_{m,n} \alpha^{(r,s)}_{m,n} [X^m,P^{n-1}] + i\sum_{m,n} \alpha^{(r,s)}_{m,n} \{X^{m+1},P^n\}\\
&= \frac{i}2s \{P^{s-1}, X^r\} + \frac{i}2\sum_{u<r-1,v<s-1}\{X^u,P^v\} \, \left(\sum_{u<m,v<n} \alpha^{(r,s)}_{m,n} \alpha^{(m,n)}_{u,v}\right)\\
&\quad + i \sum_{m<r,n<s}\alpha^{(r,s)}_{m,n} \{X^{m+1},P^n\}.
\end{split}
\end{align}
Note that this is a proof based on induction on $r+s$. The statement is trivial for $(r,s)=(0,1), (r,s)=(1,0)$ that are our base cases.
\end{proof}

We need one more step, provided in the following lemma.

\begin{lemma}\label{lem:basic-matrix}
Let $\{A_i:i\in[n]\}$, $\{B_i: i\in[n]\}$ be two sets of operators, and $\alpha,\beta\in\mathbb C$. We have that any operator of the form $\sum_{i} \sum_{S_i\in\{A_i,B_i\}} \alpha_{S_i,i} \bigotimes S_i $ can be expressed as a linear combination of tensor products of commutators and anti-commutators of $A_i,B_i$'s.
\end{lemma}
\begin{proof}
Let $\ket{+}:=A+B$ and $\ket- := A-B$. With a abuse of notation, we know that $\ket-,\ket+$ and their tensor products form a basis for the space $(\mathbb C^2)^{\otimes n}$. Using this simple fact, we conclude the result.
\end{proof}

\noindent Note that given both results above, \cref{prop:SOS} follows immediately.

\section{Position and momentum moments for the vacuum}
\label{app:XP-vaccuum-exp}

In this section, we compute an expression for $\langle0|X^a P^b|0\rangle$. To do this we evaluate the generating function corresponding to these moments
\begin{align}
    \langle0|e^{iX s} e^{i P t}|0\rangle = e^{-\frac{s^2 + t^2 - 2 ist}{4}}.
    \label{eq:gen-func}
\end{align}
To see this we evaluate $e^{i s P} \ket {0} = \ket {-i \frac{s}{\sqrt{2}}}$ and $e^{-i X t} \ket{0} = \ket{-\frac{t}{\sqrt{2}}}$ (coherent states). Therefore $\langle0|e^{iX t} e^{i P t}|0\rangle = \braket{-t/\sqrt{2}}{is/\sqrt{2}}$. We then use the formula for the inner product between coherent states
$$
\braket{\alpha}{\beta} = e^{\frac{|\alpha - \beta|^2 + (\conj \alpha \beta - \conj \beta\alpha)}{2}}.
$$
We note that Eq.~\eqref{eq:gen-func} corresponds to the characteristic function of coupled Gaussian variables with zero mean and covariance matrix $\Sigma = \begin{pmatrix}
    1/2 & - i/2\\
    -i/2 & 1/2
\end{pmatrix}$.
As a result, the moment $\langle0|X^a P^b|0\rangle$ is the $a , b$ moment of coupled Gaussian variables with covariance matrix $\Sigma$. However, $\Sigma$ is not PSD. 

For simplicity we first evaluate the case where $l=0$. Then
$$
\bra{0}X^k\ket{0} = i^{-k} \frac{d^k}{ds^k} (e^{-\frac{s^2}{4}})|_{s = 0}.
$$
In order to evaluate the derivatives in this expression, we will use Hermite's polynomials. The $n$'th Hermite polynomial is defined as 
$$
He_n (x) = (-1)^n e^{\frac{x^2}{2}} \frac{d^n}{dx^n} (e^{-\frac{x^2}{2}} ).
$$
We obtain
$$
\bra{0}X^k\ket{0} = 2^{-k/2} (-1)^{k/2}He_k (0).
$$
We use the following expression for the Hermite polynomials at point $0$:
\begin{align}
    He_n (0) = 
\begin{cases}
    0 & n \text{ is odd},\\
    (-1)^{n/2} (n-1)!! & n \text{ is even}.
\end{cases}
\label{eq:Hermite-0}
\end{align}
We observe that when $k$ is odd, the expectation value is $0$. For $k$ even we obtain
\begin{align}
\bra{0} X^k \ket{0} = 2^{-k/2} (k-1)!! = 2^{-(k-1)} \frac{(k-1)!}{(k/2-1)!}.
\label{eq:X-vacuum-expectation}
\end{align}
Using Stirling estimation we obtain a $(\frac{k}{2e})^{k/2}$ growth in the leading term for this quantity up to polynomial factors. 

Next we get back to the computation of the general moment $\bra{0}X^kP^l\ket{0}$. In order to compute this quantity we decouple:
$$
\frac{s^2 + t^2 - 2 ist}{4} = \frac 18 (1-i) (s+t)^2 + \frac 18 (1+i) (s-t)^2 ,
$$
and define $X = s + t$ and $Y = s - t$. Therefore $\partial_s = \partial_X + \partial_Y$ and $\partial_t = \partial_X - \partial_Y$. Using these changes of variables 
\begin{align}
    M_{k,l} (s,t) := i^{- (k+l)}\partial^k_s \partial^l_t (e^{-\frac{s^2 + t^2 - 2 ist}{4}}) = (\partial_X + \partial_Y)^k (\partial_X - \partial_Y)^l e^{\frac 18 (1-i) X^2 + \frac 18 (1+i) Y^2 }
\end{align}
The $\bra{0}S^k P^l\ket{0}$ moment is equal to $M_{k,l} (0,0)$.
In order to evaluate this sum we use the Kravchuk polynomials:
\begin{align}
(\partial_X + \partial_Y)^{k} (\partial_X - \partial_Y)^l =  \sum_{m =0}^{k+l} K^{(k+l)}_m (l) \partial_X^{k+l-m}\partial_Y^m,
\end{align}
where the $m^{th}$ Kravchuk polynomial is defined as 
$$
K^{(n)}_m (l) = \sum_{j = 0}^{m} (-1)^j \binom lj\binom{n - l}{m - j},
$$
whenever $n\geq l$ and $n\geq m$.
For a review of the properties of these polynomials, see \cite{stanton1980some}. As a result,
$$
M_{k,l} (s,t) = i^{- (k+l)} \sum_{m =0}^{k + l} K^{(k+l)}_m (l) \partial_X^{k+l-m}[e^{-\frac 18 (1-i) X^2}] \partial_Y^m[e^{-\frac 18 (1+i) Y^2 }].
$$
Using Hermite polynomials we obtain
$$
\frac{\partial^n}{\partial X^n} (e^{-\frac{1}{8} (1-i) X^2}) = (-1)^n \epsilon^n_1 He_n(\epsilon_1 X) (e^{-\frac{1}{8} (1-i) X^2}),
$$
and
$$
\frac{\partial^n}{\partial Y^n} (e^{-\frac{1}{8} (1+i) Y^2}) = (-1)^n \epsilon^n_2 He_n(\epsilon_2 Y) (e^{-\frac{1}{8} (1+i) Y^2}),
$$
where $\epsilon_1 = \frac{1}{2^{3/4}}e^{-i 3/8 \pi}$ and $\epsilon_2 = \frac{1}{2^{3/4}}e^{i 3/8 \pi}$
therefore
\begin{align}
\begin{split}
M_{k,l} (s,t) &= i^{-(k+l)}\sum_{m =0}^n K^{(k+l)}_m (l) (-1)^{k + l - m} \epsilon^{k+l-m}_1 He_{k+ l - m}(\epsilon_1 (s+t))\\
&\qquad\qquad\qquad \times(-1)^m \epsilon^m_2 He_m(\epsilon_2 (s-t)) (e^{\frac{1}{8} (1+i) (s-t)^2}) e^{-\frac{s^2 + t^2 - 2 ist}{4}}.
\end{split}
\end{align}
Evaluating this expression at $s = t = 0$ we obtain
\begin{align}
\bra{0}X^k P^l \ket{0} = M_{k,l} (0,0) &=i^{-(k+l)} \sum_{m =0}^{k + l} K^{(k+l)}_m (l)  \frac{e^{\frac{3}{8}\pi i (2m- (k+l))}}{2^{\frac{3}{4} (k+l)}} He_{k+ l - m}(0)  He_m(0).
\end{align}
Using Eq.~\ref{eq:Hermite-0}, we immediately observe that $\bra{0}X^k P^l \ket{0} = 0$ if $k + l$ is odd, which confirms our intuition. Now, suppose $k + l = n$ is even. Then
\begin{align}\label{eq:genexpvacmon}
\bra{0}X^k P^l \ket{0} &= \frac{e^{\frac{1}{8}\pi i n}}{2^{\frac{3}{4}n}}  \sum_{m =0}^{n/2} (-i)^m K^{(n)}_{2m} (l)  (n - 2m -1) !! (2m - 1) !!.
\end{align}

\section{Bounds on Gaussian parameters from energy bounds}
\label{app:boundG}

We first introduce a lemma that allows us to put bounds on the amount of squeezing and displacement for a bounded-energy Gaussian state.

\begin{lemma}
Let $\ket\psi$ be a Gaussian state over $n$ modes with energy at most $E$, i.e., $\bra\psi \sum_i N_i\ket\psi\leq E$. It is the case that
\begin{align}
\ket\psi = U S(\xi) D(\alpha) \ket{0^n},
\end{align}
where $\alpha,\xi\in\mathbb C^n$, $U$ is a passive linear transform, $D$ is an $n$-mode displacement operator and $S$ an $n$-mode squeezing operator. Then for all $j=1,\dots,n$,
\begin{align}
|\xi_j|\le \frac12\log(4E+2),\qquad|\alpha_j|^2\le 2E(2E+1).
\end{align}
\end{lemma}

\noindent In particular, setting $E=\mathsf{poly}(n)$ implies a logarithmic bound on the squeezing and a polynomial bound on the displacement, while setting $E=\mathsf{exp}(n)$ implies a polynomial bound on the squeezing and an exponential bound on the displacement.

\begin{proof}
Note that $U$ does not change the expectation of the energy operator $N=\sum_jN_j$ (as passive linear operators elements preserve the total particle number). Hence, we have
\begin{align}
\bra\psi N\ket\psi = \bra{0^n} D(\alpha)^\dagger S(\xi)^\dagger N S(\xi) D(\alpha)\ket{0^n},
\end{align}
where the vector of squeezing parameters $\xi\in\mathbb R^n$ is taken to be real without loss of generality (since the phase can always be factored out and braided with the displacement operators and the vacuum is phase-invariant). Note that
\begin{align}
S(\xi_j)^\dagger (N_j+\frac12) S(\xi_j) = \frac12e^{2\xi_j} X_j^2 + \frac12e^{-2\xi_j} P_j^2.
\end{align}
Setting $\alpha_j=\frac{x_j+ip_j}{\sqrt2}$, we obtain that displacements act as $D(\alpha_j)^\dagger X D(\alpha_j) = X + x_j$ and $D(\alpha_j)^\dagger P_j D(\alpha_j) = P_j + p_j$. This results in
\begin{align}
D(\alpha_j)^\dagger S(\xi_j) (N + n/2) S(\xi_j) D(\alpha_j) = \frac12\left(\sum_j e^{2\xi_j} (X_j^2 + 2x_j X_j + x_j^2) + e^{-2\xi_j} (P_j^2 + 2p_j P_j + p_j^2)\right)
\end{align}
and hence
\begin{align}\label{eq:number-operator-evolution}
\bra\psi N\ket\psi &= \sum_j \frac12(\cosh(2\xi_j)-1) + \frac12e^{2\xi_j} x_j^2 + \frac12e^{-2\xi_j} p_j^2 
\end{align}
so $\bra\psi N\ket\psi\le E$ implies $\frac14e^{2|\xi_j|}-\frac12\le\frac12\cosh(2\xi_j)-\frac12\le E$, which gives
\begin{align}
|\xi_j| \leq\frac12\log(4E+2).
\end{align}
Moreover, we get
\begin{align}
\frac12e^{2\xi_j} x_j^2 \leq E, \quad \frac12e^{-2\xi_j} p_j^2 \leq E,
\end{align}
which together with $e^{2|\xi_j|}\leq 4E+2$ gives
\begin{align}
\frac12x_j^2+\frac12p_j^2\le2E\cosh(2\xi_j)\le 2E(2E+1).
\end{align}
\end{proof}

We may now generalize the bound on Gaussian states to finite stellar-rank states, of stellar rank $r$.

\begin{lemma}\label{conj:stellar-bound}
Let $\Pi_r = \sum_{r=0}^{r} \ket{n} \bra{n}$ where $\ket{n}$ are Fock basis elements. Let $x\in\mathbb R$ be any real number. It is the case that
\begin{align}
\Pi_r (X-x)^2\Pi_r \geq \Omega(r^{-1})>0.
\end{align}
\end{lemma}
\begin{proof}
Let $C_r$ denote the set of core states with stellar rank at most $r$. We have
\begin{align}
\lambda_{\min}(\Pi_r (X-x)^2\Pi_r) = \min_{\psi\in C_r} \bra{\psi} (X-x)^2 \ket{\psi} = \bra{\psi^\ast} (X-x)^2 \ket{\psi^\ast},
\end{align}
for some $\ket{\psi^\ast}\in C_r$. It is a classic result in statistics that $\langle (X-x)^2\rangle \geq \langle (X -\langle X\rangle)^2\rangle = \mathrm{Var}(X)$. Let us use the notation $\Delta x(\psi) := \sqrt{\mathrm{Var}(X)}$. Therefore, we have 
\begin{align}
\lambda_{\min}(\Pi_r (X-x)^2\Pi_r) = \Delta x(\psi^\ast)^2.
\end{align}
Now, we use the Heisenberg uncertainty principle, implying $\Delta x\geq \frac{1}{2\Delta p}$, to write
\begin{align}\label{eq:lambda-bound}
\lambda_{\min}(\Pi_r (X-x)^2\Pi_r) \geq \frac{1}{4\Delta p(\psi^\ast)^2}.
\end{align}
Finally, we note that
\begin{align}\label{eq:delta-p}
\Delta p(\psi^\ast)^2 = \langle P^2 - \langle P\rangle_{\psi^\ast}^2\rangle_{\psi^\ast} \leq \langle P^2\rangle_{\psi^\ast}\leq \norm{\Pi_r P^2 \Pi_r}.
\end{align}
Note that the $n$-th column of $P^2$ is
\begin{align}
(0,\cdots, -\frac12\sqrt{n(n-1)}, 0, n+\frac12, 0, -\frac12\sqrt{(n+1)(n+2)}, 0,\cdots,0)^T,
\end{align}
and hence, by Greshgorin's lemma (see e.g., \cite{horn_johnson_1985}), we obtain
\begin{align}
\norm{\Pi_r P^2 \Pi_r} \leq 2r( 1 + O(\frac1r)).
\end{align}
Combining this with \eqref{eq:lambda-bound} and \eqref{eq:delta-p} we get
\begin{align}
\lambda_{\min}(\Pi_r (X-x)^2\Pi_r) \geq \frac{1}{16r} (1 - O(\frac1r)).
\end{align}
\end{proof}

We present numerical computations provided in \cref{fig:plot1} and \cref{fig:plot2}. These plots indicate that the $1/r$ scaling is optimal.

\begin{figure}
    \centering
    \begin{subfigure}[b]{0.45\textwidth}
        \centering
        \includegraphics[width=\linewidth]{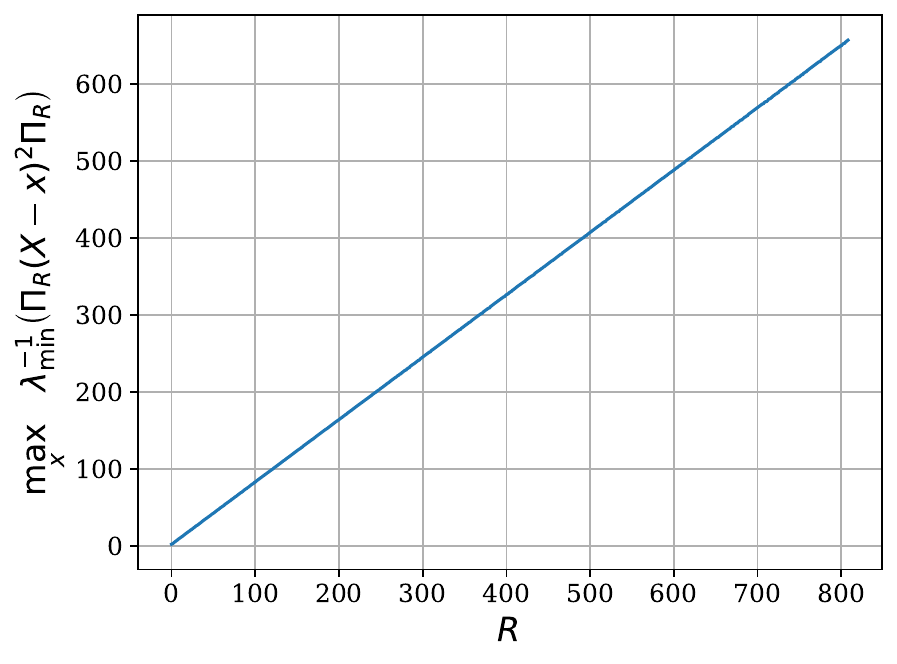}
        \caption{The numerical computation showing that  \cref{conj:stellar-bound} is tight, as the growth appears to be linear. The optimization is done using the \texttt{scipy} package.}
        \label{fig:plot1}
    \end{subfigure}
    \hfill
    \begin{subfigure}[b]{0.5\textwidth}
        \centering
        \includegraphics[width=\linewidth]{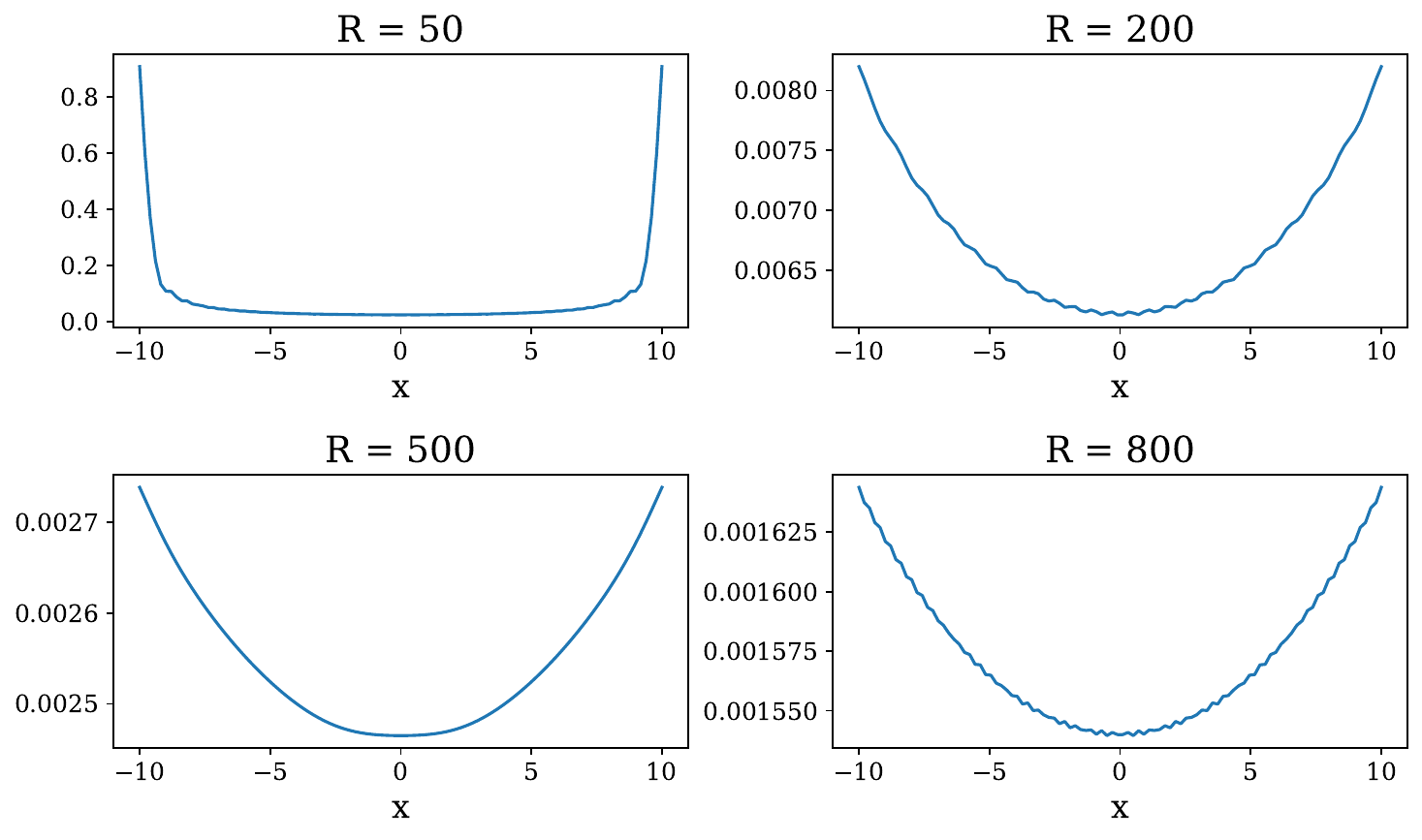}
        \caption{Illustration of a few examples of $r$. We can view that the value $\min_x \lambda_{\min}(\Pi_r (X-x)^2\Pi_r)$ is close to the value at $x=0$.}
        \label{fig:plot2}
    \end{subfigure}
    
    \caption{Numerical computations showing that \cref{conj:stellar-bound} is tight.}
    \label{fig:combined}
\end{figure}

We are now ready to present a bound on the Gaussian parameters of finite-stellar rank states. Note that the following lemma immediately implies \cref{prop:boundG}.

\begin{lemma}\label{lem:last-lem}
Let $\ket{\psi}$ be a state of stellar rank $r$ over $n$ modes with energy at most $E$, i.e., $\bra\psi \sum_i N_i\ket\psi\leq E$. It is the case that
\begin{align}
\ket\psi = U S(\xi) D(\alpha) \ket{\mathrm{core}_r},
\end{align}
where $\alpha,\xi\in\mathbb C^n$, $U$ is a passive linear transform, $D$ is an $n$-mode displacement operator and $S$ is an $n$-mode squeezing operator. Then for all $j=1,\dots,n$,
\begin{align}
|\xi_j|\le O(\log(rE)),\qquad|\alpha_j|\le O(\sqrt r E).
\end{align}
\end{lemma}
\begin{proof}
Using Eq.~\eqref{eq:number-operator-evolution}, we have
\begin{align}
2\bra{\psi} N \ket\psi =  \sum_j e^{2\xi_j} \langle (X-x)^2\rangle_{\mathrm{core}_r} + e^{-2\xi_j} \langle (P-p)^2\rangle_{\mathrm{core}_r} \geq \cosh(2\xi_j) \cdot \Omega(\frac{1}{r})
\end{align}
which concludes
\begin{align}
\xi_j \leq O(\log(rE)).
\end{align}
Moreover, since we have $\norm{\Pi_r X\Pi_r}\leq \sqrt r$, we conclude that
\begin{align}
E \geq e^{-O(\log rE)}(\sqrt r - x)^2 = \Omega(\frac{1}{E r}) (\sqrt r - x)^2
\end{align}
which gives $x\leq O(E\sqrt r)$ for a maximum possible displacement $x$.
\end{proof}

\end{document}